\pdfoutput=1

\documentclass[letterpaper, 10 pt, conference]{ieeeconf}

\IEEEoverridecommandlockouts

\usepackage[T1]{fontenc}

\usepackage[utf8]{inputenc}

\usepackage{mathtools}
\usepackage{amssymb}

\usepackage[center]{caption}
\usepackage{parskip}
\makeatletter
\newcommand{\subalign}[1]{%
  \vcenter{%
    \Let@ \restore@math@cr \default@tag
    \baselineskip\fontdimen10 \scriptfont\tw@
    \advance\baselineskip\fontdimen12 \scriptfont\tw@
    \lineskip\thr@@\fontdimen8 \scriptfont\thr@@
    \lineskiplimit\lineskip
    \ialign{\hfil$\m@th\scriptstyle##$&$\m@th\scriptstyle{}##$\crcr
      #1\crcr
    }%
  }
}
\makeatother

\makeatletter
\def\munderbar#1{\underline{\sbox\tw@{$#1$}\dp\tw@\z@\box\tw@}}
\makeatother

\usepackage{bm}

\mathtoolsset{showonlyrefs,showmanualtags}

\newcommand*{\Scale}[2][4]{\scalebox{#1}{$#2$}}%
\usepackage{subfigure}

\usepackage{cite} 

\usepackage[]{units}

\usepackage{setspace}

\usepackage{wasysym}

\usepackage{tikz}
\usepackage{tikz-3dplot}
\usepackage{tkz-euclide}
\usetkzobj{all}
\usetikzlibrary{calc,decorations.markings,,intersections,shapes.misc}

\usepackage{graphicx}

\usepackage{multirow}

\usepackage{placeins}

\definecolor{dkgreen}{rgb}{0,0.6,0}

\usepackage{arydshln}

\newcommand{\sss}[1]{\scriptscriptstyle {#1}}

\newcommand{\tp}{^{\sss{T}}}

\newcommand{\sign}[2][{}]{\ensuremath{\text{sign}^{#1}\left(#2\right)}}

\newcommand{\twocol}[2]
 {\begin{bmatrix}
    {#1} \\
    {#2}
  \end{bmatrix}}

\newcommand{\threerow}[3]
 {\begin{bmatrix} {#1}&{#2}&{#3} \end{bmatrix}}

\newcommand{\Rn}[1][n]{\mathbb{R}^{\sss{#1}}}

\newcommand{\Sn}[1][n-1]{\mathcal{S}^{\sss{#1}}}
\newcommand{\Stwo}{\Sn[2]}

\newcommand{\norm}[1]{\left\|{#1}\right\|}
\newcommand{\sk}[1]{\mathcal{S}\left({#1}\right)}

\newcommand{\OP}[1]{\Pi\left({#1}\right)}

\DeclareMathOperator{\diag}{diag}

\newtheorem{thm}{Theorem}
\newtheorem{cor}[thm]{Corollary}

\newtheorem{prop}[thm]{Proposition}

\newtheorem{prob}{Problem}
\newtheorem{rem}[thm]{Remark}
\newtheorem{defn}{Definition}
\newtheorem{exa}{Example}

\newcommand{\amb}{\mathbf{a}}
\newcommand{\bmb}{\mathbf{b}}

\newcommand{\emb}{\mathbf{e}}
\newcommand{\fmb}{\mathbf{f}}

\newcommand{\nmb}{\mathbf{n}}

\newcommand{\umb}{\mathbf{u}}
\newcommand{\vmb}{\mathbf{v}}

\newcommand{\xmb}{\mathbf{x}}
\newcommand{\ymb}{\mathbf{y}}
\newcommand{\zmb}{\mathbf{z}}

\newcommand{\Smb}{\mathbf{S}} 
\newcommand{\Tmb}{\mathbf{T}}

\newcommand{\Idmat}{\text{I}} 							
\newcommand{\Rmat}{\ensuremath{\mathcal{R}}} 	

\newcommand{\zvec}{\mathbf{0}} 				
\newcommand{\onesvec}{\mathbf{1}} 			

\newcommand{\SO}[1][2]{\mathcal{SO}(#1)}

\newcommand{\embi}[1]{\emb_{\sss{#1}}}

\newcommand{\nmbt}[1]{{}_{\sss{#1}}\nmb}
\newcommand{\nmbh}[1]{{}_{\sss{\bar{#1}}}\nmb}

\newcommand{\Rmati}[1]{\Rmat_{\sss{#1}}}

\newcommand{\Omi}[1]{\bm{\omega}_{\sss{#1}}}
\newcommand{\OmiDot}[1]{\dot{\bm{\omega}}_{\sss{#1}}}

\newcommand{\Tmbi}[1][i]{ \Tmb_{\sss{#1}}}

\newcommand{\Ri}[1][i]{\Rmat_{\sss{#1}}}
\newcommand{\RiDot}[1][i]{\dot{\Rmat}_{\sss{#1}}}

\newcommand{\nmbi}[1][i]{\nmb_{\sss{#1}}}
\newcommand{\nmbiDot}[1][i]{\dot{\nmb}_{\sss{#1}}}

\newcommand{\nmbibody}[1][i]{\bar{\nmb}_{\sss{#1}}}

\newcommand{\Oi}[1][i]{\bm{\omega}_{\sss{#1}}}
\newcommand{\OiDot}[1][i]{\dot{\bm{\omega}}_{\sss{#1}}}

\def\paperextended{1}

\def\papercolor{1}

\numberwithin{equation}{section}

\newcommand{\review}[1]{{\color{black}#1}}
\newcommand{\ReviewAdded}[1]{{\color{black}#1}}

\let\oldbibliography\thebibliography
\renewcommand{\thebibliography}[1]{\oldbibliography{#1}
\setlength{\itemsep}{-3pt}} 

\newcommand{\qed}{ }

\begin{document}
	
	\abovedisplayskip=4pt plus 4pt minus 4pt
	\belowdisplayskip=4pt plus 4pt minus 4pt

	\title{Family of Controllers for Attitude Synchronization on the Sphere}
	
	\date{\today}
	\author{
		Pedro~O.~Pereira and Dimos~V.~Dimarogonas
		\thanks{ %
			The authors are with the ACCESS Linnaeus Center, School of Electrical Engineering, KTH Royal Institute of Technology, SE-100 44, Stockholm, Sweden. %
			\texttt{\{ppereira, dimos\}@kth.se}. This work was supported by the EU H2020 (AEROWORKS) project and the Swedish Research Council (VR).
		}
	}
	
	\maketitle	

	\begin{abstract}
		In this paper we study a family of controllers that guarantees attitude synchronization for a network of agents in the unit sphere domain, i.e., $\mathcal{S}^{\sss{2}}$. 
We propose distributed continuous controllers for elements whose dynamics are controllable, i.e., control with torque as command, and which can be implemented by each individual agent without the need of a common global orientation frame among the network, i.e., it requires only local information that can be measured by each individual agent from its own orientation frame. 
The controllers are constructed as functions of distance functions in $\mathcal{S}^{\sss{2}}$, and we provide conditions on those distance functions that guarantee that \emph{i)} a synchronized network of agents is locally asymptotically stable for an arbitrary connected network graph; \emph{ii)} a synchronized network is asymptotically achieved for almost all initial conditions in a tree network graph.
When performing synchronization along a principal axis, we propose controllers that do not require full torque, but rather torque orthogonal to that principal axis; 
while for synchronization along other axes, the proposed controllers require full torque. 
We also study the equilibria configurations that come with specific types of network graphs.
The proposed strategies can be used in attitude synchronization of swarms of under actuated rigid bodies, such as satellites.
	\end{abstract}

	\section{Introduction}
	Decentralized control in a multi-agent environment has been a topic of active research for the last decade, with applications in large scale robotic systems. 
Attitude synchronization in satellite formations is one of those applications~\cite{lawton2002synchronized}, where the control goal is to guarantee that a network of fully actuated rigid bodies acquires a common attitude. Coordination of underwater vehicles in ocean exploration missions can also be casted as an attitude synchronization problem~\cite{leonard2007collective}.

In the literature of attitude synchronization, different solutions for consensus in the special orthogonal group are found~\cite{lawton2002synchronized,sarlette2009autonomous,dimarogonas2009leader,bondhus2005leader,krogstad2006coordinated,cai2014leader,nair2007stable,thunberg2014distributed,song2015distributed}, which focus on \emph{complete} attitude synchronization. 
In this paper, we focus on \emph{incomplete} attitude synchronization, which has not received the same attention: in this scenario, each rigid body has a main direction and the global objective is to guarantee alignment of all rigid bodies' main directions; the space orthogonal to each main direction can be left free of actuation or controlled to accomplish some other goals. 
Complete attitude synchronization requires more measurements when compared to incomplete  attitude synchronization, and it might be the case that a rigid body is not fully actuated but rather only actuated in the space orthogonal to a specific direction, in which case incomplete  attitude synchronization is still feasible.
\review{\label{com:c1}\emph{Incomplete} attitude synchronization is also denoted synchronization on the sphere in~\cite{olfati2006swarms,moshtagh2007distributed,paley2009stabilization,li2014unified,dorfler2014synchronization,sarlette2008global}, where the focus has been on kinematic or point mass dynamic agents, i.e., dynamical agents without moment of inertia.
}

In~\cite{dimarogonas2009leader}, attitude control in a leader-follower network of rigid bodies has been studied, with the special orthogonal group being parametrized with Modified Rodrigues Parameters. 
The proposed solution guarantees attitude synchronization for connected graphs, but it requires all rigid bodies to be aware of a common and global orientation frame. 
In~\cite{bondhus2005leader,krogstad2006coordinated}, a controller for a single-leader single-follower network is proposed that guarantees global attitude synchronization at the cost of introducing a discontinuity in the control laws. 
In~\cite{cai2014leader}, attitude synchronization in a leader-follower network is accomplished by designing a non-linear distributed observer for the leader.
In~\cite{chung2009application,chung2013phase}, a combination of a tracking input and a synchronization input is used; 
the tracking input adds robustness if connectivity is lost and it is designed in the spirit of leader-following, where the leader is a virtual one and it encapsulates a desired trajectory;
however, this strategy requires all agents to be aware of a common and global reference frame.
In another line of work, in~\cite{nair2007stable,sarlette2009autonomous}, attitude synchronization is accomplished without the need of a common orientation frame among agents. Additionally, in~\cite{sarlette2009autonomous}, a controller for switching and directed network topologies is proposed, and local stability of consensus in connected graphs is guaranteed, provided that the control gain is sufficiently high. 
%
In~\cite{lawton2002synchronized}, attitude synchronization is accomplished with controllers based on behavior based approaches and for a bidirectional ring topology. The special orthogonal group is parametrized with quaternions, and the proposed strategy also requires a common attitude frame among agents. In~\cite{6111242}, a quaternion based controller is proposed that guarantees a synchronized network of rigid bodies is a global equilibrium configuration, provided that the network graph is acyclic. This comes at the cost of having to design discontinuous (hybrid) controllers.
A discrete time protocol for complete synchronization of kinematic agents is found in~\cite{tron2012intrinsic}.
The authors introduce the notion of \emph{reshaping function}, and a similar concept is presented in this manuscript. 
The protocol provides almost global convergence to a synchronized configuration, which relies on proving that all other equilibria configurations, apart from the equilibria configuration where agents are synchronized, are unstable.
In~\cite{thunberg2014distributed}, controllers for complete attitude synchronization and for switching topologies are proposed, but this is accomplished at the kinematic level, i.e., by controlling the agents' angular velocity (rather than their torque). This work is extended in~\cite{song2015distributed} by providing controllers at the torque level, and similarly to~\cite{lawton2002synchronized}, stability properties rely of high gain controllers. 

\review{\label{com:c2}In~\cite{olfati2006swarms,moshtagh2007distributed}, incomplete synchronization of kinematic agents on the sphere is studied, with a constant edge weight function for all edges.
In particular, in~\cite{moshtagh2007distributed}, incomplete synchronization is used for accomplishing a flocking behavior, where a group of agents moves in a common direction.
In~\cite{paley2009stabilization}, dynamic agents, which move at constant speed on a sphere, are controlled by a state feedback control law that steers their velocity vector so as to force the agents to attain a collective circular motion; 
since the agents are mass points, the effect of the moment of inertia is not studied.
In~\cite{li2014unified}, dynamic point mass agents, constrained to move on a sphere, are controlled to form patterns on the sphere, by constructing attractive and repelling forces;
in the absence of repelling forces, synchronization is achieved.
Also, the closed-loop dynamics of these agents are invariant to rotations, or symmetry preserving, as those in~\cite{olfati2006swarms,moshtagh2007distributed}, in the sense that two trajectories, whose initial condition -- composed of position and velocity -- differs only on a rotation, are the same at each time instant apart from the previous rotation.
In our framework this property does not hold, since our dynamic agents have a moment of inertia, unlike the agents in~\cite{olfati2006swarms,moshtagh2007distributed,li2014unified}, which is another novelty of the paper in hand.
}

We propose a distributed control strategy for synchronization of elements in the unit sphere domain. 
%
%
The controllers for accomplishing synchronization are constructed as functions of distance functions (or \emph{reshaping functions} as denoted in~\cite{tron2012intrinsic}), and, in order to exploit results from graph theory, we impose a condition on those distance functions that will restrict them to be invariant to rotations of their arguments.
As a consequence, the proposed controllers can be implemented by each agent without the need of a common orientation frame. 
We restrict the proposed controllers to be continuous, which means that a synchronized network of agents cannot be a global equilibrium configuration, since $\mathcal{S}^{\sss{2}}$ is a non-contractible set~\cite{liberzon2003switching}. 
\ReviewAdded{Our main contributions lie in proposing for the first time a controller that does not require full torque when performing synchronization along a principal axis, but rather torque orthogonal to that axis;
in finding conditions on the distance functions that guarantee that a synchronized network is locally asymptotically stable for arbitrary connected network graphs, and that guarantee that a synchronized network is achieved for almost all initial conditions in a tree graph;
in providing explicit domains of attraction for the network to converge to a synchronized network;
and in characterizing the equilibria configurations for some general, yet specific, types of network graphs. }
A preliminary version of this work was submitted to the 2015 IEEE Conference on Decision and Control~\cite{PreliminaryS2}. With respect to this preliminary version, this paper presents significantly more details on the derivation of the main theorems and provides additional results. In particular, the concept of cone has been modified, with a clearer intuitive interpretation; the proof for the proposition that supports the result on local stability of the synchronized network has been simplified; further details on the condition imposed on the distance functions are provided; additional examples on possible distance functions, and their properties, are presented; and supplementary simulations are provided which further illustrate the theoretical results.

The remainder of this paper is structured as follows. In Section~\ref{sec:ProblemStatement}, the problem statement is described; in Section~\ref{sec:ProposedSolution}, the proposed solution is presented; in Sections~\ref{eq:TreeGraphs} and~\ref{eq:NonTreeGraphs}, convergence to a synchronized network is discussed for tree and arbitrary graphs, respectively; and, in Section~\ref{sec:Simulations}, simulations are presented that illustrate the theoretical results.
	
		\section{Notation}
%
$\zvec_{\sss{n}} \in \Rn[n]$ and $\onesvec_{\sss{n}} \in \Rn[n]$ denote the zero column vector and the column vector with all components equal to 1, respectively; 
$\bar{\onesvec}_{\sss{n}} \in \Rn[n]$ denotes a column vector in $\{(x_{\sss{1}},\cdots,x_{\sss{n}}) \in \Rn[n] : x_{\sss{i}} = \pm 1, \forall i = \{1,\cdots,n\}\} $;
when the subscript ${n}$ is omitted, the dimension $n$ is assumed to be of appropriate size. $\Idmat_{\sss{n}} \in \Rn[n \times n]$ stands for the identity matrix, and we omit its subscript when $n=3$.  
%
%
The matrix $\sk{\cdot} \in  \Rn[3\times 3]$ is a skew-symmetric matrix and it satisfies $\sk{\amb} \, \bmb = \amb \times \bmb$, for any $\amb,\bmb\in \Rn[3]$.
The map $\Pi:\{\xmb \in \Rn[3]: \xmb\tp \xmb =1 \} \mapsto\Rn[3\times 3]$, defined as $\OP{\xmb} = \Idmat - \xmb \xmb\tp$, yields a matrix that represents the orthogonal projection operator onto the subspace perpendicular to $\xmb$.	
We denote the Kronecker product between $A \in \Rn[m\times n] $ and $B \in \Rn[s\times t] $ by $A \otimes B \in \Rn[m \, s\times n \, t]$.
Given $A_{\sss{1}}, \cdots, A_{\sss{n}} \in \Rn[m \times m]$, for some $n,m \in \mathbb{N}$, we denote $A = A_{\sss{1}} \oplus \cdots \oplus A_{\sss{n}} \in \Rn[nm \times nm]$ (direct sum of matrices) as the block diagonal matrix with block diagonal entries $A_{\sss{1}}$ to $A_{\sss{n}}$.
%
%
Given $\amb,\bmb \in \Rn[n]$, $\amb = \pm \bmb \Leftrightarrow \amb = \bmb \vee \amb = -\bmb$; 
additionally, we say $\amb \ne \zvec$ and $\bmb \ne \zvec$ have the same direction if there exists $\lambda \in \Rn[{}]$ such that $\bmb = \lambda \amb$.
%
%
We say a function $f:\Omega_{\sss{1}} \mapsto \Omega_{\sss{2}}$ is of class $\mathcal{C}^{\sss{n}}$, or equivalently $f \in \mathcal{C}^{\sss{n}}(\Omega_{\sss{1}},\Omega_{\sss{2}})$,  if its first $n+1$ derivatives (i.e., $f^{\sss{(0)}}$, $f^{\sss{(1)}}$, $\cdots$, $f^{\sss{(n)}}$) exist and are continuous on $\Omega_{\sss{1}}$.
Finally, given a set $\mathcal{H}$, we use the notation $|\mathcal{H}|$ for the cardinality of $\mathcal{H}$.

		\section{Problem Statement}
		\label{sec:ProblemStatement}
		We consider a group of $N$ agents, indexed by the set $\mathcal{N} = \{ 1, \cdots, N\}$, operating in the unit sphere domain, i.e., in $\Stwo = \{\xmb \in \Rn[3]: \xmb\tp \xmb =1 \}$.
The agents' network is modeled as an undirected static graph, $\mathcal{G} = \{\mathcal{N}, \mathcal{E}\}$, with $\mathcal{N}$ as the vertices' set indexed by the team members, and $\mathcal{E}$ as the edges' set. 
For every pair of agents $(i,j) \in \mathcal{N} \times (\mathcal{N}\backslash\{ i \})$, that are \emph{aware of} and can measure each other's relative attitude, we  say that agent $j$ is a neighbor of agent $i$, and vice-versa; also, we denote $\mathcal{N}_i \subset \mathcal{N}$ as the neighbor set of agent $i$.

Each agent $i$ has its own orientation frame (w.r.t. an unknown inertial orientation frame), represented  by $\Rmati{i} \in \SO[3]$. 
Let the unit vector $\nmbi[i] \in \Stwo $ be a direction along agent's $i$ orientation, i.e., $\nmbi[i] = \Rmati{i} \bar{\nmb}_{\sss{i}}$, where $\bar{\nmb}_{\sss{i}} \in \Stwo$ is a constant unit vector, specified in the agent's~$i$ body orientation frame, and known by agent~$i$ and its neighbors.
In this paper, the goal of attitude synchronization is not that all agents share the same \emph{complete} orientation, i.e., that $\Rmati{1} = \cdots = \Rmati{N}$, but rather that all agents share the same orientation along a specific direction, i.e., that $\nmbi[1] = \cdots = \nmbi[N] \Leftrightarrow \Rmati{1}\bar{\nmb}_{\sss{1}} = \cdots = \Rmati{N} \bar{\nmb}_{\sss{N}}$. 
For example, in a group of $N$ satellites that must align their principal axis associated to the smallest moment of inertia, it follows that, for each $i \in \mathcal{N}$, $\bar{\nmb}_{\sss{i}} \in \mathcal{S}^{\sss{2}} : \exists \lambda_{\sss{i}} > 0: J_{\sss{i}} \bar{\nmb}_{\sss{i}} = \lambda_{\sss{i}} \bar{\nmb}_{\sss{i}}$ with $J_{\sss{i}}$ as the satellite's moment of inertia and with $\lambda_{\sss{i}}$ as the smallest eigenvalue of $J_{\sss{i}}$;
and that the desired synchronized network of satellites satisfies $\Rmati{1}\bar{\nmb}_{\sss{1}} = \cdots = \Rmati{N}\bar{\nmb}_{\sss{N}}$. 
Figures~\ref{fig:2RigidBodies} and~\ref{fig:2RigidBodiesExtra} illustrate  the concept of incomplete synchronization for two agents. 
Notice that agent $i$ is not aware of $\nmbi[i]$, since this is specified w.r.t. an unknown inertial orientation frame; instead, agent $i$ is aware of its own direction $\bar{\nmb}_{\sss{i}}$ -- fixed in its own orientation frame -- and the projection of its neighbors directions onto its own orientation frame. 
\begin{figure}
	\centering
	\subcapcentertrue	
	\subfigure[Two non-synchronized rigid bodies.]{
		\includegraphics[clip=true,trim=0cm 0cm 0cm 3cm,width=0.3\textwidth]
		{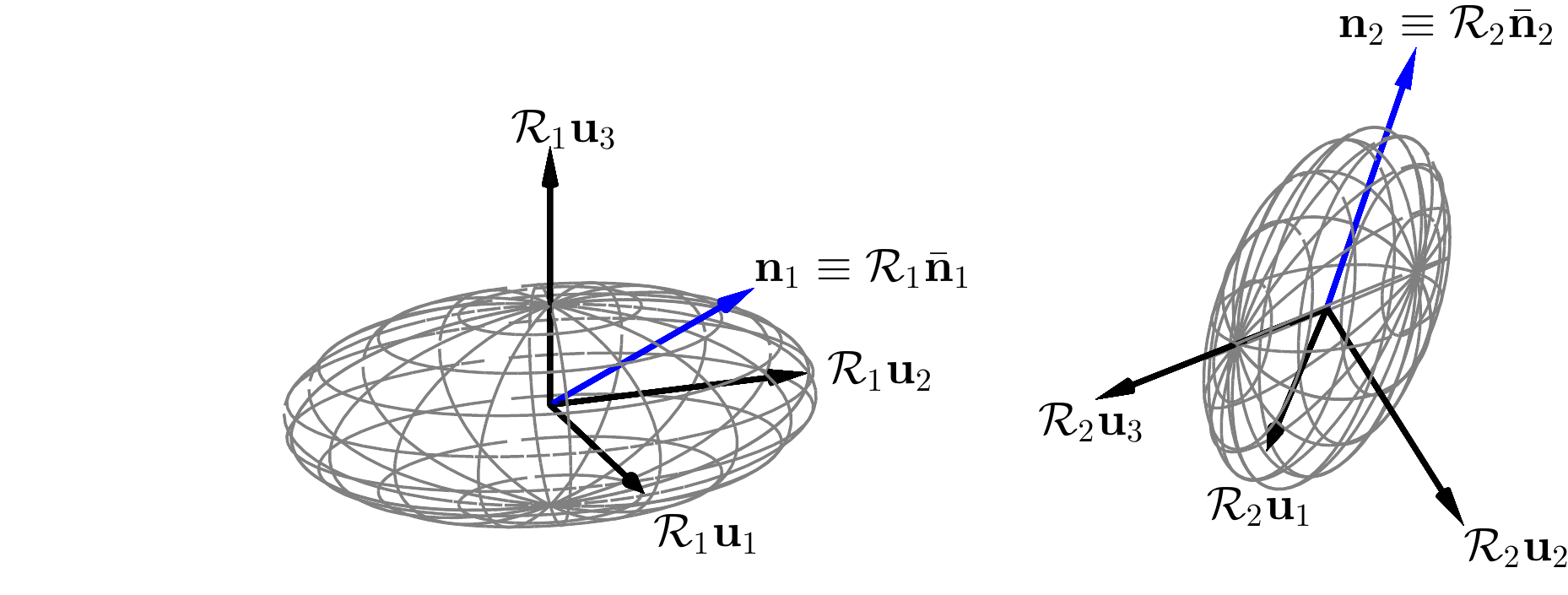}
		\label{fig:Ellipsoid1}	
	}
	\subfigure[Two synchronized rigid bodies.]{
		\includegraphics[clip=true,trim=0cm 0cm 0cm 3cm,width=0.3\textwidth]
		{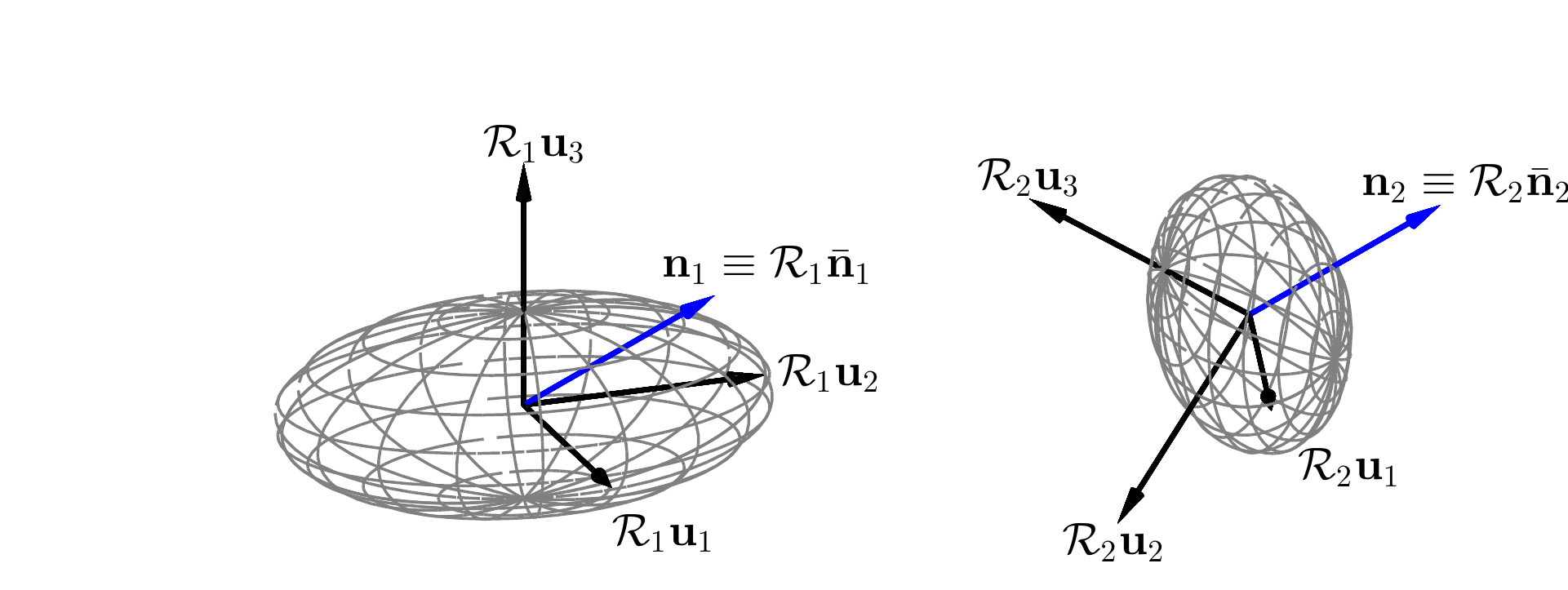}	
		\label{fig:Ellipsoid2}
	}	
	\caption{In incomplete synchronization, $n$ rigid bodies, indexed by $i = \{1,\cdots,n\}$, synchronize the unit vectors $\nmbi[i] = \Rmati{i}\bar{\nmb}_{\sss{i}}$, where $\bar{\nmb}_{\sss{i}}$ is fixed in rigid body $i$. In Figs.~\ref{fig:Ellipsoid1}-\ref{fig:Ellipsoid2}, $\bar{\nmb}_{\sss{1}} = - \bar{\nmb}_{\sss{2}} = \frac{1}{\sqrt{3}} [1 \, \, 1 \, \, 1]\tp$ ($\umb_{\sss{1}}$,$\umb_{\sss{2}}$ and $\umb_{\sss{3}}$ stand for the canonical basis vectors of $\Rn[3]$).}
	\label{fig:2RigidBodies}
\end{figure}
\begin{figure*}
	\centering
	\subcapcentertrue	
	\subfigure[Two non-synchronized rigid bodies.]{
		\includegraphics[clip=true,trim=0cm 0cm 0cm 3cm,width=0.3\textwidth]
		{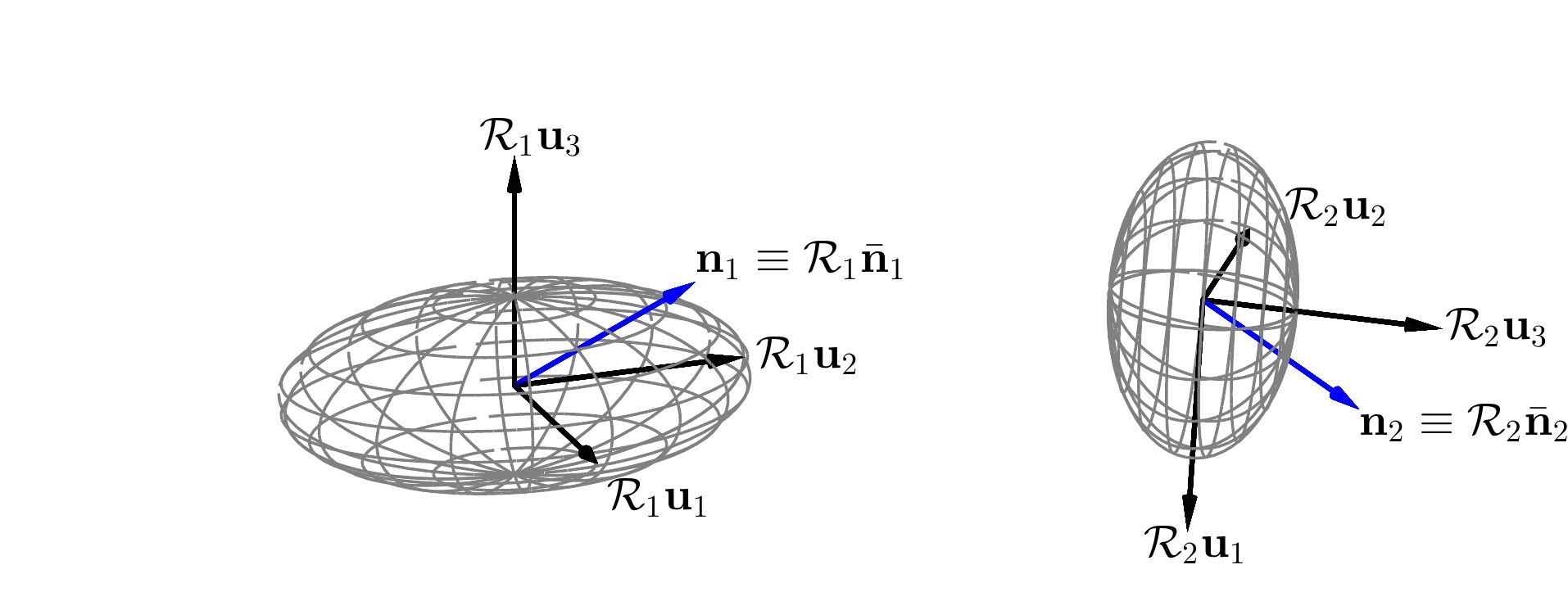}
		\label{fig:Ellipsoid1Extra}	
	}
	\subfigure[Two synchronized rigid bodies (where $\Rmati{1} = \Rmati{2}$).]{
		\includegraphics[clip=true,trim=0cm 0cm 0cm 3cm,width=0.3\textwidth]
		{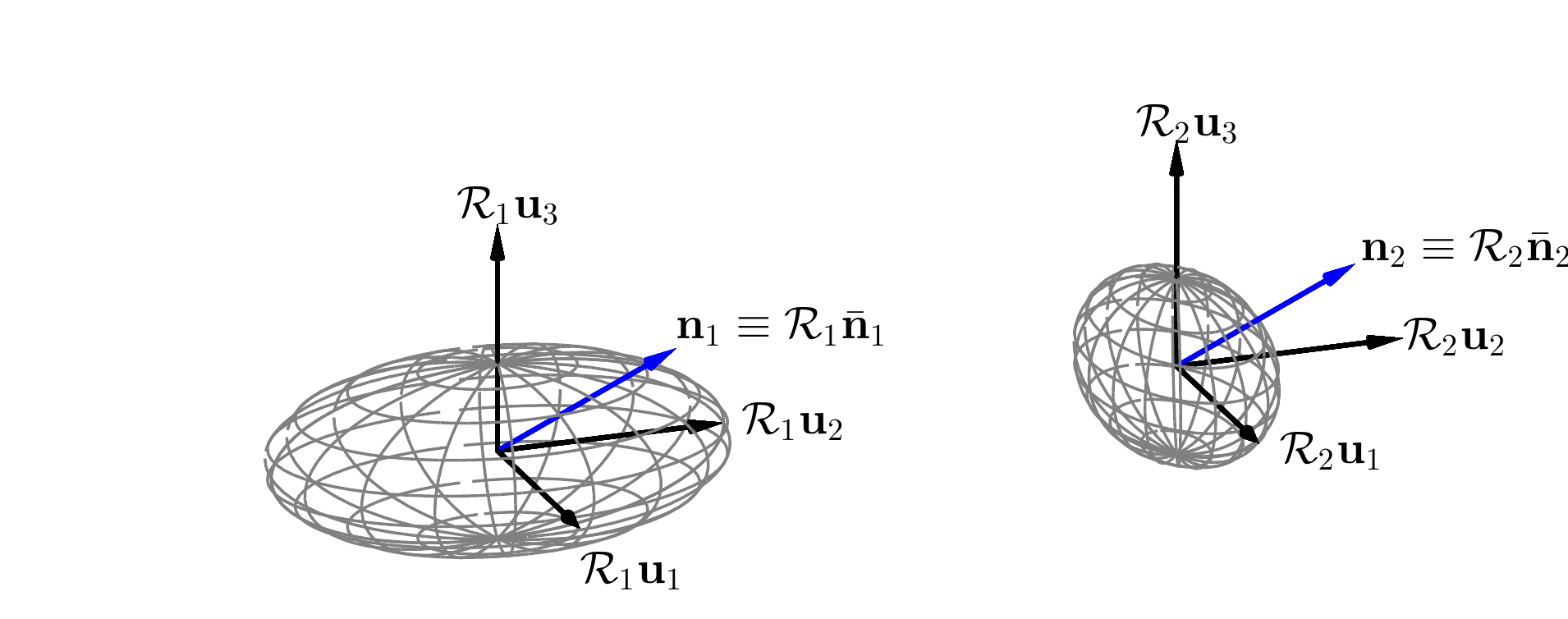}	
		\label{fig:Ellipsoid2Extra}
	}	
	\subfigure[Two synchronized rigid bodies (where $\Rmati{1} \ne \Rmati{2}$).]{
		\includegraphics[clip=true,trim=0cm 0cm 0cm 3cm,width=0.3\textwidth]
		{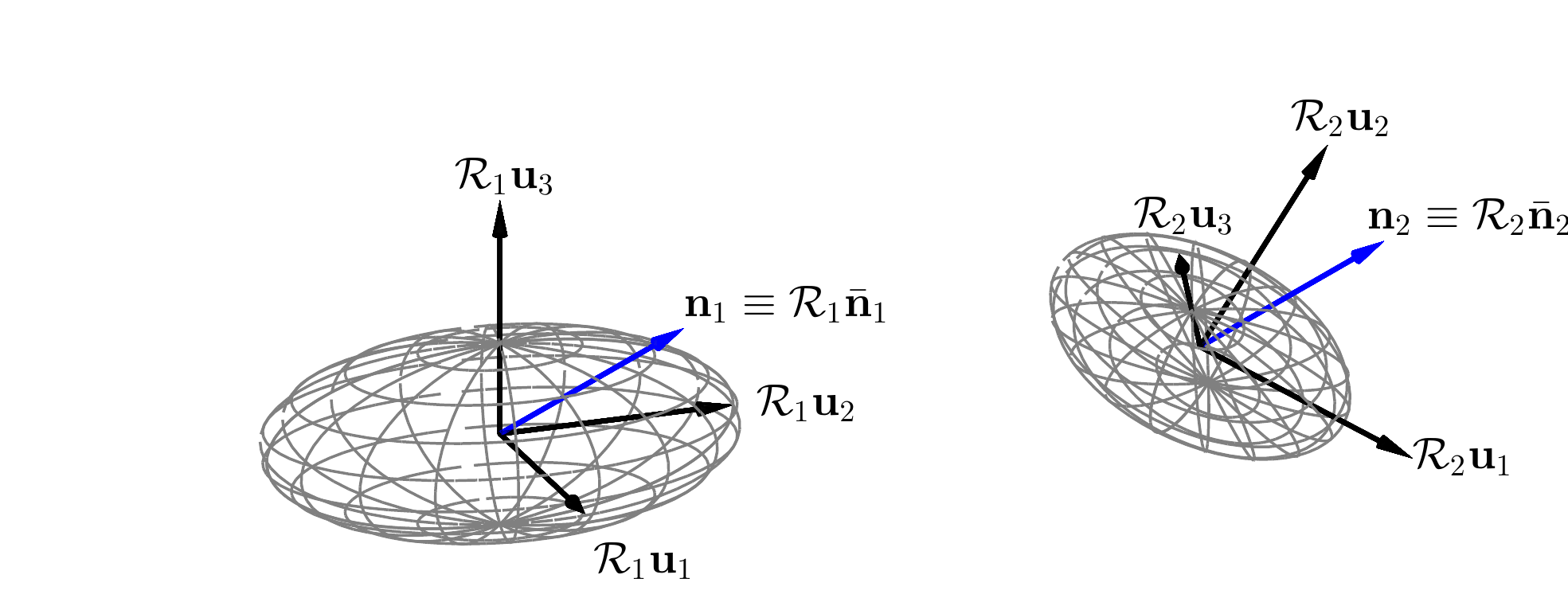}	
		\label{fig:Ellipsoid3Extra}
	}		
	\caption{In incomplete synchronization, $n$ rigid bodies, indexed by $i = \{1,\cdots,n\}$, synchronize the unit vectors $\nmbi[i] = \Rmati{i}\bar{\nmb}_{\sss{i}}$, where $\bar{\nmb}_{\sss{i}}$ is fixed in rigid body $i$. In Figs.~\ref{fig:Ellipsoid1Extra}--\ref{fig:Ellipsoid3Extra}, $\bar{\nmb}_{\sss{1}} = \bar{\nmb}_{\sss{2}} = \frac{1}{\sqrt{3}} [1 \, \, 1 \, \, 1]\tp$ ($\umb_{\sss{1}}$,$\umb_{\sss{2}}$ and $\umb_{\sss{3}}$ stand for the canonical basis vectors of $\Rn[3]$).}
	\label{fig:2RigidBodiesExtra}
\end{figure*}

Consider then any agent~$i \in \mathcal{N}$, with rotation matrix $\Ri : \Rn[]_{\sss{\ge 0}} \mapsto \SO[3]$, unit vector $\nmbi: \Rn[]_{\sss{\ge 0}} \mapsto \Stwo$ where $\nmbi(\cdot) = \Ri(\cdot) \nmbibody$, body-framed angular velocity $\Oi : \Rn[]_{\sss{\ge 0}} \mapsto \Rn[3]$, moment of inertia $J_{\sss{i}} \in \Rn[3 \times 3]$ ($J_{\sss{i}}>0$), and body frame torque $\Tmb_{\sss{i}} : \Rn[]_{\sss{\ge 0}} \mapsto \Rn[3]$.
The rotation matrix $\Ri: \Rn[]_{\sss{\ge 0}} \mapsto \SO[3]$ evolves according to
\begin{align}
	\RiDot(t)  = \fmb_{\sss{\Rmat}}(\Ri(t),\Oi(t)),
	\label{eq:RotationMatrixKinematicsaux}
\end{align}
where $\fmb_{\sss{\Rmat}} : \SO[3] \times \Rn[3] \mapsto \Rn[3 \times 3]$ is defined as
\begin{align}
	\fmb_{\sss{\Rmat}}(\Rmat,\bm{\omega})
	=
	\Rmat \sk{\bm{\omega}},
	\label{eq:RotationMatrixKinematics}
\end{align}
while each unit vector $\nmbi: \Rn[]_{\sss{\ge 0}} \mapsto \Stwo$ evolves according to $\nmbiDot(t)  = \fmb_{\sss{\nmbi}}(t,\nmbi(t),\Oi(t))$, where $\fmb_{\sss{\nmbi}} : \Rn[{}]_{\sss{\ge 0}} \times \mathcal{S}^{\sss{2}} \times \Rn[3] \mapsto \Rn[3]$ is defined as
\begin{align}
	\fmb_{\sss{\nmbi}}(t,\nmb,\bm{\omega})
	=
	\sk{\Ri(t) \bm{\omega}} \nmb.
	\label{eq:UnitVectorField}
\end{align}
The previous results follows from the fact that $\nmbi(\cdot) = \Ri(\cdot) \nmbibody$ for some constant $\nmbibody \in \Stwo$, and therefore $\RiDot(t) \nmbibody  = \sk{\Ri(t) \Oi(t)} \Ri(t) \nmbibody \Rightarrow \nmbiDot(t)  = \sk{\Ri(t) \Oi(t)} \nmbi(t)$.
Finally, the body-framed angular velocity $\Oi: \Rn[]_{\sss{\ge 0}} \mapsto \Rn[3]$ evolves according to the dynamics
\begin{align}
	& \frac{d}{dt}
	\left(
		\Ri(t) J_{\sss{i}} \Oi(t) 
	\right)
	\label{eq:AngularVelocityDynamics}	
	=
	\Ri(t) \Tmb_{\sss{i}}(t)
	\Leftrightarrow
	\\
	\Leftrightarrow
	& \OiDot(t)
	=
	J_{\sss{i}}^{\sss{-1}}
	\left(
		-\sk{\Oi(t)} J_{\sss{i}} \Oi(t) + \Tmb_{\sss{i}}(t)
	\right),
	\label{eq:AngularVelocityDynamics2}	
\end{align}
and therefore $\OiDot(t) = \fmb_{\sss{\Oi}} ( \Oi(t), \Tmb_{\sss{i}}(t) )$, where $\fmb_{\sss{\Oi}} : \Rn[3] \times \Rn[3] \mapsto \Rn[3]$ is defined as
\begin{align}
	\fmb_{\sss{\Oi}}
	\left(
		\bm{\omega}, \Tmb
	\right)	
	=
	J_{\sss{i}}^{\sss{-1}}
	\left(
		-\sk{\bm{\omega}} J_{\sss{i}} \bm{\omega} + \Tmb
	\right).
	\label{eq:OmegaVectorField}
\end{align}
\begin{defn}
	\normalfont
	Two unit vectors $(\nmbi[1],\nmbi[2]) \in (\Stwo)^{\sss{2}}$ are diametrically opposed if $\nmbi[1]\tp\nmbi[2] = -1$, and synchronized if $\nmbi[1]\tp\nmbi[2] = 1$. 
	A group of unit vectors $(\nmbi[1],\cdots,\nmbi[N]) \in (\Stwo)^{\sss{N}}$ is synchronized if $\nmbi[i]\tp\nmbi[j] = 1$ for all $i,j \in \{1,\cdots,N\}$. 
\end{defn}

\begin{prob}
	\label{prob:ProblemUnitVectorDynamics}
	\normalfont
	Given a group of rotation matrices $(\Rmati{1},\cdots,\Rmati{N}) : \Rn[]_{\sss{\ge 0}} \mapsto \SO[3]^{\sss{N}}$, with angular velocities $(\Omi{1},\cdots,\Omi{N}) : \Rn[]_{\sss{\ge 0}} \mapsto\Rn[3]$ and moments of inertia $J_{\sss{1}},\cdots,J_{\sss{N}}$ satisfying~\eqref{eq:RotationMatrixKinematicsaux} and~\eqref{eq:AngularVelocityDynamics2}, design distributed control laws for the torques $\{\Tmb_{\sss{i}}: \Rn[]_{\sss{\ge 0}} \mapsto \Rn[3] \}_{\sss{i\in\mathcal{N}}}  $, in the absence of a common inertial orientation frame, that guarantee that the group of unit vectors $(\nmbi[1],\cdots,\nmbi[N]) : \Rn[]_{\sss{\ge 0}} \mapsto (\Stwo)^{\sss{N}}$ is asymptotically synchronized.
\end{prob}

For the purposes of analysis, we consider the state $\xmb := (\nmb,\bm{\omega}):=((\nmbi[1],\cdots,\nmbi[N]),(\Omi{1},\cdots,\Omi{N})) : \Rn[]_{\sss{\ge 0}} \mapsto (\mathcal{S}^{\sss{2}})^{\sss{N}}  \times (\Rn[3])^{\sss{N}}$, and the control input $\Tmb := (\Tmbi[1] , \cdots , \Tmbi[N]) : \Rn[]_{\sss{\ge 0}} \mapsto (\Rn[3 ])^{\sss{N}}$;
where $\xmb(\cdot)$ evolves according to $\dot{\xmb}(t)  =  \fmb_{\sss{\xmb}}(t,\xmb(t),\Tmb(t))$ where
\begin{align}
	\hspace{-0.5cm}
	\fmb_{\sss{\xmb}}(t,\xmb,\Tmb)
	=
	(\fmb_{\sss{\nmb}}(t,\nmb,\bm{\omega}),\fmb_{\sss{\bm{\omega}}}(\bm{\omega},\Tmb)) \in \Rn[3N] \times \Rn[3N],
	\label{eq:StateVectorField}
\end{align}
with $\fmb_{\sss{\nmb}}(t,\nmb,\bm{\omega}) = (\fmb_{\sss{\nmbi[1]}}(t,\nmbi[1],\Oi[1]), \cdots, \fmb_{\sss{\nmbi[N]}}(t,\nmbi[n],\Oi[N])) \in (\Rn[3])^{\sss{N}}$ and $\fmb_{\sss{\bm{\omega}}}(\bm{\omega},\Tmb) = (\fmb_{\sss{\Oi[1]}}(\Oi[1],\Tmbi[1]), \cdots, \fmb_{\sss{\Oi[N]}}(\Oi[N],\Tmbi[N])) \in (\Rn[3])^{\sss{N}}$.
\begin{rem}
	Given our choice of state, the corresponding vector field in~\eqref{eq:StateVectorField} depends explicitly on time.
	The explicit time dependency encapsulates the effect of the space orthogonal to $\nmbi[i](\cdot)$ on the kinematics of  $\nmbiDot[i](\cdot)$.
	Therefore, in the analysis we would have to consider the effect of initializing the system at different time instants $t_{\sss{0}} \in \Rn[]$, corresponding to different initializations of the space orthogonal to $\nmbi[i](t_{\sss{0}})$.
	In general, that initialization does have an effect on a trajectory $\xmb(\cdot)$ of $\dot{\xmb}(t) = \fmb_{\sss{\xmb}}(t,\xmb(t),\Tmb(t))$, but, as shall be verified later, it does not affect the attainment of the goal defined in Problem~\ref{prob:ProblemUnitVectorDynamics}.
	Therefore, we assume that the initial time instant is always $0$, i.e., $t_{\sss{0}} = 0$.
\end{rem}

		\section{Proposed Solution}
		\label{sec:ProposedSolution}
		
		\subsection{Preliminaries}
		\label{subsec:Preliminaries}
		We first present some definitions and results from graph theory that are used in later sections~\cite{godsil2001algebraic}.
A graph $\mathcal{G} = \{\mathcal{N}, \mathcal{E}\}$ is said to be connected if there exists a path between any two vertices in $\mathcal{N}$. $\mathcal{G}$ is a tree if it is connected and it contains no cycles. An orientation on the graph $\mathcal{G}$ is the assignment of a direction to each edge $(i,j) \in \mathcal{E}$, where each edge vertex is either the tail or the head of the edge. 
For brevity, we denote $N = |\mathcal{N}|$, $M = |\mathcal{E}|$ and $\mathcal{M} = \{1, \cdots, M\}$.
Additionally, and for notational convenience in the analysis that follows, consider the sets $\mathcal{E} = \{(i,j) \in \mathcal{N} \times \mathcal{N}: j \in \mathcal{N}_i  \}$, i.e., the set of edges of the graph $\mathcal{G}$; and $\bar{\mathcal{E}} = \{(i,j) \in \mathcal{E}: j > i  \}$. For undirected network graphs, we can construct an injective function $\bar{\kappa} : \bar{\mathcal{E}} \mapsto \mathcal{M}$ from which it is possible to construct a second, now surjective, function  $\kappa : \mathcal{E} \mapsto \mathcal{M}$, which satisfies $\kappa(i,j) = \bar{\kappa}(i,j)$ when $j > i$ and $\kappa(i,j) = \bar{\kappa}(j,i)$ when $j < i$. As such, by construction, for every $(i,j) \in \mathcal{E}$, $\kappa(i,j) = \kappa(j,i)$, since we consider undirected graphs. The function $\kappa(\cdot,\cdot)$ thus assigns an edge index to every unordered pair of neighbors.
The incidence matrix $B \in \Rn[N \times M]$ of $\mathcal{G}$ is such that, for every $k\in \mathcal{M}$ and for $(i,j) = \bar{\kappa}^{\sss{-1}}(k)$, $B_{\sss{ik}} = 1$, $B_{\sss{jk}} = - 1$ and $B_{\sss{lk}} = 0$ for all $l \in \mathcal{N}\backslash\{i,j\}$. 
Finally, for each edge $k\in \mathcal{M}$  and $(i,j) = \bar{\kappa}^{\sss{-1}}(k)$, we denote $\nmbt{k} := \nmbi[i] $ and $\nmbh{k} := \nmbi[j] $, i.e., we identify an agent by its node index but also by its edges' indexes  ($\nmbt{k}$ if $\nmbi[i]$ is the tail of edge $k$, and $\nmbh{k}$ if $\nmbi[i]$ is the head of edge $k$).

\begin{prop}\label{prop:Bdefinitepositive}
	\normalfont
	If $\mathcal{G}$ is a tree, then $B\tp B $ and $(B \otimes \Idmat)\tp (B \otimes \Idmat)$ are positive definite~\cite{dimarogonas2009further}.
\end{prop}

\begin{prop}\label{prop:BNullSpace}
	\normalfont
	If $\mathcal{G}$ is connected but not a tree, then the null space of the incidence matrix, i.e., $\mathcal{N}(B)$, is non-empty, and it corresponds to the cycle space of $\mathcal{G}$ (Lemma~3.2 in~\cite{GraphEmbeddings}).
\end{prop}

Denote by $C \subseteq \{1,\cdots,M\}$ the set of indices corresponding the edges that form a cycle.
Consider a network with $n\in\mathbb{N}$ cycles, $\{C_{\sss{i}}\}_{\sss{i=\{1,\cdots,n\}}}$.
A cycle $C_{\sss{i}}$ is said to be independent if $C_{\sss{i}} \cap C_{\sss{j}} = \emptyset$ for all $j \in \{1,\cdots,n\} \backslash\{i\}$.
In Fig.~\ref{fig:IndependentCycles}, a graph with two independent cycles is presented.
Additionally, we say two cycles $C_{\sss{1}}$ and $C_{\sss{2}}$ share only one edge when $|C_{\sss{1}} \cap C_{\sss{2}}| =1$ and $C_{\sss{1}} \cup C_{\sss{2}}$ contains edges from only the following three cycles (in  $\{C_{\sss{i}}\}_{\sss{i=\{1,\cdots,n\}}}$): $C_{\sss{1}}$, $C_{\sss{2}}$ and $C_{\sss{3}} = C_{\sss{1}} \cup C_{\sss{2}} \backslash \{ C_{\sss{1}} \cap C_{\sss{2}} \}$, with $|C_{\sss{3}}| = |C_{\sss{1}}| + |C_{\sss{2}}| - 2$. 
Figures~\ref{fig:SharedCycles} and~\ref{subfig:Agents6Network} present graphs with two cycles that share only one edge.

\begin{prop}\label{prop:BIndependentCycles}
	\normalfont
	Consider a graph $\mathcal{G}$ with $m$ independent cycles, $\{C_{\sss{i}}\}_{\sss{i=\{1,\cdots,m\}}}$. Then the null space of $B$ is given by $\mathcal{N}(B) = \{ \emb = (e_{\sss{1}},\cdots,e_{\sss{M}}) \in \Rn[M]: e_{\sss{k}} = \pm e_{\sss{l}}, \forall k,l \in C_i, i = \{ 1, \cdots, m\}   \}$; and the null space of $B \otimes \Idmat_{\sss{n}}$ is given by $\mathcal{N}(B \otimes \Idmat_{\sss{n}}) = \{ \emb = (\emb_{\sss{1}},\cdots,\emb_{\sss{M}}) \in (\Rn[n])^{\sss{M}}: \emb_k = \pm \emb_l, \forall k,l \in C_i, i = \{ 1, \cdots, m\}   \}$.
\end{prop}

Notice that, for an incidence matrix $B \in \Rn[N \times M]$, there are $M$ edges and these \emph{belong} to $\Rn[]$. On the other hand, for the incidence matrix $B \otimes \Idmat_{\sss{n}}$ (with $B \in \Rn[N \times M]$), there are $M$ edges, but, since the agents operate in an $n$-dimensional space, those edges \emph{belong} to $\Rn[n]$.
With that in mind, under the conditions of Proposition~\ref{prop:BIndependentCycles}, $\mathcal{N}(B) = \{ \emb = (e_{\sss{1}},\cdots,e_{\sss{M}}) \in \Rn[M]: e_{\sss{k}} = \pm e_{\sss{l}}, \forall k,l \in C_i, i = \{ 1, \cdots, m\}   \}$ means that $\mathcal{N}(B)$ is the space where all edges of an independent cycle have the same absolute value.
On the other hand, $\mathcal{N}(B \otimes \Idmat_{\sss{n}}) = \{ \emb = (\emb_{\sss{1}},\cdots,\emb_{\sss{M}}) \in (\Rn[n])^{\sss{M}}: \emb_k = \pm \emb_l, \forall k,l \in C_i, i = \{ 1, \cdots, m\}   \}$ means that $\mathcal{N}(B \otimes \Idmat_{\sss{n}})$ is the space where all edges of an independent cycle have the same direction and norm (or are all zero).
\begin{prop}\label{prop:BAlmostIndependentCycles}
	\normalfont
	Consider a graph $\mathcal{G}$ with $n_{\sss{1}}$ independent cycles, $\{C_{\sss{i}}\}_{\sss{i=\{1,\cdots,n_{\sss{1}}\}}}$, and $n_{\sss{2}}$ pairs of cycles that share only one edge, $
	\{
		(C_{\sss{i}}^{\sss{1}},
		C_{\sss{i}}^{\sss{2}})
	\}_{\sss{i=\{1,\cdots,n_{\sss{2}}\}}}
	$. 
	Then the null space of $B$ is given by 
	$
		\mathcal{N}(B) = 
		\{ 
			\emb = (e_{\sss{1}},\cdots,e_{\sss{M}}) \in \Rn[M]: 
			e_{\sss{k}} = \pm e_{\sss{l}}, \forall k,l \in C_{\sss{i}},
			i = \{1,\cdots,n_{\sss{1}} \} 
		\} 
		\cup 
		\{ 
			\emb = (e_{\sss{1}},\cdots,e_{\sss{M}}) \in \Rn[M]:
			e_{\sss{k}} =  \pm e_{\sss{l}}, 
			\forall k,l \in 
			C_{\sss{i}}^{\sss{1}}
			\backslash 
			\{C_{\sss{i}}^{\sss{1}} \cap C_{\sss{i}}^{\sss{2}}\},
			e_{\sss{p}} = \pm e_{\sss{q}}, 
			\forall p,q \in 
			C_{\sss{i}}^{\sss{2}}
			\backslash 
			\{C_{\sss{i}}^{\sss{1}} \cap C_{\sss{i}}^{\sss{2}}\},
			i = \{1,\cdots,n_{\sss{2}} \} 
		\}$;
		and the null space of $B \otimes \Idmat_{\sss{n}}$ is given by 
		$
		\mathcal{N}(B \otimes \Idmat_{\sss{n}}) = 
		\{ 
			\emb = (\emb_{\sss{1}},\cdots,\emb_{\sss{M}}) \in (\Rn[n])^{\sss{M}}: 
			\emb_{\sss{k}} = \pm \emb_{\sss{l}}, \forall k,l \in C_{\sss{i}},
			i = \{1,\cdots,n_{\sss{1}} \} 
		\} 
		\cup 
		\{ 
			\emb = (\emb_{\sss{1}},\cdots,\emb_{\sss{M}}) \in (\Rn[n])^{\sss{M}}:
			\emb_{\sss{k}} = \pm \emb_{\sss{l}}, 
			\forall k,l \in 
			C_{\sss{i}}^{\sss{1}}
			\backslash 
			\{C_{\sss{i}}^{\sss{1}} \cap C_{\sss{i}}^{\sss{2}}\},
			\emb_{\sss{p}} = \pm \emb_{\sss{q}}, 
			\forall p,q \in 
			C_{\sss{i}}^{\sss{2}}
			\backslash 
			\{C_{\sss{i}}^{\sss{1}} \cap C_{\sss{i}}^{\sss{2}}\},
			i = \{1,\cdots,n_{\sss{2}} \} 
		\}
		$.
\end{prop}
The results of Proposition~\ref{prop:BAlmostIndependentCycles} can be interpreted as follows: $\mathcal{N}(B)$ is the space where, for each cycle, all its edges ($\in \Rn[]$), except the one that is shared, have the same absolute value; while $\mathcal{N}(B \otimes  \Idmat_{n})$ is the space where, for each cycle, all its edges ($\in \Rn[n]$), except the one that is shared, have the same direction and norm (or are all zero).
Some examples that illustrate Propositions~\ref{prop:BIndependentCycles} and~\ref{prop:BAlmostIndependentCycles}, and details on equivalent incidence matrices are found in~\ref{app:PreliminariesAuxiliaryExamples}.
Also, proofs of Propositions~\ref{prop:BIndependentCycles} and~\ref{prop:BAlmostIndependentCycles} are found in Appendix~\ref{app:Proofs}.
These Propositions are useful in a later section, where we prove that for network graphs that satisfy the conditions of either Proposition, the agents converge to a configuration where all unit vectors belong to a common plane.

\begin{figure}
	\centering
	\subcapcentertrue	
	\subfigure[Graphs two independent cycles.]{
			\centering
			\begin{tikzpicture}[-,>=stealth',shorten >=1pt,auto,node distance=3cm,
			                    thick,main node/.style={circle,draw,font=\sffamily\Large\bfseries},scale=0.4,every node/.style={scale=0.4}]
			
			  \node[main node] (1) {$\nmbi[1]$};
			  \node[main node] (2) [above left of=1] {$\nmbi[2]$};
			  \node[main node] (3) [below left of=2] {$\nmbi[3]$};
			  \node[main node] (4) [above right of=1] {$\nmbi[4]$}; 
			  \node[main node] (5) [below right of=4] {$\nmbi[5]$};

			  \path[every node/.style={font=\sffamily\small}]
			    (1) edge node[xshift=0.7cm, yshift=0.25cm] {{\fontsize{0.05cm}{1em}$\bar{\kappa}(1,2) = 1 $}} (2)
			        edge node[xshift=1.0cm, yshift=-0.2cm] {{\fontsize{0.05cm}{1em}$\bar{\kappa}(1,4) = 4 $}} (4)
			    (2) edge node[xshift=-1.3cm, yshift=0.25cm] {{\fontsize{0.05cm}{1em}$\bar{\kappa}(2,3) = 2 $}} (3)
			    (3) edge node[xshift= 0.0cm, yshift=-0.4cm] {{\fontsize{0.05cm}{1em}$\bar{\kappa}(1,3) = 3 $}} (1)
			    (4) edge node[xshift=-0.0cm, yshift=-0.2cm] {{\fontsize{0.05cm}{1em}$\bar{\kappa}(4,5) = 6 $}} (5)
			    (5) edge node[xshift= 0.0cm, yshift= 0.0cm] {{\fontsize{0.05cm}{1em}$\bar{\kappa}(1,5) = 5 $}} (1)
			    ;
			\end{tikzpicture}	
			\label{fig:IndependentCycles}
	}
	\subfigure[Graph with one independent cycle and two cycles that share only one edge.]{	
			\centering
			\begin{tikzpicture}[-,>=stealth',shorten >=1pt,auto,node distance=3cm,
			                    thick,main node/.style={circle,draw,font=\sffamily\Large\bfseries},scale=0.4,every node/.style={scale=0.4}]
			
			  \node[main node] (1) {\color{gray}$\nmbi[1]$};
			  \node[main node] (2) [above left of=1] {\color{gray}$\nmbi[2]$};
			  \node[main node] (3) [below left of=2] {\color{gray}$\nmbi[3]$};
			  \node[main node] (4) [above right of=1] {\color{gray}$\nmbi[4]$}; 
			  \node[main node] (5) [below right of=4] {\color{gray}$\nmbi[5]$};
			  \node[main node] (6) [below left of=1] {$\nmbi[6]$};
			
			  \path[every node/.style={font=\sffamily\small}]
			    (1) edge node[xshift=0.7cm, yshift=0.25cm] {\color{gray}{\fontsize{0.05cm}{1em}$\bar{\kappa}(1,2) = 1 $}} (2)
			        edge node[xshift=1.0cm, yshift=-0.2cm] {\color{gray}{\fontsize{0.05cm}{1em}$\bar{\kappa}(1,4) = 4 $}} (4)
			        edge node[] {{\fontsize{0.05cm}{1em}$\bar{\kappa}(1,6) = 7 $}} (6)			        
			    (2) edge node[xshift=-1.3cm, yshift=0.25cm] {\color{gray}{\fontsize{0.05cm}{1em}$\bar{\kappa}(2,3) = 2 $}} (3)
			    (3) edge node[xshift= 0.0cm, yshift=-0.4cm] {\color{gray}{\fontsize{0.05cm}{1em}$\bar{\kappa}(1,3) = 3 $}} (1)
			    (4) edge node[xshift=-0.0cm, yshift=-0.2cm] {\color{gray}{\fontsize{0.05cm}{1em}$\bar{\kappa}(4,5) = 6 $}} (5)
			    (5) edge node[xshift= 0.0cm, yshift= 0.0cm] {\color{gray}{\fontsize{0.05cm}{1em}$\bar{\kappa}(1,5) = 5 $}} (1)	
			    (6) edge node {{\fontsize{0.05cm}{1em}$\bar{\kappa}(3,6) = 8 $}} (3)
			    ;
			\end{tikzpicture}		
		\label{fig:SharedCycles}
	}	
	\caption{Two graphs, one with two independent cycles and another with one independent cycle and two cycles that share only one edge.}
	\label{fig:Graphs}
\end{figure}
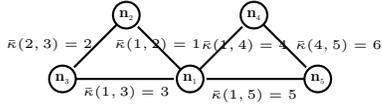
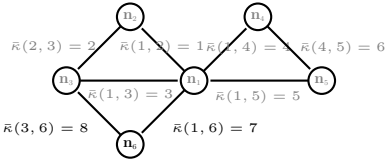

We now present a definition and some results that will prove useful in a later section.
\begin{defn}
	\label{def:AlphaCone}
	\normalfont
	We say that a group of unit vectors $\nmb = (\nmbi[1],\cdots,\nmbi[N]) \in (\Stwo)^{\sss{N}}$ belongs to an open (closed) $\alpha \in [0,\pi]$ cone, denoted by $\nmb \in \mathcal{C}(\alpha)$ ($\nmb \in \bar{\mathcal{C}}(\alpha)$), if there exists a unit vector $\nmb^{\sss{\star}} \in \mathcal{S}^{\sss{2}}$ such that  $\nmb^{\sss{\star T}}\nmbi[i] > \cos(\alpha)$ ($\nmb^{\sss{\star T}}\nmbi[i] \ge \cos(\alpha)$) for all $i \in \mathcal{N}$.
\end{defn}
The concept of open $\alpha$-cone is exemplified In Fig.~\ref{fig:Cone50Deg}, with three unit vectors $\nmbi[1]$, $\nmbi[2]$ and $\nmbi[3]$ contained in a $30^{\circ}$-cone formed by a unit vector $\nmb^{\sss{\star}}$. 
In fact, any group of unit vectors contained in the sphere surface region marked in bold is contained in a $30^{\circ}$-cone associated to the unit vector $\nmb^{\sss{\star}}$.

The following propositions follow from Definition~\ref{def:AlphaCone} and Proposition~\ref{prop:TriangularInequality} in Appendix~\ref{app:TriangularInequality}.
\begin{prop}
	\label{prop:BelongToCone1}
	\normalfont
	If $\nmb = (\nmbi[1],\cdots,\nmbi[N]) \in \mathcal{C}(\alpha)$, for some $\alpha \in [0,\frac{\pi}{2}]$, then $\max_{\sss{(i,j) \in \mathcal{N}^{\sss{2}}}} (1 - \nmbi[i]\tp\nmbi[j]) < 1 - \cos(2 \alpha)$. 
\end{prop}
This proposition follows immediately after Proposition~\ref{prop:TriangularInequality}.
\begin{prop}
	\label{prop:BelongToCone2}
	\normalfont
	If, given $\nmb = (\nmbi[1],\cdots,\nmbi[N]) \in (\Stwo)^{\sss{N}}$, $\max_{\sss{(i,j) \in \mathcal{N}^{\sss{2}}}} (1 - \nmbi[i]\tp\nmbi[j]) \le 1 - \cos(\frac{2}{3}\alpha)$ holds for some $\alpha \in [0,\pi]$, then $\nmb \in \bar{\mathcal{C}}(\alpha)$.
\end{prop}
A proof of Proposition~\ref{prop:BelongToCone2} is found in Appendix~\ref{app:TriangularInequality}.

\begin{figure}
	\centering
	\if\papercolor1
		\includegraphics[clip=true,trim=0cm 0cm 0cm 0cm,width=0.25\textwidth]
		{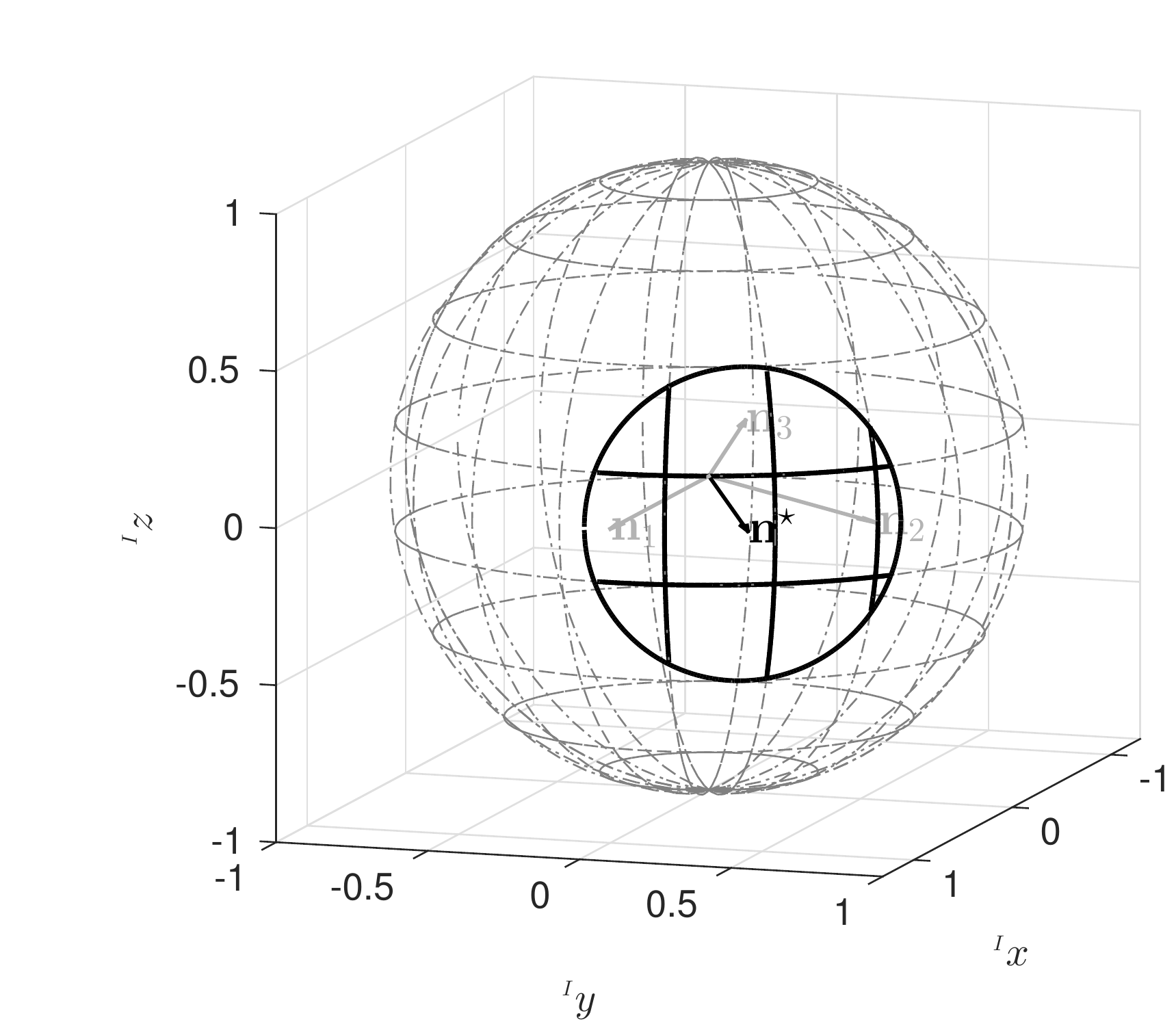}
	\else
		\includegraphics[clip=true,trim=0cm 0cm 0cm 0cm,width=0.25\textwidth]
		{./Figures/Cone30_bw}	
	\fi
	\caption{Three unit vectors $\nmbi[1]$, $\nmbi[2]$ and $\nmbi[3]$ in a $30^{\circ}$-cone associated to the unit vector $\nmb^{\sss{\star}}$.}
	\label{fig:Cone50Deg}
\end{figure}

		\subsection{Distance in $\mathcal{S}^2$}
		Consider an arbitrary distance function between unit vectors $d: \mathcal{S}^{2} \times \mathcal{S}^{2} \rightarrow \mathbb{R}_{0}^{+}$, satisfying $d(\nmbi[1],\nmbi[2]) = 0 \Leftrightarrow \nmbi[1] = \nmbi[2]$.
%
	We say that $d \in \mathcal{D}$, if, for all $(\nmbi[1],\nmbi[2]) \in \mathcal{S}^{\sss{2}} \times \mathcal{S}^{\sss{2}}$, it satisfies the partial differential equation
	\begin{align}
		\sk{\nmb_1}
		\frac{\partial d(\nmbi[1],\nmbi[2])}{\partial \nmbi[1]} 
		+
		\sk{\nmb_2}
		\frac{\partial d(\nmbi[1],\nmbi[2])}{\partial \nmbi[2]}
		=
		\zvec
		.
		\label{eq:DistanceFunctionProperty}
	\end{align}	
%
Let us motivate the introduction of~\eqref{eq:DistanceFunctionProperty}.
Given $d \in \mathcal{D}$, it follows that, along along a trajectory $\xmb(\cdot)$ of~\eqref{eq:StateVectorField}, (for brevity, we omit the time dependencies below)
\begin{align}	
	&
	\dot{d}(\nmbi[1],\nmbi[2])
	=
	\Scale[0.85]{
		\frac{\partial d(\nmbi[1],\nmbi[2])}{\partial \nmbi[1]}\tp \fmb_{\sss{\nmbi[1]}}(t,\nmbi[1],\Oi[1])
		+
		\frac{\partial d(\nmbi[1],\nmbi[2])}{\partial \nmbi[1]}\tp \fmb_{\sss{\nmbi[2]}}(t,\nmbi[2],\Oi[2])
	}
	\\
	&
	\Scale[0.9]{
		\overset{\sss{\eqref{eq:UnitVectorField}}}{=} 	
		\Omi{1}\tp \Rmati{1}\tp \sk{\nmbi[1]}
		\frac{\partial d(\nmbi[1],\nmbi[2])}{\partial \nmbi[1]}
		+
		\Omi{2}\tp \Rmati{2}\tp \sk{\nmbi[2]}
		\frac{\partial d(\nmbi[1],\nmbi[2])}{\partial \nmbi[2]}
	}
	\\
	&
	\Scale[0.9]{
		\overset{\sss{\eqref{eq:DistanceFunctionProperty}}}{=} 
		\twocol{\Omi{1}}{\Omi{2}}\tp
		\begin{bmatrix}
			 \Rmati{1}\tp & \zvec \\
			 \zvec & \Rmati{2}\tp 
		\end{bmatrix}
		\left(
			\twocol{\hphantom{-}1}{-1} \otimes \Idmat
		\right)
		\sk{\nmbi[1]}
		\frac{\partial d(\nmbi[1],\nmbi[2])}{\partial \nmbi[1]}
		.
	}
	\label{eq:ExplotingProperty}
\end{align}
where we identify an incidence matrix $[1 \, -1]\tp$ corresponding to an edge between unit vectors $\nmbi[1]$ and $\nmbi[2]$.
\ReviewAdded{Motivated by~\eqref{eq:ExplotingProperty}, let us define two functions.
Firstly, consider $M$ distance functions $d_{\sss{k}} \in \mathcal{D}$, one for each edge $k \in \mathcal{M}$.
Define then $\embi{k}: \mathcal{S}^{\sss{2}} \times \mathcal{S}^{\sss{2}} \mapsto \Rn[3]$ as
\begin{align}
	\embi{k}(\nmbi[1],\nmbi[2]) = \sk{\nmbi[1]} \frac{\partial d_{\sss{k}}(\nmbi[1],\nmbi[2])}{\partial \nmbi[1]},
	\label{eq:EdgeError}
\end{align}
to be the error of edge $k$, and for each $k \in \mathcal{M}$,.
Define also $\emb: (\mathcal{S}^{\sss{2}})^{\sss{N}} \mapsto \Rn[3 N]$ as
\begin{align}
	\emb(\nmb)
	=
	\begin{bmatrix}
		\embi{1}\tp(\nmbt{1},\nmbh{1})
		&
		\cdots
		&
		\embi{M}\tp(\nmbt{M},\nmbh{M})
	\end{bmatrix}
	\tp.
	\label{eq:TotalEdgeError}
\end{align}
Define, also, $D : (\mathcal{S}^{\sss{2}})^{\sss{N}} \mapsto \Rn[]_{\sss{\ge 0}}$ as
\begin{align}
	D(\nmb)
	=
	\sum\nolimits_{\sss{k = 1}}^{\sss{k = M}}
	d_{\sss{k}}(\nmbt{k},\nmbh{k}),
	\label{eq:TotalDistanceFunction}
\end{align}
named, hereafter, total distance function in the network of unit vectors. 
Note that $D(\nmb) = 0 \Leftrightarrow \exists \nmb^{\sss{\star}} \in \Stwo: \nmb = (\onesvec_{\sss{N}} \otimes \nmb^{\sss{\star}})$, which means Problem~\ref{prob:ProblemUnitVectorDynamics} is solved, if along a trajectory $\xmb(\cdot)$ of~\eqref{eq:StateVectorField}, $\lim_{\sss{t \rightarrow \infty }} D(\nmb(t)) = 0$.
It follows from~\eqref{eq:ExplotingProperty} that, if $d_{\sss{k}} \in \mathcal{D}$ for all $k \in \mathcal{M}$, then, along a trajectory $\xmb(\cdot)$ of~\eqref{eq:StateVectorField},
\begin{align}
	\dot{D}(\nmb(t))
	=
	\bm{\omega}(t)
	\bm{\Rmat}\tp(t)
	(B \otimes \Idmat) 
	\emb(\nmb(t))
	.
	\label{eq:TotalDistanceFunctionDerivative}
\end{align}
Notice that~\eqref{eq:TotalDistanceFunctionDerivative} follows from~\eqref{eq:DistanceFunctionProperty}, which is the reason for imposing such condition.}
Now, the point in question is what \emph{types} of distance functions satisfy~\eqref{eq:DistanceFunctionProperty}.
In Proposition~\ref{prop:UniquenessAuxiliar}, in Appendix, we prove that a distance function $d(\cdot,\cdot)$ satisfies~\eqref{eq:DistanceFunctionProperty}, i.e. $d \in \mathcal{D}$, if and only if there exits $f \in \mathcal{C}^{\sss{1}}((0,2),\Rn[{}]_{\sss{> 0}})$, with $\lim_{\sss{s \rightarrow 0}} f(s) = 0$, such that $d(\nmbi[1],\nmbi[2]) = f(1 - \nmbi[1]\tp \nmbi[2])$.
Physically, only distance functions that are functions of the angle between two unit vectors satisfy~\eqref{eq:DistanceFunctionProperty} (where the angle between $\nmbi[1]$ and $\nmbi[2]$ can be defined as $\theta(\nmbi[1],\nmbi[2]) = \arccos(\nmbi[1]\tp \nmbi[2])$, with $\theta : \mathcal{S}^{\sss{2}} \times \mathcal{S}^{\sss{2}} \mapsto [0,\pi]$). 
Moreover, distance functions that satisfy~\eqref{eq:DistanceFunctionProperty} are invariant to rotations of their arguments, i.e., $d(\Rmat \nmb_1, \Rmat\nmb_2) = d(\nmb_1,\nmb_2)$ for any $\Rmat \in \SO[3]$.
This property guarantees that the proposed controllers can be implemented without the need of a common inertial orientation frame, a property that is verified later.

To summarize, there are $M$ distance functions $d_{\sss{k}} \in \mathcal{D}$, one for each edge $k \in \mathcal{M}$, and, therefore, there are $M$ $f_{\sss{k}} \in \mathcal{C}^{\sss{1}}((0,2),\Rn[{}]_{\sss{> 0}})$, one for each edge $k \in \mathcal{M}$, such that $d_{\sss{k}}(\nmbi[1],\nmbi[2]) = f_{\sss{k}}(1 - \nmbi[1]\tp \nmbi[2])$.
For reasons that will become apparent later, we restrict the previous functions $f_{\sss{k}}(\cdot)$ to satisfy some more properties.
\begin{defn}
	\label{eq:gClasses}
	\normalfont
	Consider a function $f \in \mathcal{C}^{\sss{2}}((0,2),\mathbb{R}_{\sss{> 0}})$, satisfying \emph{i)} $ f'(s) > 0 \forall s \in (0,2)$, \emph{ii)} $\lim_{\sss{s \rightarrow 0^{+}}} f(s) = 0$, and \emph{iii)} $\limsup_{\sss{s \rightarrow 0^{+}}} f'(s),f''(s) < \infty$.
	Denote $f_{\sss{2}} := \lim_{\sss{s \rightarrow 2^{-}}} f(s) $ and $f_{\sss{0}}^{\sss{\prime}} := \lim_{\sss{s \rightarrow 0^{+}}} f'(s) $.
	We say 
	\begin{itemize}
		\item $f \in \mathcal{P}_{\sss{0}}$ if $f_{\sss{0}}^{\sss{\prime}} = 0 $ and $f \in \mathcal{P}_{\sss{\bar{0}}}$ if $f_{\sss{0}}^{\sss{\prime}} \ne 0 $,
		\item $f \in \mathcal{P}^{\sss{\infty}}$ if $f_{\sss{2}} = \infty$, and $f \in \mathcal{P}^{\sss{\bar{\infty}}}$ if $f_{\sss{2}} < \infty$,
		\item $f \in \mathcal{P}^{\sss{0}}$ if $f \in \mathcal{P}^{\sss{\bar{\infty}}}  \wedge \lim_{\sss{s \rightarrow 2^{-}}} f'(s) \sqrt{2 - s} = 0$,
		\item $f \in \mathcal{P}^{\sss{\bar{0}}}$ if $ f \in \mathcal{P}^{\sss{\bar{\infty}}}  \wedge \lim_{\sss{s \rightarrow 2^{-}}} f'(s) \sqrt{2 - s}  \ne 0 $,
		\item $f \in \bar{\mathcal{P}}$ if $f(\cdot)$ is of any of the previous classes.
	\end{itemize}
\end{defn}
\begin{figure}
	\centering
	\includegraphics[width=1\linewidth]{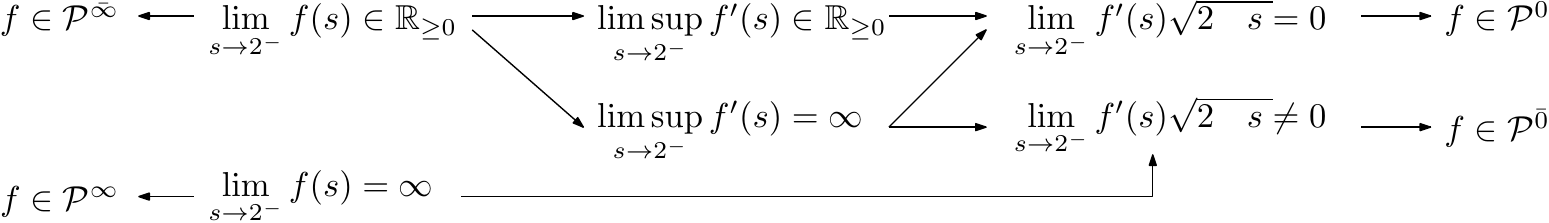}
	\caption{Relation between properties of $f(\cdot)$ and the classes it belongs to.}
	\label{fig:ClassesOf_f_Functions}
\end{figure}
Figure~\ref{fig:gClasses} illustrates the different classes introduced in Definition~\ref{eq:gClasses} while Fig.~\ref{fig:ClassesOf_f_Functions} illustrates how the properties that $f(\cdot)$ satisfies affects the classes it belongs to (see Remark~\ref{rem:DistanceFunction} in Appendix). 
In~\cite{tron2012intrinsic}, the notion of \emph{reshaping function} is introduced, whose definition is within the same spirit as that of Definition~\ref{eq:gClasses}.
\review{\label{com:c4}
For the rest of this manuscript, we assume that, for each edge $k \in \mathcal{M}$, there exists a function $d_{\sss{k}}(\nmbi[1],\nmbi[2]) = f_{\sss{k}}(1 - \nmbi[1]\tp\nmbi[2])$ where $f_{\sss{k}} \in \bar{\mathcal{P}}$; in particular,   $f_{\sss{k}}^{\sss{\prime}}(\cdot)$ plays the role of an edge weight. 
In e.g.~\cite{olfati2006swarms,moshtagh2007distributed}, $f_{\sss{k}}(s) = a_{\sss{k}} s$ and $f_{\sss{k}}^{\sss{\prime}}(s) = a_{\sss{k}}$, for all $k \in \mathcal{M}$ ($a_{\sss{k}}$ is the weight of edge $k$ and it is denoted by $a_{\sss{ij}}$ in~\cite{olfati2006swarms}, where $(i,j)= \bar{\kappa}^{\sss{-1}}(k)$).
}
As such, \eqref{eq:TotalDistanceFunction} and \eqref{eq:EdgeError} can be rewritten as $D : \Omega_{\sss{\nmb}}^{\sss{D}} \mapsto \Rn[]_{\sss{\ge 0}}$ (see~\eqref{eq:TotalDistanceFunctionDomain})
\begin{align}
	D(\nmb)
	=
	\sum\nolimits_{\sss{k = 1}}^{\sss{k = M}}
	f_{\sss{k}}(1 - \nmbt{k}\tp\nmbh{k})
	\label{eq:TotalDistanceFunction2}
\end{align}
and $\emb_{\sss{k}} : \{ (\nmbi[1],\nmbi[2]) \in \Stwo \times \Stwo : \nmbi[1]\tp\nmbi[2] \ne -1 \text{ if } f_{\sss{k}} \in \mathcal{P}^{\sss{\bar{0}}} \} \mapsto \Rn[]_{\sss{\ge 0}}$ as
\begin{align}
	&
	\hspace{-0.5cm}
	\Scale[0.85]{
		\embi{k}(\nmbi[1],\nmbi[2]) 
		=
		\sk{\nmbi[1]} \frac{\partial f_{\sss{k}}(1 - \nmbi[1]\tp \nmbi[2])}{\partial \nmbi[1]}
		=
		f_{\sss{k}}^{\sss{\prime}}(1 - \nmbi[1]\tp \nmbi[2])
		\sk{\nmbi[1]} \nmbi[2],
	}
	\label{eq:EdgeError2}
	\\
	&
	\Rightarrow
	\Scale[0.85]{
		\|\embi{k}(\nmbi[1],\nmbi[2]) \|
		= 
		f_{\sss{k}}^{\sss{\prime}}(s) \sqrt{s( 2 - s)}|_{\sss{s = 1 - \nmbi[1]\tp\nmbi[2]}}
	},
	\label{eq:EdgeError2Norm}	
\end{align}
respectively. 
Note that $\embi{k}(\nmbi[1],\nmbi[2])$ is well defined for all $(\nmbi[1],\nmbi[2]) \in \mathcal{S}^{2}\times \mathcal{S}^{2}$ if  $f_{\sss{k}} \in \mathcal{P}^{\sss{0}}$, and if $f_{\sss{k}} \in \mathcal{P}^{\sss{\bar{0}}}$, note that $\lim\limits_{\sss{\nmbi[2] \rightarrow \nmbi[1]}} \embi{k}(\nmbi[1],\nmbi[2]) = \lim\limits_{\sss{s \rightarrow 2^{\sss{-}}}} f_{\sss{k}}^{\sss{\prime}}(s) \sqrt{s( 2 - s)} \lim\limits_{\sss{\nmbi[2] \rightarrow \nmbi[1]}} \frac{\sk{\nmbi[1]} \nmbi[2]}{\|\sk{\nmbi[1]} \nmbi[2]\|}$ does not exist.
Note that the total distance function~\eqref{eq:TotalDistanceFunction2} depends on $f_{\sss{k}}(\cdot)$, for all $k \in \mathcal{M}$, while~\eqref{eq:TotalEdgeError} depends on $f^{\sss{\prime}}_{\sss{k}}(\cdot)$, for all $k \in \mathcal{M}$. 
As such, a distance function may or may not be defined when two unit vectors are diametrically opposed, depending on whether $f \in \mathcal{P}^{\sss{\infty}}$ or $f \in \mathcal{P}^{\sss{\bar{\infty}}}$;
similarly, an edge error may or may not be defined when two unit vectors are diametrically opposed, depending on whether $f \in \mathcal{P}^{\sss{0}}$ or $f \in \mathcal{P}^{\sss{\bar{0}}}$.
Thus, the domains of~\eqref{eq:TotalDistanceFunction} and of~\eqref{eq:TotalEdgeError} depend on the classes $f_{\sss{k}}(\cdot)$ belongs to, for all $k \in \mathcal{M}$.
In particular, $D : \Omega_{\sss{\nmb}}^{\sss{D}} \mapsto \Rn[]_{\sss{\ge 0}}$ and $\emb: \Omega_{\sss{\nmb}}^{\sss{\emb}} \mapsto \Rn[3 N]$, where
\begin{align}
	\Omega_{\sss{\nmb}}^{\sss{D}} 
	=
	& 
	\{  \nmb \in (\mathcal{S}^{\sss{2}})^{\sss{N}}:  \nmbt{k}\tp\nmbh{k} \ne -1 \, \forall f_{\sss{k}} \in \mathcal{P}^{\sss{\infty}} \},
	\label{eq:TotalDistanceFunctionDomain}
	\\
	\Omega_{\sss{\nmb}}^{\sss{\emb}} 
	=
	& 
	\{  \nmb \in (\mathcal{S}^{\sss{2}})^{\sss{N}}:  \nmbt{k}\tp\nmbh{k} \ne -1 \, \forall f_{\sss{k}} \in \mathcal{P}^{\sss{\bar{0}}} \},	
\end{align}
and where we emphasize that $\Omega_{\sss{\nmb}}^{\sss{\emb}}  \subseteq \Omega_{\sss{\nmb}}^{\sss{D}}$, since $f_{\sss{k}} \in \mathcal{P}^{\sss{\infty}} \Rightarrow f_{\sss{k}} \in \mathcal{P}^{\sss{\bar{0}}}$ (see Fig.~\ref{fig:ClassesOf_f_Functions}).
These domains play a role later on, since $D(\cdot)$ is used in constructing a Lyapunov function, while $\emb(\cdot)$ is used in constructing the control law.
As such, the Lyapunov function can be well defined, while the control law is not, while if the control law is well defined, so is the Lyapunov function.
Consequently, it is important to guarantee that along trajectories of the closed-loop system, the control law is well defined.
Additionally, notice that~\eqref{eq:EdgeError2} provides some insight on why we denote $\embi{k}(\cdot,\cdot)$ as edge error of edge $k$.
Indeed, if $f_{\sss{k}} \in \bar{\mathcal{P}} , \forall k \in \mathcal{M}$, it follows that $\embi{k}(\nmbt{k},\nmbh{k}) = \zvec$ implies that $\nmbt{k} = \pm \nmbh{k}$, i.e., it implies that the neighbors that form edge $k$ are either synchronized or diametrically opposed. 
Moreover, if $f_{\sss{k}} \in \bar{\mathcal{P}} \, \forall k \in \mathcal{M}$, the distance between unit vectors is supremum when two unit vectors are diametrically opposed, i.e., for each $k \in \mathcal{M}$, (denote $\Omega = \{ (\nmbi[1],\nmbi[2])\in \mathcal{S}^{\sss{2}} \times \mathcal{S}^{\sss{2}} : \nmbi[1]\tp\nmbi[2] = -1  \}$)
\begin{align}
	\hspace{-0.3cm}
	\Scale[0.82]{
		\sup\limits_{\sss{ (\nmbi[1],\nmbi[2]) \in \Omega  }}
		\, 
		d_{\sss{k}}(\nmbi[1],\nmbi[2])
		=
		\sup\limits_{\sss{ (\nmbi[1],\nmbi[2]) \in \Omega  }}
		\, 
		f_{\sss{k}}(1 - \nmbi[1]\tp\nmbi[2])
		=
		\lim\limits_{\sss{s \rightarrow 2}} \, f_{\sss{k}}(s)	
		=:
		d_{\sss{k}}^{\sss{\max}}.
	}
	\label{eq:MaximumDistance}
\end{align}
For convenience, denote
\begin{align}
	\label{eq:Dmin}
	d^{\sss{\min}} 
	:=
	\min_{\sss{k \in \mathcal{M}}}
	d_{\sss{k}}^{\sss{\max}},
\end{align}
which plays an important role in this and the following sections.
\ReviewAdded{\begin{prop}
	\label{prop:RelationDistanceAndEdgeError}
	\normalfont
	Consider the total edge error in~\eqref{eq:TotalEdgeError} and the total distance function in~\eqref{eq:TotalDistanceFunction2}.
	Consider $\Omega_{\sss{\nmb}}^{\sss{\prime}}: \Rn[]_{\sss{\ge 0}} \rightrightarrows \Omega_{\sss{\nmb}}^{\sss{D}}$ as $\Omega_{\sss{\nmb}}^{\sss{\prime}}(\bar{D}) = \{ \nmb \in \Omega_{\sss{\nmb}}^{\sss{D}} : D(\nmb)  \le \bar{D}\}$, where $\Omega_{\sss{\nmb}}^{\sss{\prime}}(\bar{D})$ is compact for all positive $\bar{D}$ (where $\rightrightarrows$ refers to a set-valued function).
	Then, it follows that
	\begin{align}
		\forall \bar{D} < d^{\sss{\min}} 
		\,
		,
		\max_{\sss{\nmb \in \Omega_{\sss{\nmb}}^{\sss{\prime}}(\bar{D}) }}
		\| \emb(\nmb)\|
		< \infty,
	\end{align}
	and that there are no diametrically opposed neighbors, i.e., $| \{ q \in \mathcal{M} : \forall \nmb \in \Omega_{\sss{\nmb}}^{\sss{\prime}}(\bar{D}) , \nmbt{q}\tp\nmbh{q} = -1 \} | = 0$.
	If $f_{\sss{k}} \in \mathcal{P}^{\sss{0}}$ for all $k \in \mathcal{M}$, it follows that $\max_{\sss{\nmb \in \Omega_{\sss{\nmb}}^{\sss{\emb}}}} \| \emb(\nmb) \| = \max_{\sss{\nmb \in (\mathcal{S}^{\sss{2}})^{\sss{N}} }} \| \emb(\nmb) \| < \infty$;
	moreover, given $\bar{D} < p d^{\min} $ for some $p \in \mathcal{M}$, it follows that there are at most $p - 1$ diametrically opposed neighbors, i.e., $| \{ q \in \mathcal{M} : \forall \nmb \in \Omega_{\sss{\nmb}}^{\sss{\prime}}(\bar{D}) , \nmbt{q}\tp\nmbh{q} = -1 \} | \le p -1$.	
\end{prop}
\begin{proof}
	\normalfont
	Recall Definition~\ref{eq:gClasses}, \eqref{eq:TotalDistanceFunction2} and~\eqref{eq:EdgeError2}.
	It follows that $\emb(\cdot)$ and $D(\cdot)$ grow unbounded \emph{only if} two neighbors are diametrically opposed (however, this is not sufficient).
	%
	Recall also that $\Omega_{\sss{\nmb}}^{\sss{\emb}} \subseteq \Omega_{\sss{\nmb}}^{\sss{D}}$, meaning that $\emb(\cdot)$ can grow unbounded while $D(\cdot)$ remains bounded.
	Finally, notice that, if $\nmb \in (\mathcal{S}^{\sss{2}})^{\sss{N}}$ is such that there are $q$ diametrically opposed neighbors, it follows that $D(\nmb) \ge q d^{\sss{\min}}$. 
	Consider then a $\bar{D} < d^{\sss{\min}} \le \infty$.
	Then, for all $\nmb \in \Omega_{\sss{\nmb}}^{\sss{\prime}}(\bar{D})$, it follows that $D(\nmb) \le  \bar{D} \Rightarrow \forall k \in \mathcal{M} ,  \nmbt{k}\tp \nmbh{k} \ge 1 - f_{\sss{k}}^{\sss{-1}}(\bar{D}) > -1$, which means that two neighbors are not arbitrarily close to a configuration where they are diametrically opposed and, therefore, $\max_{\sss{ \nmb \in \Omega_{\sss{\nmb}}^{\sss{\prime}}(\bar{D})}} \|\emb(\nmb)\|$ is bounded.
	For the second part of the Proposition, consider that $f_{\sss{k}} \in \mathcal{P}^{\sss{0}}$ for all $k \in \mathcal{M}$.
	Then $\Omega_{\sss{\nmb}}^{\sss{\emb}} =  \Omega_{\sss{\nmb}}^{\sss{D}} = (\mathcal{S}^{\sss{2}})^{\sss{N}}$, which means $\max_{\sss{ \nmb \in \Omega_{\sss{\nmb}}^{\sss{\prime}}(\bar{D})}} \|\emb(\nmb)\|$ and $\max_{\sss{ \nmb \in \Omega_{\sss{\nmb}}^{\sss{\prime}}(\bar{D})}} D(\nmb)$ are bounded (since the domain is compact, and the functions are continuous).
	Consider then a $p \in \mathcal{M} $ and a $\bar{D} < p d^{\sss{\min}} \le \infty$.
	Since $D(\nmb) \ge q d^{\sss{\min}}$, then, if there are $q$ diametrically opposed neighbors in $\nmb \in (\mathcal{S}^{\sss{2}})^{\sss{N}}$, it follows that $| \{ q \in \mathcal{M} : \forall \nmb \in \Omega_{\sss{\nmb}}^{\sss{\prime}}(\bar{D}) , \nmbt{q}\tp\nmbh{q} = -1 \} | \le p -1$.
\end{proof}}

\begin{figure}
	\centering
		\if\papercolor1
			\includegraphics[clip=true,trim=3cm 7.5cm 4cm 7.5cm,width=0.25\textwidth]{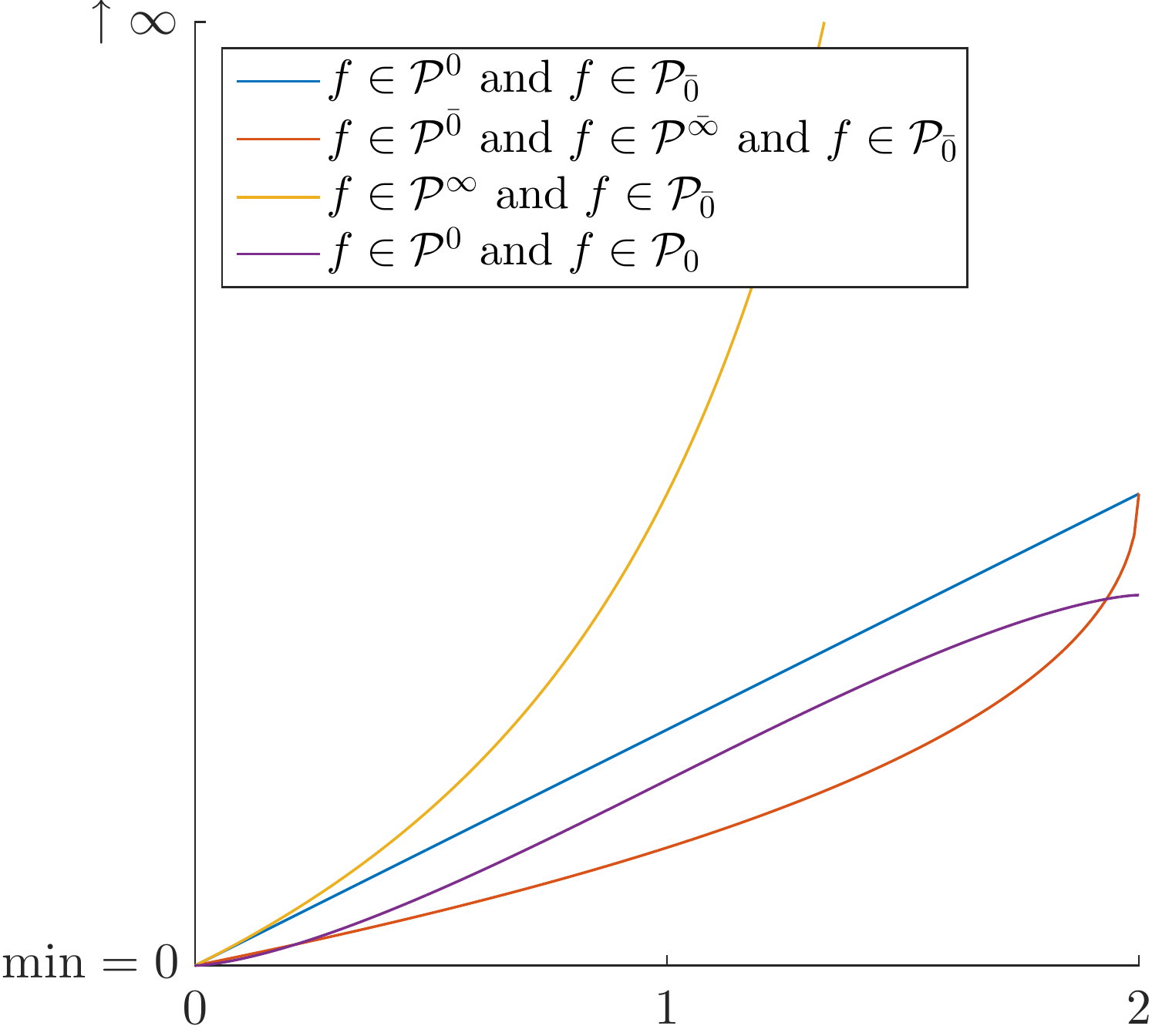}	
		\else
			\includegraphics[clip=true,trim=3cm 7.5cm 4cm 7.5cm,width=0.25\textwidth]{./Figures/fClassesPlot}
		\fi
	\caption{Four functions belonging to different classes as introduced in Definition~\ref{eq:gClasses}: (from top to bottom in legend) $f(s) = s$, $f(s) = \pi^{\sss{-2}}\arccos^{\sss{2}}(1 - s)$, $f(s) = \tan^{\sss{2}}\left(0.5 \arccos(1 - s)\right)$ and $f(s) = 0.25(\sqrt{s (2 -s)}(s -1) + \arccos(1 -s))$.} 
	\label{fig:gClasses}
\end{figure}

		\subsection{Solution to Problem~\ref{prob:ProblemUnitVectorDynamics}}
		In this section, we present the controllers for the torques of each agent~$i$.
\review{\label{com:c5}For each agent~$i$, we design a controller that is a function of $|\mathcal{N}_{\sss{i}}|+1$ measurements: $|\mathcal{N}_{\sss{i}}|$ measurements corresponding to the \emph{distance measurements} between agent~$i$ and its $|\mathcal{N}_{\sss{i}}|$ neighbors, and $1$ measurement corresponding to the body frame angular velocity.
More specifically, we assume each agent~$i$ measures $\Rmati{i}\tp(t) \nmbi[j](t) = \Rmati{i}\tp(t) \Rmati{j}(t) \nmbibody[j]$, at each time instant $t \ge 0$, and for each $j \in \mathcal{N}_{\sss{i}}$; physically, this means that agent~$i$ knows $\nmbibody[j]$ (the \emph{constant} unit vector that it is required to synchronize with), and that it can measure the projection of this unit vector on its orientation frame;}
each agent~$i$ must also measure $\Oi(t)$, which does not require an inertial reference frame.
For convenience denote $\mathcal{N}_{\sss{i}} = \{i_{\sss{1}},\cdots, i_{\sss{|\mathcal{N}_{\sss{i}}|}}\}$, and, given $\nmbibody[i] \in \Stwo$, denote %
$
		\Omega_{\nmb_{i}} 
		= 
		\{ 
			(\nmbi[i_{\sss{1}}],\cdots,\nmbi[i_{|\mathcal{N}_{\sss{i}}|}]) \in \mathcal{S}^{\sss{2 |\mathcal{N}_{\sss{i}}|}} :   
			\nmbibody[i]\tp \nmbi[\sss{i_l}] \ne  -1 , 
			\forall 
			l \in \{1,\cdots, |\mathcal{N}_{\sss{i}}| \} 
			\wedge  
			f_{\sss{\kappa(i,i_{\sss{l}})}} \in \mathcal{P}^{\sss{\bar{0}}} 
		\}
$
%
%
which provides the domain where the control law for agent~$i$ is well defined (recall that if $f_{\sss{k}} \in \mathcal{P}^{\sss{\bar{0}}}$, for some $k \in \mathcal{M}$, then~\eqref{eq:EdgeError2} is not defined when two unit vectors are diametrically opposed).	 
We then propose, for each agent $i \in \mathcal{N}$, the following decentralized control law $\Tmbi^{\sss{cl}}: (\bm{\nu}_{\sss{i}},\Omi{i}) = ((\Rmati{i}\tp\nmbi[i_{\sss{1}}],\cdots,\Rmati{i}\tp\nmbi[i_{|\mathcal{N}_{\sss{i}}|}]),\Omi{i}) \in  \Omega_{\nmb_{i}} \times \Rn[3] \mapsto \Rn[3]$
\begin{align}
	\hspace{-0.7cm}
	\Tmbi[i]^{\sss{cl}}(\bm{\nu}_{\sss{i}},\Oi)	
	=	
	- 
	\bm{\sigma}( \Omi{i})
	-
	\sum\nolimits_{\sss{l = 1}}^{\sss{l = |\mathcal{N}_{\sss{i}}|}} 
	\emb_{\sss{\kappa(i,i_{\sss{l}})}}
	\left(
		\nmbibody,
		\Rmati{i}\tp\nmbi[i_{\sss{l}}]
	\right)
	,	
	\label{eq:DistributedControlLawDynamics} 		
\end{align}
with $\bm{\sigma} \in \Sigma$ (see Definition~\ref{def:SigmaFunction} in Appendix).
The timed control laws for each agent~$i \in \mathcal{N}$ are then $\Tmbi[i]: \Rn[]_{\sss{\ge 0}} \mapsto \Rn[3]$ as
\begin{align}
	\hspace{-0.5cm}
	\Scale[0.9]{
		\Tmbi[i](t) = \Tmbi[i]^{\sss{cl}}((\Rmati{i}\tp(t)\nmbi[j_{\sss{1}}](t),\cdots,\Rmati{i}\tp(t)\nmbi[j_{|\mathcal{N}_{\sss{i}}|}](t)),\Oi(t))
	}.
	\label{eq:DistributedControlLawDynamicsTimed}
\end{align}
\review{\label{com:c6}
The proposed torque control law exhibits the following properties.
The controller function in~\eqref{eq:DistributedControlLawDynamicsTimed} is decentralized in the sense that it does not depend on the measurement of the global state $\xmb(\cdot)$.
Also, \eqref{eq:DistributedControlLawDynamicsTimed} can be implemented without the knowledge of an inertial reference, since measuring $\ymb_{\sss{l}}^{\sss{i}}(t) := \Rmati{i}\tp(t)\Rmati{i_{l}}(t)\nmbibody[i_{l}]$ at every time instant $t \ge 0$ and for all $l \in \{1,\cdots, |\mathcal{N}_{\sss{i}}| \} $ only requires the measurement of the projection of $\nmbibody[i_{l}]$ in agent's~$i$ body orientation frame; 
while $\ymb_{\sss{|\mathcal{N}_{\sss{i}}|+1}}^{\sss{i}}(t) := \Oi[i](t)$ is also measured in agent's~$i$ body orientation frame.
Finally, notice that $\norm{\Tmbi[i](\cdot)} \le \sigma^{\sss{\max}} + |\mathcal{N}_i| \max_{\sss{j \in \mathcal{N}_i}} \sup_{\sss{0 < s < 2} } f'_{\sss{\kappa(i,j)}}(s)$. 
%
As such, the proposed control law, for each agent $i$, can be implemented with bounded actuation provided that $\sigma^{\sss{\max}}<\infty$ and that $f_{\sss{\kappa(i,j)}} \in \mathcal{P}^{\sss{0}}$ for all $j \in \mathcal{N}_{\sss{i}}$.
}

Notice that
$
	\sum_{\sss{ k \in \mathcal{M}' \subseteq \mathcal{M}}}
	\|\emb_{\sss{k}}(\cdot,\cdot)\|
	\le 
	\sum_{\sss{ k \in \mathcal{M}}}
	\|\emb_{\sss{k}}(\cdot,\cdot)\|	
	\le 
	\sqrt{M} \| \emb(\cdot) \|,
$
and, therefore, for any $\Rmati{i} \in \SO[3]$, and for all $\xmb \in \Omega_{\sss{\nmb}}^{\sss{\emb}} \times \Rn[3N]$,
\begin{align}
	\hspace{-0.7cm}
	\Scale[0.87]{
		\| \Tmbi[i]^{\sss{cl}}((\Rmati{i}\tp\nmbi[i_{\sss{1}}],\cdots,\Rmati{i}\tp\nmbi[i_{|\mathcal{N}_{\sss{i}}|}]),\Oi)\|	
		\le 
		\sigma^{\sss{\prime}} \| \Omi{i} \|
		+ 
		\sqrt{M}
		\| \emb(\nmb) \|
	}
	,
	\label{eq:DistributedControlLawDynamicsBound} 		
\end{align}
a relation to be used later.
We can now present the complete control law $\Tmb^{\sss{cl}}: \Rn[]_{\sss{\ge 0}} \times \Omega_{\sss{\nmb}}^{\sss{\emb}} \times \Rn[3 N] \mapsto \Rn[3 N]$,
\begin{align}
	&
	\hspace{-0.65cm}
	\Tmb^{\sss{cl}}(t,\xmb)
	=
	\begin{bmatrix}
		\Tmbi[1]((\Rmati{1}\tp(t)\nmbi[1_{\sss{1}}],\cdots,\Rmati{1}\tp\nmbi[1_{|\mathcal{N}_{\sss{1}}|}]),\Oi[1])\\
		\vdots \\
		\Tmbi[N]((\Rmati{N}\tp(t)\nmbi[N_{\sss{1}}],\cdots,\Rmati{N}\tp\nmbi[N_{|\mathcal{N}_{\sss{N}}|}]),\Oi[N])
	\end{bmatrix}	
	\\
	&
	\hspace{-0.65cm}	
	=
	-
	[
		\bm{\sigma}\tp( \Omi{1})\,
		\cdots \,
		\bm{\sigma}\tp(\Omi{N})
	]\tp
	-
	\Rmat\tp(t)
	(B \otimes \Idmat)
	\emb(\nmb),
	\label{eq:DistributedControlLawVectorial}
\end{align}
where $\Rmat(\cdot) = \Rmati{1}(\cdot) \oplus \cdots \oplus \Rmati{N}(\cdot)$ (see Notation).
For the remainder of this paper, we dedicate efforts in studying the equilibria configurations induced by this control law (for different types of graphs), their stability, and what is the effect of the chosen distance functions.
Notice that~\eqref{eq:DistributedControlLawVectorial} is defined on $\Rn[]_{\sss{\ge 0}} \times \Omega_{\sss{\nmb}}^{\sss{\emb}} \times \Rn[3 N]$.
As such, when $f_{\sss{k}} \in \mathcal{P}^{\sss{0}} \, \forall k \in \mathcal{M}$, $\Omega_{\sss{\nmb}}^{\sss{\emb}} = (\mathcal{S}^{\sss{2}})^{\sss{N}}$, and the analysis is simplified;
when, however, $\exists k \in \mathcal{M} : f_{\sss{k}} \in \mathcal{P}^{\sss{\bar{0}}} $, $\Omega_{\sss{\nmb}}^{\sss{\emb}} \subset  (\mathcal{S}^{\sss{2}})^{\sss{N}}$ (where $\Omega_{\sss{\nmb}}^{\sss{\emb}}$ is open), and it is necessary to guarantee that a trajectory $\xmb(\cdot)$ never approaches the boundary of $\Omega_{\sss{\nmb}}^{\sss{\emb}}$.

		\subsection{Lyapunov Function}
		\label{subsubsec:LyapunovFunction}
In addition to the total distance function of the network~\eqref{eq:TotalDistanceFunction}, let us also define the total rotational kinetic energy of the network, as $H: \Rn[3N] \mapsto \Rn[{}]_{\sss{\ge 0}}$, where
\begin{align}
	H(\bm{\omega}) 
	= 
	\frac{1}{2}
	\sum\nolimits_{\sss{i =1}}^{\sss{i = N}} 
	\Omi{i}\tp J_{\sss{i}}\Omi{i}
	\label{eq:TotalKineticEnergy}
\end{align}
and which satisfies $\frac{\partial H(\bm{\omega})}{\partial \xmb}\tp \fmb_{\sss{\xmb}}(\cdot,\xmb,\Tmb) =  \frac{\partial H(\bm{\omega})}{\partial \bm{\omega}}\tp \fmb_{\sss{\bm{\omega}}}(\bm{\omega},\Tmb)  =	\sum\nolimits_{\sss{i =1}}^{\sss{i = N}} 
	\Omi{i}\tp J_{\sss{i}} 	\fmb_{\sss{\Oi}}(\Oi[i], \Tmbi[i])	 =\bm{\omega}\tp \Tmb$, for all $(\bm{\omega},\Tmb) \in \Rn[6 N] $.
Combining~\eqref{eq:TotalDistanceFunction2} and~\eqref{eq:TotalKineticEnergy}, consider then the Lyapunov function $V : \Omega_{\sss{\nmb}}^{\sss{D}} \times \Rn[3N] \mapsto \Rn[{}]_{\sss{\ge 0}}$,
\begin{align}
		V(\xmb)
		& 
		=
		D(\nmb)
		+
		H(\bm{\omega})
		,
		\label{eq:Lyapunov}
\end{align}
and the function $W : \Rn[3N] \mapsto \Rn[{}]_{\sss{\ge 0}}$ defined as
\begin{align}
		W(\bm{\omega})
		=
		&
		-
		\frac{\partial V(\xmb)}{\partial \xmb}\tp
		\fmb_{\sss{\xmb}}(\cdot,\xmb,\Tmb^{\sss{cl}}(\cdot,\xmb))
		\\
		\overset{\sss{\eqref{eq:TotalDistanceFunctionDerivative}}}{=}
		&
		-
		\bm{\omega}\tp
		\Rmat\tp(\cdot)
		(B \otimes \Idmat)
		\emb(\nmb)
		-
		\bm{\omega}\tp
		\Tmb^{\sss{cl}}(\cdot,\xmb)
		\\
		\overset{\sss{\eqref{eq:DistributedControlLawVectorial}}}{=}
		& 
		\sum\nolimits_{\sss{i =1}}^{\sss{i = N}} 
		\Omi{i}\tp \bm{\sigma}(\Omi{i})
		,
		\label{eq:LyapunovDerivative}
\end{align}
and it follows that, along a trajectory $\xmb(\cdot)$ of $\dot{\xmb}(t) = \fmb_{\sss{\xmb}}(t,\xmb(t),\Tmb^{\sss{cl}}(t,\xmb))$, $\dot{V}(\xmb(t)) = - W(\bm{\omega}(t)) \le 0, \forall t \ge 0$.
Note that, for every $i \in \mathcal{N}$,
\begin{align}
	&
	\hspace{-0.6cm}
	\Scale[0.76]{
		\|\fmb_{\sss{\bm{\omega}_{\sss{i}}}}(\Oi[i],\Tmbi[i]^{\sss{cl}}(\cdot) )\|
		\overset{\sss{\eqref{eq:OmegaVectorField}}}{\le}
		\frac{1}{\lambda_{\sss{\min}}(J_{\sss{i}})}
		\left(
			\Tmbi[i]^{\sss{cl}}(\cdot)
			+
			\lambda_{\sss{\max}}(J_{\sss{i}}) \norm{\Oi[i]}^2
		\right)
	}
	\\
	&
	\hspace{-0.6cm}
	\Scale[0.76]{
		\overset{\sss{\eqref{eq:DistributedControlLawDynamicsBound}}}{\le}
		\frac{1}{\lambda_{\sss{\min}}(J_{\sss{i}})}
		\left(
			\sigma^{\sss{\prime}} \| \Oi[i] \|
			+ 
			\sqrt{M} \| \emb(\nmb) \|
			+
			\lambda_{\sss{\max}}(J_{\sss{i}}) \norm{\Oi[i]}^2
		\right)	
		=:
		f_{\sss{\bm{\omega}_{\sss{i}}}}^{\sss{\infty}}(\xmb).
	}
	\label{eq:omegaiDotUniformContinuity}
\end{align}
Moreover, along a trajectory $\xmb(\cdot)$, (for brevity, below, we omit the time dependency of the state)
\begin{align}
		& 
		\hspace{-0.6cm}
		\dot{W}(\bm{\omega})
		=
		\frac{\partial W(\xmb)}{\partial \xmb}\tp
		\fmb_{\sss{\xmb}}(t,\xmb,\Tmb^{\sss{cl}}(t,\xmb))
		\\
		&	
		\hspace{-0.6cm}
		=
		\sum\nolimits_{\sss{1 = 1}}^{\sss{i = N}} 
		\left(
		\bm{\sigma}( \Omi{i}) + \Omi{i}\tp \frac{\partial \bm{\sigma}( \Omi{i})}{\partial \Omi{i}} 
		\right)\tp
		\fmb_{\sss{\bm{\omega}_{\sss{i}}}}(\Oi[i],\Tmbi[i]^{\sss{cl}}(\cdot) ),
		\label{eq:Wdot}
\end{align}
for which it follows that $\forall t \ge 0$
\begin{align}
		& 
		\hspace{-0.5cm}
		|\dot{W}(\bm{\omega}(t))|
		\le
		\sum\nolimits_{\sss{1 = 1}}^{\sss{i = N}} 
		\left(
			\sigma_{\sss{s}} + \sigma^{\sss{\prime}}
		\right)\tp
		\| \Oi[i](t) \| 
		f_{\sss{\bm{\omega}_{\sss{i}}}}^{\sss{\infty}}(\xmb(t)).
		\label{eq:WDot}
\end{align}
Along the same trajectory, it also follows that, for every $i \in \mathcal{N}$ (we again omit the time dependency)
\begin{align} 
		&
		\hspace{-0.6cm}
		\Scale[0.80]{
			\ddot{\bm{\omega}}_{\sss{i}}
			=
			-
			J_{\sss{i}}^{\sss{-1}}
			\left(
				\sk{\dot{\bm{\omega}}_{\sss{i}}} J_{\sss{i}} \bm{\omega}_{\sss{i}}
				+
				\sk{\bm{\omega}_{\sss{i}}} J_{\sss{i}} \dot{\bm{\omega}}_{\sss{i}}
				+ 
				D\bm{\sigma}( \Omi{i}) \OmiDot{i} 
			\right)
			+
		}
		\\
		&
		\hspace{-0.6cm}
		\Scale[0.80]{
			J_{\sss{i}}^{\sss{-1}}
			\Ri[i]\tp 
			\sk{\nmbi[i]}
			\sum\limits_{\sss{j \in \mathcal{N}_{\sss{i}}}}
			\left(
				f'(\cdot) \Idmat
				-
				f''(\cdot) \nmbi[j] \nmbi[i]\tp
			\right)		
			\sk{\nmbi[j]}
			(\Ri[i]\Oi[i] - \Ri[j] \Oi[j])
		}.
		\label{eq:Omegai2Dot}
\end{align}
It follows from~\eqref{eq:omegaiDotUniformContinuity}, \eqref{eq:WDot} and~\eqref{eq:Omegai2Dot}  that if $\sup_{\sss{ t \ge 0}} \| \emb(\nmb(t)) \| < \infty$ and $\sup_{\sss{ t \ge 0}}  \norm{\Oi[i](t)} < \infty$, then $\sup_{\sss{ t \ge 0}} |\dot{W}(\bm{\omega}(t))| < \infty$ and $\sup_{\sss{ t \ge 0}} \|\ddot{\bm{\omega}}_{\sss{i}}(t)\| < \infty$;
this in turn implies that $W(\bm{\omega}(\cdot))$ and $\OiDot[i](\cdot)$ are uniformly continuous, which plays a role in proving that $\lim_{\sss{t \rightarrow \infty}}W(\bm{\omega}(t)) = 0$ and that $\lim_{\sss{t \rightarrow \infty}}\OiDot[i](t) = 0$.

\begin{prop}	
	\label{prop:ConvergenceToNullSpace}
	\normalfont
	Consider the vector field~\eqref{eq:StateVectorField}, the control law~\eqref{eq:DistributedControlLawVectorial}, and a trajectory $\xmb(\cdot)$ of $\dot{\xmb}(t) = \fmb_{\sss{\xmb}}(t,\xmb(t),\Tmb^{\sss{cl}}(t,\xmb(t)))$.
	If $\xmb(0) \in \Omega_{\sss{\xmb}}^{\sss{0}} = \{ \xmb \in \Omega_{\sss{\nmb}}^{\sss{D}} \times \Rn[3N] : V(\xmb) < d^{\sss{\min}} \} $, then, along the trajectory $\xmb(\cdot)$, $\lim_{\sss{t \rightarrow \infty}} (B \otimes \Idmat) \emb(\nmb(t)) = \zvec$ and $\lim_{\sss{t \rightarrow \infty}} \bm{\omega}(t) = \zvec$.
\end{prop}

We refer to the proof of Proposition~\ref{prop:ConvergenceToNullSpaceExtended}, which provides a result for a more general control law than that in~\eqref{eq:DistributedControlLawVectorial}, and which is described in the next section.

		\subsection{Constrained Torque}
		\label{sec:ConstrainedTorque}
		A natural constraint in a physical system is to require the torque provided by agent~$i$ to be orthogonal to $\bar{\nmb}_{\sss{i}}$. 
In satellites, thrusters that provide torque along $\bar{\nmb}_{\sss{i}}$ might be unavailable; also, controlling the space orthogonal to $\bar{\nmb}_{\sss{i}}$ can be left as an additional degree of freedom, in order to accomplish some other control objectives.
However, the control laws proposed in~\eqref{eq:DistributedControlLawDynamics} require full torque actuation, in particular, \eqref{eq:DistributedControlLawDynamics} requires each agent to provide torque on the plane orthogonal to $\bar{\nmb}_{\sss{i}}$.
Indeed, since $ \nmbi[1]\tp \emb_{\sss{k}}(\nmbi[1],\cdot) = 0$, $\forall \nmbi[1] \in \Stwo , \forall k\in \mathcal{M}$ (see~\eqref{eq:EdgeError}),  it follows that, for all $i \in \mathcal{N}$, 
\begin{align}
	\nmbibody[i]\tp \Tmb_{\sss{i}}^{\sss{cl}}(\cdot,\Oi) 
	= 
	\nmbibody[i]\tp \bm{\sigma}(\Oi[i]),
\end{align}
which is not necessarily $0$. 
In short, previously, we provided control laws $\Tmbi^{\sss{cl}}: \Omega_{\nmb_{i}} \times \Rn[3] \mapsto \Rn[3]$ which require full torque by each agent $i \in \mathcal{N}$, and in this section we provide constrained control laws $\bar{\Tmb}_{\sss{i}}^{\sss{cl}}: \Omega_{\nmb_{i}} \times \Rn[3] \mapsto \{ \zmb \in \Rn[3] : \zmb\tp\nmbibody[i] = 0 \}$, i.e., control laws which do not require torque along $\nmbibody[i]$.
Let us anticipate a future result by announcing that the constrained control law can only be used by agent $i \in \mathcal{N}$ when the unit vector to be synchronized by agent $i \in \mathcal{N}$, namely $\nmbibody[i]$, is a principal axis of that agent (i.e., when $\nmbibody[i]$ is an eigenvector of $J_{\sss{i}}$).
Consider then $\bar{\Tmb}_{\sss{i}}^{\sss{cl}}: (\bm{\nu}_{\sss{i}},\Omi{i}) = ((\Rmati{i}\tp\nmbi[i_{\sss{1}}],\cdots,\Rmati{i}\tp\nmbi[i_{|\mathcal{N}_{\sss{i}}|}]),\Omi{i}) \in  \Omega_{\nmb_{i}} \times \Rn[3] \mapsto \{ \zmb \in \Rn[3] : \zmb\tp\nmbibody[i] = 0 \}$ defined as
\begin{align}
	&
	\hspace{-0.5cm}
	\bar{\Tmb}_{\sss{i}}^{\sss{cl}}(\bm{\nu}_{\sss{i}},\Oi)
	=
	\OP{\nmbibody[i]}
	\Tmbi[i]^{\sss{cl}}(\bm{\nu}_{\sss{i}},\Oi)
	\\
	&
	\hspace{-0.5cm}
	\overset{\sss{\eqref{eq:DistributedControlLawDynamics}}}{=}
	- 
	\bm{\sigma}( \OP{\nmbibody[i]}\Omi{i})
	-
	\sum\nolimits_{\sss{l = 1}}^{\sss{l = |\mathcal{N}_{\sss{i}}|}} 
	\emb_{\sss{\kappa(i,i_{\sss{l}})}}
	\left(
		\nmbibody,
		\Rmati{i}\tp\nmbi[i_{\sss{l}}]
	\right)
	.
	\label{eq:DistributedControlLawModified2}
\end{align}
Additionally, consider a partition of $\mathcal{N}$, i.e., $\bar{\mathcal{L}} \cup \mathcal{L} = \mathcal{N}$ with $\bar{\mathcal{L}} \cap \mathcal{L} = \emptyset$; where $\bar{\mathcal{L}}$ is a subset (possibly empty) of the agents whose unit vector to synchronize is an eigenvector of their moment of inertia, i.e., $\bar{\mathcal{L}} \subseteq \{i \in \mathcal{N}: \exists \lambda_i \, \text{s.t.}\, J_i \bar{\nmb}_{\sss{i}} = \lambda_i \bar{\nmb}_{\sss{i}} \}$.
Then we propose the complete control law $\bar{\Tmb}^{\sss{cl}}: (t,\xmb) \in \Rn[]_{\sss{\ge 0}} \times (\Omega_{\sss{\nmb}}^{\sss{\emb}} \times \Rn[3N] \mapsto \Rn[3 N]$ defined as
\begin{align}
	\hspace{-0.5cm}
	\begin{cases}
		\Scale[0.75]{
			(\emb_{\sss{i}} \otimes \onesvec_{\sss{3}})\tp \bar{\Tmb}^{\sss{cl}}(t,\xmb)
			=
			\bar{\Tmb}_{\sss{i}}^{\sss{cl}}((\Rmati{i}\tp(t)\nmbi[i_{\sss{1}}],\cdots,\Rmati{i}\tp\nmbi[i_{|\mathcal{N}_{\sss{i}}|}]),\Oi[i])
			\,\,
			\forall  i \in \bar{\mathcal{L}}
		},
		\\
		\Scale[0.75]{
			(\emb_{\sss{i}} \otimes \onesvec_{\sss{3}})\tp \bar{\Tmb}^{\sss{cl}}(t,\xmb)
			=
			\Tmbi[i]^{\sss{cl}}((\Rmati{i}\tp(t)\nmbi[i_{\sss{1}}],\cdots,\Rmati{i}\tp\nmbi[i_{|\mathcal{N}_{\sss{i}}|}]),\Oi[i])
			\,\,
			\forall i \in \mathcal{L}
		},		
	\end{cases}
	\label{eq:DistributedControlLawVectorialModified}
\end{align}
i.e., for agents whose unit vector to synchronize is a principal axis, either control law~\eqref{eq:DistributedControlLawDynamics} or~\eqref{eq:DistributedControlLawModified2} is chosen, and, for all other agents, control law~\eqref{eq:DistributedControlLawDynamics} is chosen.
As such, agents whose unit vector to synchronize is a principal axis have an option between using full torque control or constrained torque control. 
The disadvantage with the control law in~\eqref{eq:DistributedControlLawModified2} is that, along a trajectory of the closed-loop system, and for all $i \in \bar{\mathcal{L}}$, $\lim_{\sss{t \rightarrow \infty }}\nmbibody[i]\tp \Oi[i](t) $ is not guaranteed to exist and be $0$;
i.e., an agent that opts for~\eqref{eq:DistributedControlLawModified2} can asymptotically spin, with non-zero angular velocity, around $\nmbibody[i]$ (nonetheless, we can guarantee that $\sup_{\sss{t \ge 0 }} \|\Oi[i](t)\| < \infty \Rightarrow \sup_{\sss{t \ge 0 }} |\nmbibody[i]\tp \Oi[i](t)| < \infty$, i.e., if an agent applies~\eqref{eq:DistributedControlLawModified2} it never spins infinitely fast around its principal axis $\nmbibody[i]$).

Since $\|\OP{\nmbibody[i]} \Oi\| \le \| \Oi\|, \forall \nmbibody[i] \in \Stwo, \Oi \in \Rn[3]$, it follows that $\|\fmb_{\sss{\bm{\omega}_{\sss{i}}}}(\Oi[i],\bar{\Tmb}_{\sss{i}}^{\sss{cl}}(\cdot) )\| \le f_{\sss{\bm{\omega}_{\sss{i}}}}^{\sss{\infty}}(\xmb)$, with $f_{\sss{\bm{\omega}_{\sss{i}}}}^{\sss{\infty}}(\cdot)$ as defined in~\eqref{eq:omegaiDotUniformContinuity}.
Similarly to~\eqref{eq:LyapunovDerivative}, denote $\bar{W} : \Rn[3N] \mapsto \Rn[{}]_{\sss{\ge 0}}$ as
\begin{align}
		\bar{W}(\bm{\omega})
		=
		&
		-
		\frac{\partial V(\xmb)}{\partial \xmb}\tp
		\fmb_{\sss{\xmb}}(\cdot,\xmb,\bar{\Tmb}^{\sss{cl}}(t,\xmb))
		\\
		=
		& 
		\sum_{\sss{i \in \mathcal{L}}} 
		\Omi{i}\tp \bm{\sigma}(\Omi{i})
		+
		\sum_{\sss{j \in \bar{\mathcal{L}}}} 
		\Omi{j}\tp \OP{\nmbibody[j]}\bm{\sigma}(\Omi{j})
		,
		\label{eq:LyapunovDerivativeModified}
\end{align}
and it follows that, along a trajectory $\xmb(\cdot)$ of $\dot{\xmb}(t) = \fmb_{\sss{\xmb}}(t,\xmb(t),\bar{\Tmb}^{\sss{cl}}(t,\xmb))$, $\dot{V}(\xmb(t)) = - \bar{W}(\bm{\omega}(t)) \le 0 \, \forall t\ge 0$.
Moreover, along a trajectory $\xmb(\cdot)$, (for brevity, below, we omit the time dependency of the state)
\begin{align} 
		\hspace{-0.7cm}
		\dot{\bar{W}}(\bm{\omega})
		&
		=
		\frac{\partial \bar{W}(\xmb)}{\partial \xmb}\tp
		\fmb_{\sss{\xmb}}(t,\xmb,\bar{\Tmb}^{\sss{cl}}(t,\xmb))
		\Rightarrow
		\\
		\hspace{-0.7cm}
		\Rightarrow
		|\dot{\bar{W}}(\bm{\omega}(t))|
		&
		\le
		\sum\limits_{\sss{1 = 1}}^{\sss{i = N}} 
		\left(
			\sigma_{\sss{s}} + \sigma^{\sss{\prime}}
		\right)\tp
		\| \Oi[i](t) \| 
		f_{\sss{\bm{\omega}_{\sss{i}}}}^{\sss{\infty}}(\xmb(t)).
		\label{eq:WdotModified}
\end{align}

\begin{prop}
	\label{prop:ConvergenceToNullSpaceExtended}
	\normalfont
	Consider the vector field~\eqref{eq:StateVectorField}, the control law~\eqref{eq:DistributedControlLawVectorialModified}, and a trajectory $\xmb(\cdot)$ of $\dot{\xmb}(t) = \fmb_{\sss{\xmb}}(t,\xmb(t),\bar{\Tmb}^{\sss{cl}}(t,\xmb(t)))$.
	If $\xmb(0) \in \Omega_{\sss{\xmb}}^{\sss{0}} = \{ \xmb \in \Omega_{\sss{\nmb}}^{\sss{D}} \times \Rn[3N] : V(\xmb) < d^{\sss{\min}} \} $, then $\lim_{\sss{t \rightarrow \infty}} (B \otimes \Idmat) \emb(\nmb(t)) = \zvec$, $\lim_{\sss{t \rightarrow \infty}} \Oi[i](t) = \zvec$ for $i \in \mathcal{L}$ and $\lim_{\sss{t \rightarrow \infty}} \OP{\nmbibody[j]}\Oi[j](t) = \zvec$ for $j \in \bar{\mathcal{L}}$. 
	Moreover, $\sup_{\sss{t \ge 0 }} |\nmbibody[j]\tp\Oi[j](t)| < \infty$ for $j \in \bar{\mathcal{L}}$.
\end{prop}
\begin{proof}
	\normalfont
	For brevity, we say $\fmb: \Rn[{}]_{\sss{\ge 0}} \mapsto \Rn[n]$ is bounded, if $\sup_{\sss{t \ge 0}} \| \fmb(t)\| < \infty$; we say $\fmb: \Rn[{}]_{\sss{\ge 0}} \mapsto \Rn[n]$ converges to a constant, if $\exists \fmb^{\sss{\infty}} \in \Rn[n] : \lim_{\sss{t \rightarrow \infty}} \fmb(t) = \fmb^{\sss{\infty}}$.
	Let us provide a brief summary for the proof.
	First, we prove that, along a trajectory $\xmb(\cdot)$, $\norm{\bm{\omega}(\cdot)}$ and $\norm{\emb(\nmb(\cdot))}$ are bounded. 
	This, in turn, guarantees uniform continuity of $\dot{V}(\xmb(\cdot))$ and of $\dot{\bm{\omega}}(\cdot)$.
	And finally, since both $V(\xmb(\cdot))$ and $\bm{\omega}(\cdot)$ converge to a constant, we invoke Barbalat's lemma (see~\cite{slotine1991applied}, Lemma~4.2) to conclude that $\emb(\nmb(\cdot))$ converges to the null space of $B \otimes \Idmat$.
	Recall then the functions in~\eqref{eq:Lyapunov} and~\eqref{eq:LyapunovDerivative}.
	Since $\dot{V}(\xmb(\cdot)) \le - \bar{W}(\bm{\omega}(\cdot)) \le 0$, it follows that $ V(\xmb(\cdot)) \le V(\xmb(0)) <  d^{\sss{\min}} $.
	Therefore $D(\nmb(\cdot)) < d^{\sss{\min}}$ and $ H(\bm{\omega}(\cdot)) < d^{\sss{\min}}$.
	From $D(\nmb(\cdot)) < d^{\sss{\min}}$, it follows, with the help of Proposition~\ref{prop:RelationDistanceAndEdgeError}, that $ \emb(\nmb(\cdot))$ is bounded;
	while from $ H(\bm{\omega}(\cdot)) < d^{\sss{\min}}$, it follows that $\bm{\omega}(\cdot)$ is also bounded.
	From boundedness of $ \emb(\nmb(\cdot))$ and $\bm{\omega}(\cdot)$, it follows that $\bar{\Tmb}^{\sss{cl}}(\cdot,\xmb(\cdot))$ is bounded (see~\eqref{eq:DistributedControlLawVectorialModified} and~\eqref{eq:DistributedControlLawVectorial});
	that $\|\dot{\bm{\omega}}_{\sss{i}}(\cdot) \| \le  f_{\sss{\bm{\omega}_{\sss{i}}}}^{\sss{\infty}}(\xmb(\cdot))$ is bounded (see~\eqref{eq:omegaiDotUniformContinuity});
	that $|\ddot{V}(\xmb(\cdot))| = |\dot{W}(\bm{\omega}(\cdot))|$ is bounded (see~\eqref{eq:WdotModified}); 
	and, finally, that $\ddot{\bm{\omega}}_{\sss{i}}(\cdot) $ is bounded (see~\eqref{eq:Omegai2Dot}).
	The previous conclusions imply that $\dot{V}(\xmb(\cdot))$ and that $\dot{\bm{\omega}}(\cdot)$ are both uniformly continuous.
	Since $V(\cdot) \ge 0$ and $\dot{V}(\xmb(\cdot)) \le - W(\bm{\omega}(\cdot)) \le 0$, it follows that $V(\xmb(\cdot))$ converges to a constant; by Barbalat's lemma, uniform continuity of $\dot{V}(\xmb(\cdot))$ then implies that  $\dot{V}(\xmb(\cdot)) = - W(\bm{\omega}(\cdot))$ converges to $0$.
	As such, it follows from~\eqref{eq:LyapunovDerivative}, that $\bm{\omega}_{\sss{i}}(\cdot)$ converges to $\zvec$, for all $i \in \mathcal{L}$, while $\OP{\nmbibody[j]}\bm{\omega}_{\sss{j}}(\cdot)$ converges to $\zvec$, for all $j \in \bar{\mathcal{L}}$; also, notice that
	\begin{align}
		\hspace{-0.5cm}
		\Scale[0.9]{
			\lim\limits_{\sss{t \rightarrow \infty}} 
			\OP{\bar{\nmb}_{\sss{j}}} \Omi{j}(t) 
			= 
			\zvec
			\Rightarrow 
			\lim\limits_{\sss{t \rightarrow \infty}} 
			(\Omi{j}(t) 
			-
			\bar{\nmb}_{\sss{j}} (\bar{\nmb}_{\sss{j}}\tp \Omi{j}(t))
			)
			=\zvec .
		}
		\label{eq:LimitOmega}
	\end{align}	
	Let us now study agents in $\mathcal{L}$ and $\bar{\mathcal{L}}$ separately.
	Also, for convenience, and with some abuse of notation, denote $\Tmb_{\sss{i}}^{\sss{cl}}(t) = (\emb_{\sss{i}} \otimes \onesvec_{\sss{3}})\tp \bar{\Tmb}^{\sss{cl}}(t,\xmb(t))$, for $i \in \mathcal{L}$, and $\bar{\Tmb}_{\sss{j}}^{\sss{cl}}(t) = (\emb_{\sss{j}} \otimes \onesvec_{\sss{3}})\tp \bar{\Tmb}^{\sss{cl}}(t,\xmb(t))$,  for $j \in \bar{\mathcal{L}}$.
	For $i \in \mathcal{L}$ (for which~\eqref{eq:DistributedControlLawDynamics} is the chosen control law), and again by Barbalat's lemma, convergence of $\bm{\omega}_{\sss{i}}(\cdot)$ to $\zvec$ and uniform continuity of $\dot{\bm{\omega}}_{\sss{i}}(\cdot)$ imply that  $\dot{\bm{\omega}}_{\sss{i}}(\cdot) = \fmb_{\sss{\bm{\omega}_{\sss{i}}}}(\bm{\omega}_{\sss{i}}(\cdot), \Tmb_{\sss{i}}^{\sss{cl}}(\cdot))$ converges to $\zvec$;
	since $\bm{\omega}_{\sss{i}}(\cdot)$ converges to $\zvec$, so does $\Tmb_{\sss{i}}^{\sss{cl}}(\cdot,\xmb(\cdot))$ (see~\ref{eq:OmegaVectorField}).
	Now, for $j \in \bar{\mathcal{L}}$ (for which~\eqref{eq:DistributedControlLawModified2} is the chosen control law), and once again by Barbalat's lemma, convergence of $\OP{\nmbibody[j]}\bm{\omega}_{\sss{j}}(\cdot)$ to $\zvec$ and uniform continuity of $\frac{d}{dt}(\OP{\nmbibody[j]} \bm{\omega}(t)) = \OP{\nmbibody[j]} \dot{\bm{\omega}}(t)$ implies that  $\OP{\nmbibody[j]} \dot{\bm{\omega}}(\cdot)$ converges to $\zvec$, and therefore
	\begin{align}
		\hspace{-0.5cm}
		\Scale[0.9]{
			\lim\limits_{\sss{t \rightarrow \infty}} 
			\OP{\bar{\nmb}_{\sss{j}}} \OmiDot{j}(t)
			= 
			\zvec
			\Rightarrow 
			\lim\limits_{\sss{t \rightarrow \infty}} 
			(
			\OmiDot{j}(t) 
			-
			\bar{\nmb}_{\sss{j}} (\bar{\nmb}_{\sss{j}}\tp \OmiDot{j}(t))
			)
			= \zvec,
		}
		\label{eq:LimitOmegaDot}
	\end{align}	
	Now, recall~\eqref{eq:AngularVelocityDynamics} where $ J_{\sss{j}} \OmiDot{j}(t) = - \sk{\Omi{j}(t)} J_{\sss{j}} \Omi{j}(t) + \bar{\Tmb}_{\sss{j}}^{\sss{cl}}(t)$, and, from~\eqref{eq:LimitOmega} and~\eqref{eq:LimitOmegaDot}, it follows that 
	%
	\begin{align}
		&
		\Scale[0.82]{
			\lim\limits_{\sss{t \rightarrow \infty}}
			\left(
				J_{\sss{j}} \bar{\nmb}_{\sss{j}} 
				(\bar{\nmb}_{\sss{j}}\tp \OmiDot{j}(t)) 		
				+
				\sk{\bar{\nmb}_{\sss{j}}}
				J_{\sss{j}} \bar{\nmb}_{\sss{j}}
				(\bar{\nmb}_{\sss{i}}\tp \Omi{j}(t))^2
				-
				\bar{\Tmb}_{\sss{j}}^{\sss{cl}}(t)					
			\right)
			=
			\zvec
		}
		\\
		&
		\Scale[0.9]{
			\overset{\sss{J_{\sss{j}} \bar{\nmb}_{\sss{j}}  = \lambda_{\sss{j}} \bar{\nmb}_{\sss{j}} }}{\Rightarrow}
			\lim\limits_{\sss{t \rightarrow \infty}}
			\left(
				\bar{\nmb}_{\sss{j}}
				(\bar{\nmb}_{\sss{j}}\tp \OmiDot{j}(t)) 
				-
				\bar{\Tmb}_{\sss{j}}^{\sss{cl}}(t)		
			\right)
			=
			\zvec
		}		
		\label{eq:DynamicsInvariantSet}
	\end{align}
	If we take the inner product of~\eqref{eq:DynamicsInvariantSet} with $\bar{\nmb}_{\sss{j}}$, and since $\bar{\Tmb}_{\sss{j}}^{\sss{cl}}(\cdot) \perp \nmbibody[j]$, it follows that $\lim_{\sss{t \rightarrow \infty}} (\bar{\nmb}_{\sss{j}}\tp \OmiDot{i}(t))  = 			\lim_{\sss{t \rightarrow \infty}} \bar{\nmb}_{\sss{j}}\tp \bar{\Tmb}_{\sss{j}}^{\sss{cl}}(\cdot)  = 0$. 
	As such, it follows from~\eqref{eq:DynamicsInvariantSet} that $\bar{\Tmb}_{\sss{j}}^{\sss{cl}}(\cdot)$ converges to $\zvec$ for all $j \in \bar{\mathcal{L}}$.
	Now to summarize, recall that, for all $i \in \mathcal{L}$, both $\bm{\omega}_{\sss{i}}(\cdot)$ and $\Tmb_{\sss{i}}^{\sss{cl}}(\cdot,\xmb(\cdot))$ converge to $\zvec$, which implies, from~\eqref{eq:DistributedControlLawDynamics}, that $	
	\sum\nolimits_{\sss{l = 1}}^{\sss{l = |\mathcal{N}_{\sss{i}}|}} 
		\emb_{\sss{\kappa(i,i_{\sss{l}})}}
		\left(
			\nmbibody,
			\Rmati{i}(\cdot)\tp\nmbi[i_{\sss{l}}](\cdot)
		\right)$ converges to $\zvec$.
	On the other hand, for all $j \in \bar{\mathcal{L}}$, both $\OP{\nmbibody[j]} \bm{\omega}_{\sss{j}}(\cdot)$ and $\bar{\Tmb}_{\sss{j}}^{\sss{cl}}(\cdot,\xmb(\cdot))$ converge to zero, which implies, from~\eqref{eq:DistributedControlLawModified2}, that $	
		\sum\nolimits_{\sss{l = 1}}^{\sss{l = |\mathcal{N}_{\sss{j}}|}} 
			\emb_{\sss{\kappa(j,j_{\sss{l}})}}
			\left(
				\nmbibody,
				\Rmati{j}(\cdot)\tp\nmbi[j_{\sss{l}}](\cdot)
			\right)$ converges to $\zvec$.
	All together, it implies that $(B \otimes \Idmat) \emb(\nmb(\cdot))$ converges to $\zvec$.
	Finally, $\sup_{\sss{t \ge 0 }} |\nmbibody[j]\tp\Oi[j](t)| < \infty$ since $\bm{\omega}(\cdot)$ is bounded ($H(\bm{\omega}(\cdot)) < d^{\sss{\min}}$).
	\qed
\end{proof}
Notice that $d^{\sss{\min}}$, defined in~\eqref{eq:Dmin}, is a design parameter, and therefore, the domain of attraction in Proposition~\ref{prop:ConvergenceToNullSpaceExtended} can be made larger by increasing this parameter.
More specifically, $d^{\sss{\min}}$ increases the domain of attraction in the state space related to $\bm{\omega}$, which is clearer in the next corollaries.

\ReviewAdded{
\begin{cor}
	\label{cor:cor1}
	\normalfont
	Proposition~\ref{prop:ConvergenceToNullSpace} holds if $r := \frac{H(\bm{\omega}(0))}{d^{\sss{\min}}} < 1 $ and $\frac{D(\nmb(0))}{d^{\sss{\min}}} <  1 - r$.
\end{cor}
\begin{proof}
	If the Corollary conditions hold, then $V(\xmb(0)) = D(\nmb(0)) + H(\bm{\omega}(0)) < d^{\sss{\min}} - H(\bm{\omega}(0)) + H(\bm{\omega}(0)) = d^{\sss{\min}}$.
\end{proof}
%
Corollary~\ref{cor:cor1} states that if the total kinetic energy is \emph{small} w.r.t. $d^{\sss{\min}}$, and if the total distance in the network is \emph{small} w.r.t. $1 - r$, asymptotic synchronization is guaranteed.
\begin{cor}
	\label{cor:cor2}
	\normalfont
	Proposition~\ref{prop:ConvergenceToNullSpace} holds if $r := \frac{H(\bm{\omega}(0))}{d^{\sss{\min}}} < 1$ and if
	\begin{align}
		\frac{d_{\sss{k}}(\nmbt{k}(0),\nmbh{k}(0))}{d^{\sss{\min}}}
		< 
		\frac{1 - r}{M},
		\forall k \in \mathcal{M}.
		\label{eq:DistanceBetweenNeighbors}
	\end{align}
\end{cor}
\begin{proof}
	If~\eqref{eq:DistanceBetweenNeighbors} holds, then $D(\nmb(0)) = \sum_{\sss{k \in \mathcal{L}}} d_{\sss{k}}(\nmbt{k}(0),\nmbh{k}(0)) \le M \max_{\sss{k \in \mathcal{M}}} d_{\sss{k}}(\nmbt{k}(0),\nmbh{k}(0)) < d^{\sss{\min}}(1 -r)$, in which case, Corollary~\ref{cor:cor1} holds.
\end{proof}
%
Corollary~\ref{cor:cor2} states that if the total kinetic energy is \emph{small}, and if all neighbors are \emph{close}, then asymptotic synchronization is guaranteed.
\begin{cor}
	\label{cor:cor3}
	\normalfont
	Proposition~\ref{prop:ConvergenceToNullSpace} holds if $r := \frac{H(\bm{\omega}(0))}{d^{\sss{\min}}} < 1$ and if
	\begin{align}
		\Scale[0.85]{
			\nmb(0) 
			\in 
			\mathcal{C}
			\left(
				\frac{1}{2}
				\arccos
				\left(
					1 - 
					\min\limits_{\sss{k \in \mathcal{M}}} 
					f_{\sss{k}}^{\sss{-1}}
					\left(
						d^{\sss{\min}}
						\frac{1 - r}{M}
					\right)
				\right)
			\right),
		}
		\label{eq:InitialCones}
	\end{align}
	where $d^{\sss{\min}} = \min_{\sss{k \in \mathcal{M}}} \lim_{\sss{s \rightarrow 2}} f_{\sss{k}}(s) $.
	If $f_{\sss{k}} \in \mathcal{P}^{\sss{\infty}}$ for all $k \in \mathcal{M}$, then $d^{\sss{\min}} = \infty$ and~\eqref{eq:InitialCones}  reduces to $\nmb(0) \in \mathcal{C}(\frac{\pi}{2})$.
\end{cor}
\begin{proof}
	By invoking Proposition~\ref{prop:BelongToCone1}, it follows  that~\eqref{eq:InitialCones} implies that $d_{\sss{k}}(\nmbt{k}(0),\nmbh{k}(0))\le f_{\sss{k}}(\min\limits_{\sss{l \in \mathcal{M}}} f_{\sss{l}}^{\sss{-1}}(d^{\sss{\min}} \frac{1 - r}{M}) ) \le d^{\sss{\min}} \frac{1 - r}{M} $ for all $ k \in \mathcal{M}$, and therefore \eqref{eq:DistanceBetweenNeighbors} holds.
\end{proof}
%
%
Corollary~\ref{cor:cor3} states that if the total kinetic energy is \emph{small}, and if all neighbors are initially contained in a \emph{small} cone, then synchronization is guaranteed.
Moreover, if $d^{\sss{\min}} = \infty$ and if all neighbors are initially contained in an open $\frac{\pi}{2}$-cone, then synchronization is also guaranteed.}

\begin{prop}	
	\label{prop:ConvergenceToNullSpace2}
	\normalfont
	Consider the vector field~\eqref{eq:StateVectorField}, the control law~\eqref{eq:DistributedControlLawVectorialModified}, and a trajectory $\xmb(\cdot)$ of $\dot{\xmb}(t) = \fmb_{\sss{\xmb}}(t,\xmb(t),\bar{\Tmb}^{\sss{cl}}(t,\xmb(t)))$.
	If $f_{\sss{k}} \in \mathcal{P}^{\sss{0}}$ for all $k \in \mathcal{M}$, then for all $\xmb(0) \in (\mathcal{S}^{\sss{2}})^{\sss{N}} \times \Rn[3N]$, $\lim_{\sss{t \rightarrow \infty}} (B \otimes \Idmat) \emb(\nmb(t)) = \zvec$, $\lim_{\sss{t \rightarrow \infty}} \Oi[i](t) = \zvec$ for $i \in \mathcal{L}$ and $\lim_{\sss{t \rightarrow \infty}} \OP{\nmbibody[j]}\Oi[j](t) = \zvec$ for $j \in \bar{\mathcal{L}}$;
	additionally, if $\xmb(0) \in \Omega_{\sss{\xmb}}^{\sss{0}} = \{ \xmb \in (\mathcal{S}^{\sss{2}})^{\sss{N}} \times \Rn[3N] : \exists p \in \mathcal{M}  , V(\xmb) < p d^{\sss{\min}} \} $, then no more than $p - 1$ neighbors are ever diametrically opposed, i.e., $\sup_{\sss{t \ge 0}} \, | \{ q \in \mathcal{M} : \nmbt{q}\tp(t)\nmbh{q}(t) = -1 \} | \le p -1$.
\end{prop}

\begin{proof}
	\normalfont
	Notice that if $f_{\sss{k}} \in \mathcal{P}^{\sss{0}}$ for all $k \in \mathcal{M}$ then $\Omega_{\sss{\nmb}}^{\sss{\emb}} = \Omega_{\sss{\nmb}}^{\sss{D}} = (\mathcal{S}^{\sss{2}})^{\sss{N}}$, which is a compact set.
	Since $\emb(\cdot)$ is continuous in $\Omega_{\sss{\nmb}}^{\sss{\emb}}$, it follows that $\max_{\sss{\nmb \in \Omega_{\sss{\nmb}}^{\sss{\emb}}}}\|\emb(\nmb)\|<\infty$, and, therefore, $\|\emb(\nmb(\cdot))\|$ is bounded regardless of the trajectory $\xmb(\cdot)$.
	To conclude that $\lim_{\sss{t \rightarrow \infty}} (B \otimes \Idmat) \emb(\nmb(t)) = \zvec$, $\lim_{\sss{t \rightarrow \infty}} \Oi[i](t) = \zvec$ for $i \in \mathcal{L}$ and $\lim_{\sss{t \rightarrow \infty}} \OP{\nmbibody[j]}\Oi[j](t) = \zvec$ for $j \in \bar{\mathcal{L}}$, it suffices to follow the same steps as in the proof of Proposition~\ref{prop:ConvergenceToNullSpaceExtended}.
	For the final statement in the Proposition, consider $\xmb(0) \in \Omega_{\sss{\xmb}}^{\sss{0}} = \{ \xmb \in \Omega_{\sss{\xmb}} : \exists p \in \mathcal{M}  , V(\xmb) < p d^{\sss{\min}} \} $.
	Since, along a trajectory $\xmb(\cdot)$, $D(\nmb(\cdot)) \le V(\xmb(\cdot)) \le V(\xmb(0)) < p d^{\sss{\min}}$, it suffices to invoke Proposition~\ref{prop:RelationDistanceAndEdgeError}, with $\bar{D} = V(\xmb(0))$, and the Proposition's conclusion follows. 
\end{proof}
Denote  $\fmb_{\sss{\xmb}}^{\sss{cl}}(t,\xmb) := \fmb_{\sss{\xmb}}(t,\xmb,\Tmb^{\sss{cl}}(t,\xmb))$ as the closed-loop vector field.
Note then that $\Omega_{\xmb}^{\sss{\text{eq}}} =\{\xmb \in (\mathcal{S}^{\sss{2}})^{\sss{N}} \times \Rn[3N] : \forall t \ge 0, \fmb_{\sss{\xmb}}^{\sss{cl}}(t,\xmb) = \zvec\}$ provides the set of all equilibrium points, and moreover $\{\xmb \in (\mathcal{S}^{\sss{2}})^{\sss{N}} \times \Rn[3N] : (B \otimes \Idmat) \emb(\nmb) = \zvec, \Oi[i] = \zvec \text{ for } i \in \mathcal{L},  \OP{\nmbibody[j]}\Oi[j] = \zvec \text{ for } j \in \bar{\mathcal{L}}\} \subseteq \Omega_{\xmb}^{\sss{\text{eq}}}$.
As such, Propositions~\ref{prop:ConvergenceToNullSpaceExtended} and~\ref{prop:ConvergenceToNullSpace2} imply that, under the respective Propositions' conditions, a trajectory $\xmb(\cdot)$ converges to the set of equilibrium points.
Note also that $[(\onesvec_{\sss{N}} \otimes \nmb^{\sss{\star}})\tp \, \cdot]\tp \in \Omega_{\xmb}^{\sss{\text{eq}}}$ for all $\nmb^{\sss{\star}} \in \Stwo$, i.e., all configurations where all agents are synchronized are equilibrium configurations (agents are synchronized and not moving, or agents are synchronized and spinning around their principal axis).
%
Finally, notice that since $\emb(\Smb \nmb) = \emb(\nmb)$ for all $\Smb \in \{ \Idmat_{\sss{N}} \otimes \Rmat \in \Rn[3N \times 3N] : \Rmat \in \SO[3]\}$ and for all $\nmb \in \Omega_{\sss{\nmb}}^{\sss{\emb}}$, it follows that $\Omega_{\xmb}^{\sss{\text{eq}}}$ has \emph{geometric isomerism}~\cite{li2014unified};
i.e. $[\nmb\tp \cdot] \in \Omega_{\xmb}^{\sss{\text{eq}}} \Rightarrow [\Smb \nmb\tp \cdot] \in \Omega_{\xmb}^{\sss{\text{eq}}}$, which means that for every equilibrium configuration, there exits infinite other equilibria configurations which are the same up to a rotation. 
In Section~\ref{eq:TreeGraphs}, for tree graphs, we show that $\Omega_{\xmb}^{\sss{\text{eq}}}$ is composed of configurations where agents are either synchronized or diametrically opposed; 
while in Section~\ref{eq:NonTreeGraphs}, for graphs discussed in Propositions~\ref{prop:BIndependentCycles} and~\ref{prop:BAlmostIndependentCycles},  we show that $\Omega_{\xmb}^{\sss{\text{eq}}}$ is composed of configurations where agents belong to a common plane.
In light of these comments, it follows that Corollaries~\ref{cor:cor1}-\ref{cor:cor3} provide conditions for when a trajectory is guaranteed to converge to a configuration where all agents are synchronized, and not any other configuration in $\Omega_{\xmb}^{\sss{\text{eq}}}$;
in particular,  if the initial kinetic energy is \emph{too large} with respect to $d^{\sss{\min}}$, the agents may escape to other equilibria configurations other than synchronized ones.

\begin{rem}
	\normalfont
	Consider the closed-loop vector field $\fmb_{\sss{\xmb}}^{\sss{cl}}(t,\xmb) = \fmb_{\sss{\xmb}}(t,\xmb,\Tmb^{\sss{cl}}(t,\xmb))$.
	Additionally, consider the alternative state $\tilde{\xmb}(t) = \diag(\Idmat_{\sss{3N}},\Rmat(t))\xmb(t) \Leftrightarrow \xmb(t) = \diag(\Idmat_{\sss{3N}},\Rmat\tp(t))\tilde{\xmb}(t) $, which evolves according to 
	\begin{align}
		\dot{\tilde{\xmb}}(t)  
		=
		& 
		\diag(\zvec,\dot{\Rmat}(t)) \xmb(t) 
		+
		\\
		&
		\Scale[0.95]{
			\diag(\Idmat_{\sss{3N}},\Rmat(t))\fmb_{\sss{\xmb}}^{\sss{cl}}(t,\xmb(t))
			|_{\xmb(t) = \diag(\Idmat_{\sss{3N}},\Rmat\tp(t))\tilde{\xmb}(t)} 
		}
		\\
		=
		&  \diag(\Idmat_{\sss{3N}},\Rmat(t))\fmb_{\sss{\xmb}}^{\sss{cl}}(t,\diag(\Idmat_{\sss{3N}},\Rmat\tp(t))\tilde{\xmb}(t)) 
		\\		
		=:
		&
		 \fmb_{\sss{\tilde{\xmb}}}^{\sss{cl}}(t,\tilde{\xmb}(t))
	\end{align}
	The difference between $\tilde{\xmb}(\cdot)$ and $\xmb(\cdot)$ is that all quantities in $\tilde{\xmb}(\cdot)$ are expressed in the inertial reference frame, while in $\xmb(\cdot)$ the angular velocities are expressed in the body reference frame of each agent.
	If $J_{\sss{i}} = j_{\sss{i}} \Idmat$ for all $i \in \mathcal{N}$, then for any $\Smb \in \{\Rmat \otimes \Idmat_{\sss{2N}} \in \Rn[6N \times 6N] : \Rmat \in \SO[3]\}$, it can be verified that $\fmb_{\sss{\tilde{\xmb}}}^{\sss{cl}}(\cdot,\Smb \tilde{\xmb}) = \Smb\fmb_{\sss{\tilde{\xmb}}}^{\sss{cl}}(\cdot, \tilde{\xmb})$, which implies that the closed-loop dynamics of the agents are invariant to rotations;
	i.e., given two initial conditions $\tilde{\xmb}^{\sss{1}}(0) $ and $\tilde{\xmb}^{\sss{2}}(0) $ satisfying $\tilde{\xmb}^{\sss{1}}(0) =  \Smb  \tilde{\xmb}^{\sss{2}}(0)$, it follows that $ \tilde{\xmb}^{\sss{1}}(t) =  \Smb  \tilde{\xmb}^{\sss{2}}(t) $ for all $t \ge 0$.
	This is the case, because when $J_{\sss{i}} = j_{\sss{i}} \Idmat$, then~\eqref{eq:OmegaVectorField} reduces to a second order integrator.
	A similar result has been reported in~\cite{li2014unified}, where the agents are mass points (i.e., agents without moment of inertia).
	However, in our framework, where in general $J_{\sss{i}} \ne j_{\sss{i}} \Idmat$ for some $i \in \mathcal{N}$, invariance of the closed-loop dynamics to rotations does not hold due to the term  $\sk{\Oi} J_{\sss{i}} \Oi$ in~\eqref{eq:OmegaVectorField}.
\end{rem}

		\section{Tree Graphs}
		\label{eq:TreeGraphs}
		
		Let us focus first on static tree graphs. 
For these graphs, Proposition~\ref{prop:Bdefinitepositive} states that $\mathcal{N}(B \otimes \Idmat) = \{\zvec\}$. 
In this section, we quantify the domain of attraction for synchronization to be asymptotically reached, i.e., we construct a domain $\Omega_{\sss{\xmb}}^{\sss{0}}$ such that if $\xmb(0) \in \Omega_{\sss{\xmb}}^{\sss{0}}$, then all trajectories of~\eqref{eq:StateVectorField} under control law~\eqref{eq:DistributedControlLawVectorialModified} asymptotically converge to a configuration where all unit vectors are synchronized.
Later, we construct another $\Omega_{\sss{\xmb}}^{\sss{0}}$, for graphs other than tree graphs, which is smaller in size, and we quantify how much smaller it is.
\begin{thm}
	\label{thm:NoFullDomain}
	\normalfont
	Consider a static tree graph, the vector field~\eqref{eq:StateVectorField}, the control law~\eqref{eq:DistributedControlLawVectorialModified}, and a trajectory $\xmb(\cdot)$ of $\dot{\xmb}(t) = \fmb_{\sss{\xmb}}(t,\xmb(t),\bar{\Tmb}^{\sss{cl}}(t,\xmb(t)))$.
	If $\xmb(0) \in \Omega_{\sss{\xmb}}^{\sss{0}} = \{ \xmb \in \Omega_{\sss{\nmb}}^{\sss{D}} \times 
	\Rn[3N] : V(\xmb) < d^{\sss{\min}} \} $ then synchronization is asymptotically reached, i.e., $\lim_{\sss{t \rightarrow \infty}} (\nmbi[i](t) - \nmbi[j](t))= \zvec$, for all $(i,j)\in\mathcal{N}^{\sss{2}}$.
	If $f_{\sss{k}} \in \mathcal{P}^{\sss{\infty}}$ for all $k \in \mathcal{M}$, then $d^{\sss{\min}} = \infty$ and synchronization is asymptotically reached for almost all initial conditions in $(\mathcal{S}^{\sss{2}})^{\sss{N}} \times \Rn[3N]$.
\end{thm}
\begin{proof}
	\normalfont
	Under the Theorem's conditions, we can invoke Proposition~\ref{prop:ConvergenceToNullSpaceExtended} to conclude that $\lim_{\sss{t \rightarrow \infty}} (B \otimes \Idmat) \emb(\nmb(t)) = \zvec $ and Proposition~\ref{prop:RelationDistanceAndEdgeError} to conclude that two neighbors are never arbitrarily close to a configuration where they are diametrically opposed.
	%
	%
	Since  $\mathcal{N}(B \otimes \Idmat) = \{\zvec\}$, it follows that $\lim_{\sss{t \rightarrow \infty}} (B \otimes \Idmat) \emb(\nmb(t))  \Rightarrow \lim_{\sss{t \rightarrow \infty}} \emb(\nmb(t)) = \zvec$.
	As such, and since two neighbors are never arbitrarily close to a configuration where they are diametrically opposed, it follows that all unit vectors converge to one another.
 	For the second part of the Theorem, notice that, if $d^{\sss{\min}} = \infty$, then $\Omega_{\sss{\xmb}}^{\sss{0}} =  \Omega_{\sss{\nmb}}^{\sss{D}} \times  \Rn[3N]$.
 	Since $\Omega_{\sss{\nmb}}^{\sss{D}} \times  \Rn[3N] \backslash \{ (\mathcal{S}^{\sss{2}})^{\sss{N}} \times \Rn[3N]\} = \{ \xmb \in (\mathcal{S}^{\sss{2}})^{\sss{N}} \times \Rn[3N] : \forall f_{\sss{k}} \in \mathcal{P}^{\sss{\infty}} , \nmbt{k}\tp\nmbh{k} = -1\}$ is a set of zero measure in the space of all initial conditions, i.e. $(\mathcal{S}^{\sss{2}})^{\sss{N}} \times \Rn[3N]$, synchronization for almost all initial conditions is guaranteed when $d^{\sss{\min}} = \infty$.
\end{proof}
Notice that in Theorem~\ref{thm:NoFullDomain}, increasing $d^{\sss{\min}}$ enlarges the region of stability, and it yields the almost global stability result for $d^{\sss{\min}} = \infty$.
However, a similar result for other graphs, than tree graphs, is not presented in this manuscript.
\begin{exa}
	\label{ex:NoFullDomain}
	\normalfont
	Consider the distance functions $d(\nmb_1,\nmb_2) = f(1 - \nmb_1\tp\nmb_2)$ where $f(s) = a \left(\frac{\arccos(1 - s)}{\pi}\right)^{\alpha}$, with $a > 0$ and  $\alpha \ge 2$. 
	For these, $d^{\sss{\max}} = a$, $f \in \mathcal{P}^{\sss{\bar{\infty}}}$ and $f \in \mathcal{P}^{\sss{\bar{0}}}$; also $f \in \mathcal{P}_{\sss{0}}$. 
	Suppose $f_{\sss{k}}(s) = f(s)$ for all $k \in \mathcal{M}$, and for some $a$ and $\alpha$.
	Invoking Corollary~\ref{cor:cor3}, it follows that if $r := \frac{H(\bm{\omega}(0))}{a} <1$ and
	\begin{align}
		\hspace{-0.5cm}
		\nmb(0) 
		\in
		\mathcal{C}
		\left(
			\frac{\pi}{2}
			\left(
				\frac{1 - r}{M}
			\right)^{\frac{1}{\alpha}}
		\right)
		\label{eq:ConvergenceCondition}
	\end{align}
	then Theorem's~\ref{thm:NoFullDomain} conclusions follow.
	Notice that by increasing $a$ convergence for arbitrary initial values of rotational kinetic energy can be guaranteed;
	on the other hand, by increasing $\alpha$ we can increase the size of the cone where the agents need to initially be contained (up to a cone $\mathcal{C}\left(\frac{\pi}{2}\right)$).
	We emphasize that the domain of attraction in Theorem~\ref{thm:NoFullDomain} is however larger, in the sense that there are more initial conditions which do not satisfy the previous conditions, but for which synchronization is still guaranteed.
\end{exa} 

\begin{exa}
	\normalfont
	Consider the distance function $ d(\nmb_1,\nmb_2)  = f(1 - \nmb_1\tp\nmb_2)$ where $ f(s) =	2 \, a  \tan^2\left( 0.5 \arccos(1 - s) \right)$, for which $d^{\sss{\max}} = \infty$ and $f \in \mathcal{P}^{\sss{\infty}}$ (and therefore $f \in \mathcal{P}^{\sss{\bar{0}}}$). If $f_{\sss{k}}(s) = f(s)$ for all $k \in \mathcal{M}$, then convergence to a synchronized network in a tree graph is guaranteed for almost all initial conditions (see Theorem~\ref{thm:NoFullDomain}).
\end{exa}


\begin{thm}
	\label{thm:FullDomain}	
	\normalfont
	Consider a static tree graph, the vector fields~\eqref{eq:StateVectorField}, the control law~\eqref{eq:DistributedControlLawVectorialModified}, and a trajectory $\xmb(\cdot)$ of $\dot{\xmb}(t) = \fmb_{\sss{\xmb}}(t,\xmb(t),\bar{\Tmb}^{\sss{cl}}(t,\xmb(t)))$.
	If $f_{\sss{k}} \in \mathcal{P}^{\sss{0}}$ for all $k \in \mathcal{M}$, and  $\xmb(0) \in \Omega_{\sss{\xmb}}^{\sss{0}} = \{ \xmb \in \Omega_{\sss{\nmb}}^{\sss{D}} \times \Rn[3N]: \exists p \in \mathcal{M}  , V(\xmb) < p d^{\sss{\min}} \} $ then the group of unit vectors converges to a configuration where no more than $p-1$ neighboring unit vectors are diametrically opposed. 
\end{thm}
\begin{proof}
	\normalfont
	Under the Theorem's conditions, Proposition~\ref{prop:ConvergenceToNullSpace2} can be invoked.
	Additionally, since $\mathcal{N}(B \otimes \Idmat) = \{\zvec\}$ in a tree graph, it follows that $\lim_{\sss{t \rightarrow \infty}} (B \otimes \Idmat) \emb(\nmb(t))  \Rightarrow \lim_{\sss{t \rightarrow \infty}} \emb(\nmb(t)) = \zvec$, which implies that all neighbors are either synchronized or diametrically opposed.
	Since, by Proposition~\ref{prop:ConvergenceToNullSpace2}, there are at most $p-1$ diametrically opposed neighboring unit vectors, it follows that the group of unit vectors converges to a configuration where no more than $p-1$ neighboring unit vectors are diametrically opposed.
\end{proof}

\begin{exa}
	\label{ex:DistanceFunctionCos}
	\normalfont
	Consider the distance functions $d(\nmb_1,\nmb_2) = f(1 - \nmb_1\tp\nmb_2)$ where $f(s) = a\left(\frac{s}{2}\right)^{\alpha}$ for positive $a$ and $\alpha \ge 1$. For these functions, $d^{\sss{\max}}= a$, $f \in \mathcal{P}^{\sss{0}}$, and $f \in \mathcal{P}_{\sss{0}}$ if $\alpha > 1$ and $f \in \mathcal{P}_{\sss{\bar{0}}}$ if $\alpha = 1$.
	Suppose $f_{\sss{k}}(s) = f(s)$ for all $k \in \mathcal{M}$, and for some $a$ and $\alpha$.
	%
	%
	It follows that if $r := \frac{H(\bm{\omega}(0))}{a} < p $ (for some $p \in \mathcal{M}$) and
	\begin{align}
		\hspace{-0.5cm}
		\nmb(0) 
		\in
		\mathcal{C}
		\left(
			\frac{1}{2}
			\arccos
			\left(
				1 
				- 
				2
				\left(
					\frac{p - r}{M}
				\right)^{\frac{1}{\alpha}}
			\right)
		\right)
		\label{eq:ConvergenceCondition}
	\end{align}
	then Theorem's~\ref{thm:FullDomain} conditions are fulfilled and its conclusions follow.
	Similarly to example~\ref{ex:NoFullDomain}, by increasing $a$ convergence for arbitrary initial values of rotational kinetic energy can be guaranteed;
	on the other hand, by increasing $\alpha$ we can increase the size of the cone where the agents need to initially be contained (up to a cone $\mathcal{C}\left(\frac{\pi}{2}\right)$).
	The main difference between this function and that in Example~\ref{ex:NoFullDomain} is that, in this example, $f \in \mathcal{P}^{\sss{0}}$.
\end{exa}

\if\paperextended1

	\review{\label{com:c8}Under Theorem's~\ref{thm:FullDomain} conditions, the group of agents can converge to configurations where one or more pairs of neighbors are diametrically opposed. However, it does not provide any insight on whether these equilibrium configurations are stable or unstable. 
	For tree graphs, a configuration where $p$ pairs of neighbors are diametrically opposed can be shown to be unstable. 
	For proving this statement, it suffices to find initial conditions which are arbitrarily close to an equilibrium where $p$ pairs of neighbors are diametrically opposed, but for which convergence to a configuration where at most $p - 1$ pairs of neighbors are diametrically opposed is guaranteed;
	this reasoning implies that all configurations where neighbors are diametrically opposed are unstable, and therefore, only synchronized configurations are stable.
	A similar result on synchronization in~$\SO[3]$ is found in~\cite{tron2012intrinsic}, and of synchronization on the sphere is found in~\cite{olfati2006swarms}.}
	%
	%
	%
	%
	%
	Also, Theorems~\ref{thm:NoFullDomain} and~\ref{thm:FullDomain} do not provide any insight on whether the limits $\lim_{\sss{t \rightarrow \infty}} \nmbi[i](t)$ (for all $i\in\mathcal{N}$) exist;
	i.e, even if the unit vectors synchronize, it is not known whether those unit vectors converge to a  constant unit vector or a time-varying unit vector. 
	Some preliminaries results are found in~\ref{app:ConvergenceToConstant}.

\else

	Under Theorem's~\ref{thm:FullDomain} conditions, the group of agents can converge to configurations where one or more pairs of neighbors are diametrically opposed. However, it does not provide any insight on whether these equilibria configurations are stable or unstable. 

\fi

		\section{Non-Tree Graphs}
		\label{eq:NonTreeGraphs}
		
		In this section, we study the equilibria configuration induced by some more general, yet specific, network graphs.
Also, we study the local stability properties of the synchronized configuration for arbitrary graphs. 
We first give the following definition.

\begin{defn}
	\label{def:CommonPlane}
	\normalfont
	Given $n$ vectors $\xmb_{\sss{i}}\in\Rn[3]$, for $i\in\{1,\cdots,n\}$, we say that $\{\xmb_{\sss{i}}\}_{\sss{i\in\{1,\cdots,n\}}}$ \emph{belong to a common plane} if there exists a unit vector $\bm{\nu} \in \mathcal{S}^2$ such that $\OP{\bm{\nu}}\xmb_{\sss{i}} = \xmb_{\sss{i}} $ for all $i\in\{1,\cdots,n\}$.
	We say that $\{\xmb_{\sss{i}}\}_{\sss{i\in\{1,\cdots,n\}}}$ \emph{belong to a common unique plane} if there exists a single pair of unit vectors $(+\bm{\nu},-\bm{\nu})$, with $\bm{\nu} \in \mathcal{S}^2$ such that $\OP{\bm{\nu}}\xmb_{\sss{i}} = \xmb_{\sss{i}} $ for all $i\in\{1,\cdots,n\}$
\end{defn}

Let us first discuss a property that is exploited later in this section. 

\begin{prop}
	\label{prop:CommonUniquePlane}
	\normalfont
	Consider $\nmbi[1],\nmbi[2] \in \Stwo$. If $\sk{\nmbi[1]} \nmbi[2] \ne \zvec$, then $\nmbi[1]$ and $\nmbi[2]$ belong a common unique plane.
\end{prop}
\begin{proof}
	\normalfont
	Consider $\bm{\nu} = \frac{\sk{\nmbi[1]} \nmbi[2]}{\norm{\sk{\nmbi[1]} \nmbi[2]}} \in \mathcal{S}^2$, which is well defined since $\sk{\nmbi[1]} \nmbi[2] \ne \zvec$.
	It follows that  $\OP{\bm{\nu}} \nmbi[1] = \nmbi[1]$ and that $\OP{\bm{\nu}} \nmbi[2] = \nmbi[2]$, which implies that $\nmbi[1]$ and $\nmbi[2]$ belong a common plane.
	Moreover, $\nmbi[1]$ and $\nmbi[2]$ belong a common unique plane, since $\nmbi[1]$ and $\nmbi[2]$ span a two dimensional space. 
\end{proof}

\begin{prop}
	\label{prop:CommonPlane}
	\normalfont
	Consider $\nmbi[1]$, \ldots, $\nmbi[n] \in \Stwo$, with $|\nmbi[i]\tp\nmbi[i+1]| \ne 1$ for all $i = \{1,\cdots, n-1\}$. If
	\begin{align}
			\pm
			\frac{\sk{\nmbi[1]} \nmbi[2]}{\norm{\sk{\nmbi[1]} \nmbi[2]}}
			=
			\cdots
			= 
			\pm
			\frac{\sk{\nmbi[n-1]} \nmbi[n]}{\norm{\sk{\nmbi[n-1]} \nmbi[n]}},
	\end{align}
	then all unit vectors belong to a common unique plane.
\end{prop}

\begin{proof}
	\normalfont
	Consider $n = 3$. 
	Since $|\nmbi[1]\tp\nmbi[2]| \ne 1$ and $|\nmbi[2]\tp\nmbi[3]| \ne 1$, it follows that $\norm{\sk{\nmbi[1]} \nmbi[2]} \ne 0$ and $\norm{\sk{\nmbi[2]} \nmbi[3]} \ne 0$.
	Additionally, by assumption,
	\begin{align}
		\pm
		\frac{\sk{\nmbi[1]} \nmbi[2]}{\norm{\sk{\nmbi[1]} \nmbi[2]}}
		= 
		\pm
		\frac{\sk{\nmbi[2]} \nmbi[3]}{\norm{\sk{\nmbi[2]} \nmbi[3]}},
		\label{eq:UnitVectorsinPlane}
	\end{align}
	is satisfied.
	Consider then $\bm{\nu} = \frac{\sk{\nmbi[1]} \nmbi[2]}{\norm{\sk{\nmbi[1]} \nmbi[2]}} \in \mathcal{S}^2$. 
	It follows immediately that $\OP{\bm{\nu}}\nmbi[1] = \nmbi[1] $ and that $\OP{\bm{\nu}}\nmbi[2] = \nmbi[2] $. 
	It also follows that $	
		\OP{\bm{\nu}} \nmbi[3]
		=
		\left(\Idmat - \bm{\nu} \bm{\nu}\tp \right)
		\nmbi[3]
		=
		\nmbi[3] 
		-
		\bm{\nu} (\bm{\nu}\tp\nmbi[3])
		=
		\nmbi[3]$,
	where $\bm{\nu}\tp\nmbi[3] = 0$  follows from taking the inner product of~\eqref{eq:UnitVectorsinPlane} with $\nmbi[3]$. 
	Altogether, it follows that $\nmbi[1]$, $\nmbi[2]$ and $\nmbi[3]$ belong to a common unique plane (see Proposition~\ref{prop:CommonUniquePlane}). 
	For $n > 3$, it suffices to apply the previous argument $n - 2$ times.
\end{proof}

\begin{prop}
	\label{prop:CommonPlane2}
	\normalfont
	Consider $\nmbi[1]$, \ldots, $\nmbi[n] \in \Stwo$. If $\pm \embi{1}(\nmbi[1],\nmbi[2]) = \cdots = \pm \embi{n-1}(\nmbi[n-1],\nmbi[n])$ then all unit vectors belong to a common plane, which is unique if $\pm \embi{1}(\nmbi[1],\nmbi[2]) = \cdots = \pm \embi{n-1}(\nmbi[n-1],\nmbi[n]) \ne \zvec$.
\end{prop}

\begin{proof}
	If $\pm \embi{1}(\nmbi[1],\nmbi[2]) = \cdots = \pm \embi{n-1}(\nmbi[n-1],\nmbi[n]) \ne \zvec$, it suffices to invoke Proposition~\ref{prop:CommonPlane}.
	If $\pm \embi{1}(\nmbi[1],\nmbi[2]) = \cdots = \pm \embi{n-1}(\nmbi[n-1],\nmbi[n]) = \zvec$, it follows that $\pm \nmbi[1] = \cdots =\pm \nmbi[n]$, and therefore, all unit vectors belong to a common plane.
\end{proof}

\begin{thm}
	\label{thm:IndependentCycleEquilibriumConfiguration}
	\normalfont	
	Consider the vector field~\eqref{eq:StateVectorField}, the control law~\eqref{eq:DistributedControlLawVectorialModified} with  $f_{\sss{k}} \in \mathcal{P}^{\sss{0}}$ for all $k \in \mathcal{M}$, and a trajectory $\xmb(\cdot)$ of $\dot{\xmb}(t) = \fmb_{\sss{\xmb}}(t,\xmb(t),\bar{\Tmb}^{\sss{cl}}(t,\xmb(t)))$.
	If the network graph contains only independent cycles, then for each cycle, all its unit vectors converge to a common plane.
\end{thm}

\begin{proof}
	\normalfont
	By invoking Proposition~\ref{prop:ConvergenceToNullSpace2}, it follows that $\lim_{\sss{t \rightarrow \infty}}(B \otimes \Idmat )\emb(\nmb(t)) = \zvec$, i.e., $\emb(\nmb(\cdot))$ converges to the null space of $B \otimes \Idmat$.
	Now, consider a graph with only independent cycles and recall Proposition~\ref{prop:BIndependentCycles}.
	Without loss of generality, consider that there is only one independent cycle and that the first $n \ge 3$ edges form that cycle.
	From Proposition~\ref{prop:BIndependentCycles}, it follows that $\emb(\nmb) \in \mathcal{N} (B \otimes \Idmat) \Rightarrow \pm \embi{1}(\nmbt{1},\nmbh{1}) = \cdots = \pm \embi{n}(\nmbt{n},\nmbh{n})$.
	In turn, from Proposition~\ref{prop:CommonPlane2}, it follows that all unit vectors that form the cycle belong to a common plane when $(B \otimes \Idmat) \emb(\cdot) = \zvec$.
\end{proof}

\if\paperextended1
	Figure~\ref{subfig:Triangular} exemplifies the statement in Theorem~\ref{thm:IndependentCycleEquilibriumConfiguration}, where three agents form a complete graph, and thus there is one independent cycle, and where the distance function is the same for all edges. 
	In this scenario, and because the distance function is the same for all edges, the unit vectors form an equilateral triangle. 
	In Fig.~\ref{subfig:Tetrahedron}, four agents form a complete graph, which does not fit the conditions of  Theorem~\ref{thm:IndependentCycleEquilibriumConfiguration} (a complete graph with fours agents has three cycles, but they all share edges with each other, i.e., they are not independent); 
	for this scenario, we can find equilibria configurations where the unit vectors do not belong to a common plane, and, in Fig.~\ref{subfig:Tetrahedron}, an equilibrium configuration is shown where the four agents form a tetrahedron.

\if\paperextended1
	\FloatBarrier
	\begin{figure}
		\centering
		\subcapcentertrue
		\subfigure[Triangular configuration]{
			\includegraphics[width=0.22\textwidth]
			{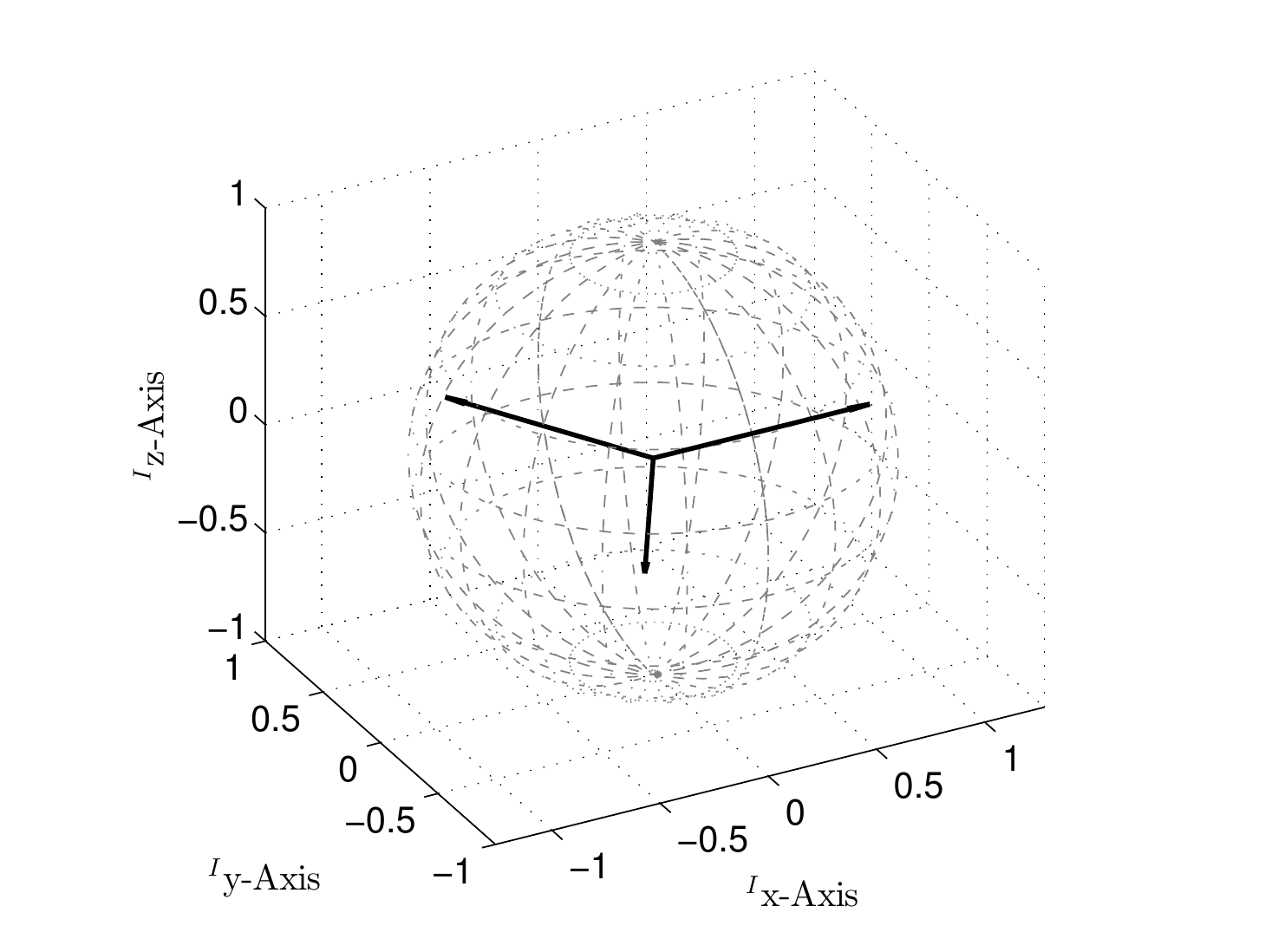}
			\label{subfig:Triangular}
		}		
		\subfigure[Tetrahedron configuration]{
			\includegraphics[width=0.22\textwidth]
			{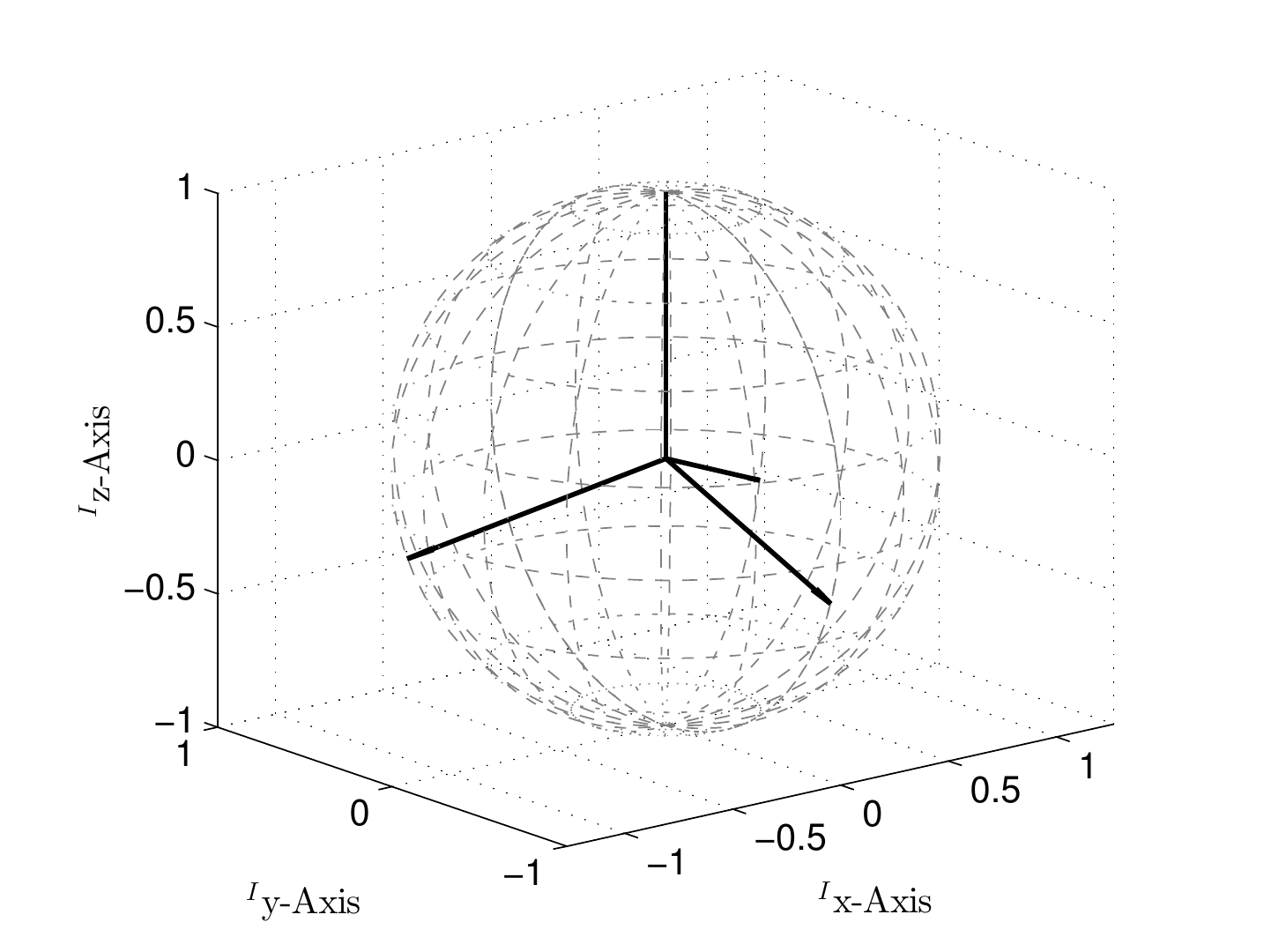}
			\label{subfig:Tetrahedron}
		}
		\caption{Equilibrium configurations in a complete graph ($\mathcal{K}_3$ in Fig.~\ref{subfig:Triangular} and $\mathcal{K}_4$ in Fig.~\ref{subfig:Tetrahedron}), with the same distance function for all edges (in Example~\ref{ex:DistanceFunctionCos} with $\alpha = 1$)}
		\label{fig:EquilibriumConfigurations}
	\end{figure}
\fi

\begin{thm}
	\label{thm:IndependentCycleEquilibriumConfiguration2}
	\normalfont
	Consider the vector field~\eqref{eq:StateVectorField}, the control law~\eqref{eq:DistributedControlLawVectorialModified} with  $f_{\sss{k}} \in \mathcal{P}^{\sss{0}}$ for all $k \in \mathcal{M}$, and a trajectory $\xmb(\cdot)$ of $\dot{\xmb}(t) = \fmb_{\sss{\xmb}}(t,\xmb(t),\bar{\Tmb}^{\sss{cl}}(t,\xmb(t)))$.
	If the network graph contains only independent cycles or cycles that share only one edge, then all unit vectors belonging to each independent cycle converge to a common plane, and all unit vectors belonging to each pair of cycles that share only one edge also converge to a common plane.
\end{thm}

\if\paperextended1
	\begin{proof}
	\normalfont
	For network graphs with only independent cycles, the proof follows the same steps as those in the proof of Theorem~\ref{thm:IndependentCycleEquilibriumConfiguration}.
	Thus, and w.l.o.g, consider graphs with only two cycles that share only one edge, and recall Proposition~\ref{prop:BAlmostIndependentCycles} (adding other cycles does not affect the following conclusions, as long as those cycles are formed by other edges other than those that form the two previously mentioned cycles). 
	W.l.o.g, assume that the first $n = q + s - 1 \ge 5 $ edges are part of the two cycles that share only one edge, where the edges $\{1,\cdots,q-1,q\}$ form the first cycle ($|\{1,\cdots,q-1,q\}|=q$), while the edges $\{q,q+1,\cdots,n\}$ form the second cycle ($|\{q,q+1,\cdots,n\}| = s$), and with the $q^{\sss{th}}$ edge as the shared edge between cycles (for brevity, denote $p = q - 1 $ and $r = q + 1$).
	Under the Theorem's conditions, and by invoking Proposition~\ref{prop:ConvergenceToNullSpace2}, it follows that $\lim_{\sss{t \rightarrow \infty}}(B \otimes \Idmat )\emb(\nmb(t)) = \zvec$, i.e., $\emb(\nmb(\cdot))$ converges to the null space of $B \otimes \Idmat$.
	From~\eqref{eq:2CyclesNullSpace}, $ \mathcal{N} (B \otimes \Idmat)$ is spanned by the vector
	\begin{align}
		\Scale[0.9]
		{
			\begin{bmatrix}
				(\bar{\onesvec}_{\sss{q-1}} \otimes \bm{\alpha})\tp
				&
				(\bm{\alpha} + \bm{\beta})\tp
				&
				(\bar{\onesvec}_{\sss{s-1}} \otimes \bm{\beta})\tp
				&
				\zvec\tp 
			\end{bmatrix}\tp
		},
	\end{align}	
	for any $\bm{\alpha},\bm{\beta} \in \Rn[3]$ (see Proposition~\ref{prop:BAlmostIndependentCycles}).

	Suppose $\nmbt{q} \ne \pm \nmbh{q}$, where $\nmbt{q}$ and $\nmbh{q}$ are the unit vectors that form edge $q$;
	then $\embi{q}(\nmbt{q},\nmbh{q}) \ne \zvec$ and therefore $\bm{\alpha} + \bm{\beta} \ne \zvec \Rightarrow \bm{\alpha} \ne \zvec \vee  \bm{\beta} \ne \zvec $; moreover, $\nmbt{q}$ and $\nmbh{q}$ form a common unique plane (see Proposition~\ref{prop:CommonUniquePlane}).
	
	If $\bm{\alpha} = \zvec$ (or $\bm{\beta} = \zvec$), it follows that $\pm \embi{1}(\nmbt{1},\nmbh{1}) = \cdots = \pm \embi{p}(\nmbt{p},\nmbh{p}) = \zvec$, and therefore the unit vectors in the first (or second) cycle belong to a common plane (Proposition~\ref{prop:CommonPlane2}). 
	Also, if $\bm{\alpha} = \zvec$ (or $\bm{\beta} = \zvec$), it follows that $\bm{\beta} \ne \zvec$ (or $\bm{\alpha} \ne \zvec$), and therefore $\pm \embi{r}(\nmbt{r},\nmbh{r}) = \cdots = \pm \embi{n}(\nmbt{n},\nmbh{n}) \ne \zvec$, and therefore all unit vectors in the second (or first) cycle belong to a common unique plane (Proposition~\ref{prop:CommonPlane2}). 
	However, since the unit vectors of the first (or second) cycle are either all synchronized or some are diametrically opposed to others, it follows that all unit vectors in both cycles belong to a common unique plane.
	
	Consider now the case where $\bm{\alpha} \ne \zvec$ and $\bm{\beta} \ne \zvec$, which implies that $\pm \embi{1}(\nmbt{1},\nmbh{1}) = \cdots = \pm \embi{p}(\nmbt{p},\nmbh{p}) \ne \zvec$ as well as $\pm \embi{r}(\nmbt{r},\nmbh{r}) = \cdots = \pm \embi{n}(\nmbt{n},\nmbh{n}) \ne \zvec$. 
	Thus, by Proposition~\ref{prop:CommonPlane2}, it follows that all unit vectors in each cycle belong to a common plane.
	Since $\nmbt{q}$ and $\nmbh{q}$ form a common unique plane, and since $\nmbt{q}$ and $\nmbh{q}$ belong \emph{simultaneously} to both cycles, it follows that all unit vectors in both cycles belong to a common unique plane. 
	
	Suppose now that $\nmbt{q} = \pm \nmbh{q}$.
	Then $\embi{q}(\nmbt{q},\nmbh{q}) = \zvec$, and it follows that $\bm{\alpha} + \bm{\beta} = \zvec \Rightarrow \bm{\alpha} = -\bm{\beta} \ne \zvec \vee \bm{\alpha} = \bm{\beta} = \zvec$. 
	If $\bm{\alpha} = \bm{\beta} = \zvec$, then $\emb(\nmb) = \zvec \Leftrightarrow \pm \nmbi[1] = \cdots = \pm \nmbi[N]$ and the Theorem's conclusions follow. 
	If $\bm{\alpha} = -\bm{\beta} \ne \zvec$, then $\pm \embi{1}(\nmbt{1},\nmbh{1}) = \cdots = \pm \embi{p }(\nmbt{p},\nmbh{p}) = \pm \embi{r }(\nmbt{r},\nmbh{r}) = \cdots =\pm \embi{n}(\nmbt{n},\nmbh{n}) \ne \zvec$, which implies that all unit vectors in both cycles belong to a common unique plane (Proposition~\ref{prop:CommonPlane2}).
	\end{proof}

\FloatBarrier
\begin{figure*}
	\centering
	\subcapcentertrue
	\subfigure[Equilibrium configuration with network graph shown in Fig.~\ref{subfig:Agents6Network}]{
		\includegraphics[width=0.25\textwidth]
		{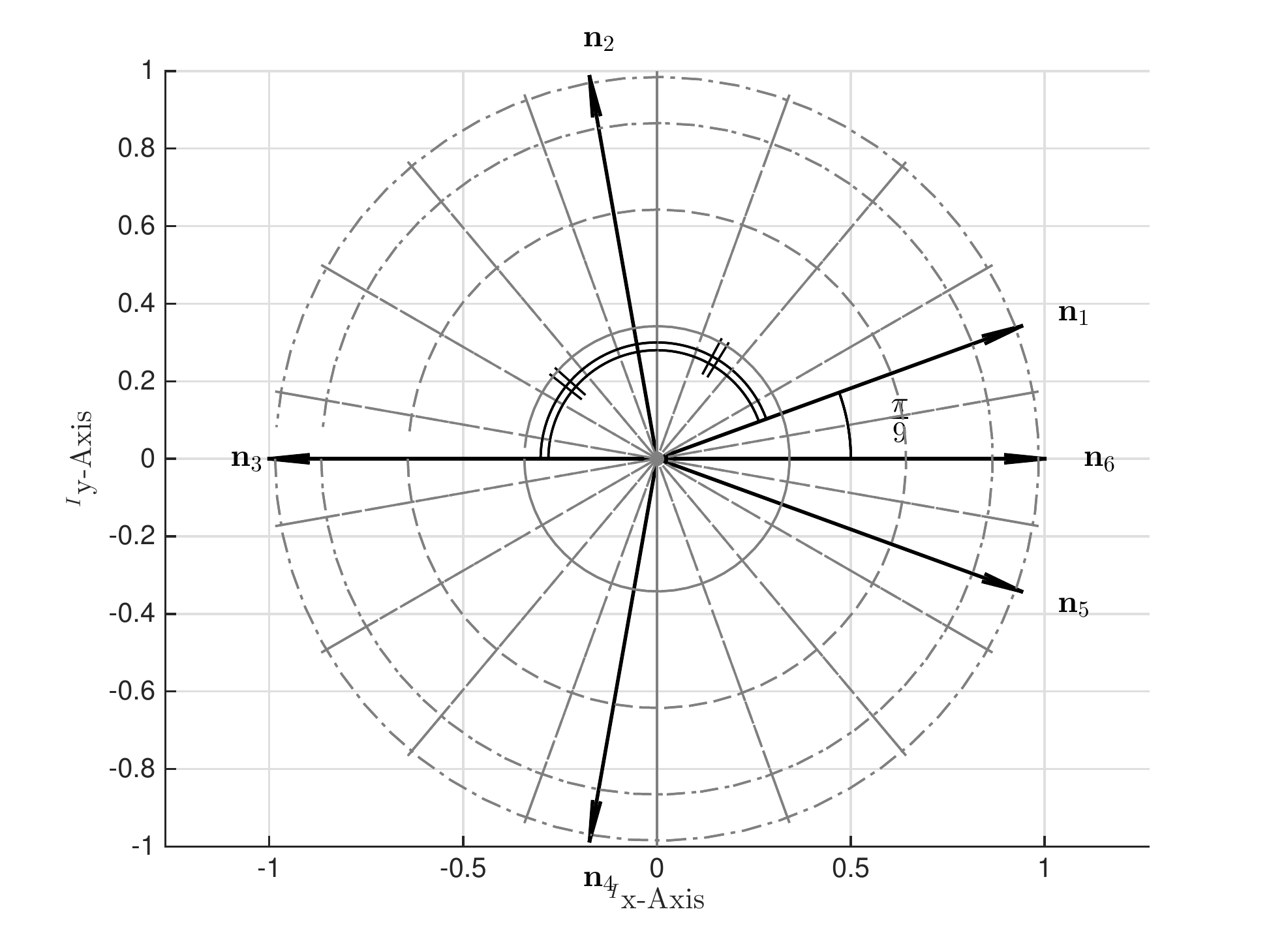}
		\label{subfig:Agents6OppositeAlternative1}
	}	
	\subfigure[Equilibrium configuration with network graph shown in Fig.~\ref{subfig:Agents6Network}]{
		\includegraphics[width=0.25\textwidth]
		{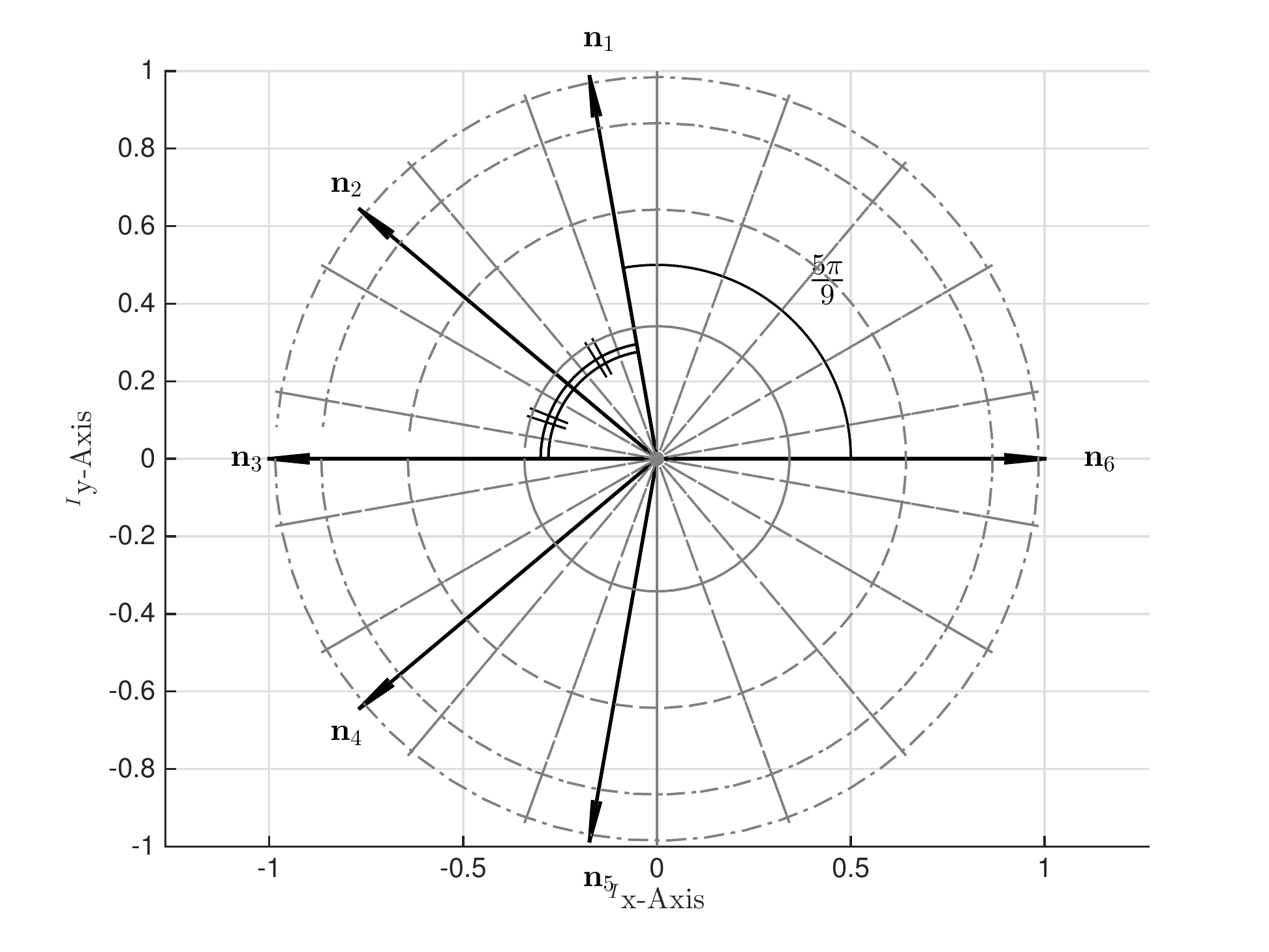}
		\label{subfig:Agents6OppositeAlternative2}
	}	
	\subfigure[Graph with two cycles that share only one edge]{
			\centering
			\begin{tikzpicture}[-,>=stealth',shorten >=1pt,auto,node distance=3cm,
			                    thick,main node/.style={circle,draw,font=\sffamily\Large\bfseries},scale=0.4,every node/.style={scale=0.35}]
			
			  \node[main node] (1) {$\nmbi[1]$};
			  \node[main node] (2) [above right of=1] {$\nmbi[2]$};
			  \node[main node] (3) [right of=2] {$\nmbi[3]$};
			  \node[main node] (4) [right of=3] {$\nmbi[4]$}; 
			  \node[main node] (5) [below right of=4] {$\nmbi[5]$};
			  \node[main node] (6) [below of=3] {$\nmbi[6]$};

			  \path[every node/.style={font=\sffamily\small}]
			    (1) edge node {} (2)
			        edge node {} (5)
			        edge node {} (6)
			    (2) edge node {} (3)
			    (3) edge node {} (4)
			    (4) edge node {} (5)
			    (6) edge node {} (5)
			    ;
			\end{tikzpicture}
			

		\label{subfig:Agents6Network}
	}
	\caption{Two equilibrium configurations for group with network graph shown in Fig.~\ref{subfig:Agents6Network}, with the same distance function for all edges (in Example~\ref{ex:DistanceFunctionCos} with $\alpha = 1$)}
	\label{fig:EquilibriumConfigurations2}
\end{figure*}
\fi

	Figure~\ref{fig:EquilibriumConfigurations2} exemplifies the statement in Theorem~\ref{thm:IndependentCycleEquilibriumConfiguration2}, with a network of six agents, with the network graph in Fig.~\ref{subfig:Agents6Network}, and where the distance function is the same for all edges (distance function in Example~\ref{ex:DistanceFunctionCos} with $\alpha = 1$).
	In this scenario, there are two cycles that share only one edge, one cycle composed of unit vectors $\{  \nmbi[1], \nmbi[2], \nmbi[3], \nmbi[4], \nmbi[5]  \}$, a second cycle composed of unit vectors $\{  \nmbi[1], \nmbi[5], \nmbi[6]  \}$, and where the shared edge is formed by $\{\nmbi[1], \nmbi[5]\}$. There are at least two equilibria configurations (apart from configurations where $\nmbi[i] = \pm \nmbi[j]$ for some $i$ and $j$), which are given in Fig.~\ref{subfig:Agents6OppositeAlternative1} and Fig.~\ref{subfig:Agents6OppositeAlternative2}, where in both cases all unit vectors belong to a common plane. 	

	Propositions~\ref{thm:IndependentCycleEquilibriumConfiguration} and~\ref{thm:IndependentCycleEquilibriumConfiguration2} focus on equilibria for some general, yet specific, network graphs.
	For arbitrary graphs, we can find equilibria configurations as exemplified in Fig.~\ref{subfig:Tetrahedron}, where four agents in a complete graph are at equilibrium when forming a tetrahedron.
	
\else
	
	The proof of Theorem~\ref{thm:IndependentCycleEquilibriumConfiguration2} relies on realizing, with the help of Proposition~\ref{prop:BAlmostIndependentCycles}, that the only way all edges belong to the null-space of $ B \otimes \Idmat$ is by belonging to a common plane, similarly to what was done in the proof of Theorem~\ref{thm:IndependentCycleEquilibriumConfiguration}. This is omitted here due to space constraints, but can be found in~\cite{TR}.

	
	
\fi

We now present a proposition, which will be useful in guaranteeing local asymptotic stability of incomplete attitude synchronization for arbitrary graphs.

\begin{prop}
	\label{prop:pi2Cone}
	\normalfont
	Consider $\nmb \in \bar{\mathcal{C}}(\alpha)$, for some $\alpha \in [0,\frac{\pi}{2})$, and assume that \emph{i)} the network graph is connected; \emph{ii)} $\emb(\nmb) \in \mathcal{N}(B \otimes \Idmat)$, with $\emb(\cdot)$ as defined in~\eqref{eq:TotalEdgeError}. This takes place if and only if $\exists \bm{\nu} \in \Stwo : \nmb  = (\onesvec_{\sss{N}} \otimes \bm{\nu})$.
\end{prop}

\begin{proof}
	\normalfont
	For the sufficiency statement, it follows that, if $\exists \bm{\nu} \in \Stwo : \nmb  = (\onesvec_{\sss{N}} \otimes \bm{\nu})$, then all unit vectors are contained in a $\frac{\pi}{2}$-cone, i.e., $\nmb \in \mathcal{C}(\frac{\pi}{2})$; and, moreover, $\emb(\onesvec_{\sss{N}} \otimes \bm{\nu}) = \zvec \subseteq \mathcal{N}(B \otimes \Idmat)$. 
	
	For the necessity statement, the proof is as follows.
	For a tree graph, $(B \otimes \Idmat) \emb(\nmb) = \zvec \Leftrightarrow \emb(\nmb) = \zvec$ follows.
	This implies that $\nmbi[i] = \pm \nmbi[j]$ for all $(i,j) \in \mathcal{N} \times \mathcal{N}_{\sss{i}}$, but since $\nmb \in \mathcal{C}(\frac{\pi}{2})$, it follows that $\nmbi[i] = \nmbi[j]$ for all $(i,j) \in \mathcal{N} \times \mathcal{N}_{\sss{i}}$.
	In a connected graph, this implies that $\nmbi[i] = \nmbi[j]$ for all $(i,j) \in \mathcal{N}\times \mathcal{N}$, and therefore $\exists \bm{\nu} \in \Stwo : \nmb  = (\onesvec_{\sss{N}} \otimes \bm{\nu})$.
	
	For an arbitrary graph, the null space of $(B \otimes \Idmat)$ may be more than $\zvec$, i.e., $(B \otimes \Idmat) \emb(\cdot) = \zvec \not\Rightarrow \emb(\cdot) = \zvec$.
	We anticipate the final result by stating that if $\nmb \in \mathcal{C}(\frac{\pi}{2})$, then $(B \otimes \Idmat) \emb(\nmb) = \zvec \Rightarrow \emb(\nmb) = \zvec$, in which case we conclude again that $\exists \bm{\nu} \in \Stwo : \nmb  = (\onesvec_{\sss{N}} \otimes \bm{\nu})$.
	Consider then an $\nmb \in (\mathcal{S}^{\sss{2}})^{\sss{N}}$, such that $(B \otimes \Idmat) \emb(\nmb) = \zvec$.
	This means that, for every $i \in \mathcal{N}$ ($B_{\sss{i\bm{:}}}$ stands for the $i^{\sss{th}}$ row of $B$),
	%
	\begin{align}
		\hspace{-0.5cm}
		\Scale[0.95]
		{
			\zvec 
			=
			(B_{\sss{i\bm{:}}}\otimes \Idmat) \emb(\nmb)
			\overset{\sss{\eqref{eq:EdgeError2}}}{=}
			\sk{\nmbi[i]} \sum\limits_{\sss{j \in \mathcal{N}_{\sss{i}}}} 
			f_{\sss{\kappa(i,j)}}^{\sss{\prime}}(1 - \nmbi[i]\tp\nmbi[j])  \nmbi[j].
		}
		\label{eq:SemiSphere1}
	\end{align}
	Since $\nmb \in \bar{\mathcal{C}}(\alpha)$, it follows that there exists a unit vector $\bm{\mu}\in \Stwo$, such that $\bm{\mu} \tp \nmbi[i] \ge \cos(\alpha) > 0$ for all $i \in \mathcal{N}$.
	Taking the inner product of~\eqref{eq:SemiSphere1} with $\sk{\nmbi[i]} \bm{\mu}$, it follows that $ \bm{\mu}\tp	\OP{\nmbi[i]} \sum_{\sss{j \in \mathcal{N}_{\sss{i}}}} f_{\sss{\kappa(i,j)}}^{\sss{\prime}}(1 - \nmbi[i]\tp\nmbi[j])  \nmbi[j] = \zvec$, which can be expanded into
	\begin{align}
		\hspace{-0.6cm}		
		\sum\nolimits_{j \in \mathcal{N}_{\sss{i}}} 
		f_{\sss{\kappa(i,j)}}^{\sss{\prime}}(1 - \nmbi[i]\tp\nmbi[j])
		\left(
			\bm{\mu}\tp \nmbi[j] - (\bm{\mu}\tp\nmbi[i]) \nmbi[i]\tp \nmbi[j]
		\right)
 		= 0.
 		\label{eq:RelationEighthSphere2}	
	\end{align}
	Now, consider the set $\mathcal{T} = \{i \in \mathcal{N}: i = \arg \max_{i \in \mathcal{N}} (1 - \bm{\mu}\tp\nmbi[i]) \} $, and choose $k \in \mathcal{T}$ (in the end, we show that, in fact, $\mathcal{T} = \mathcal{N}$).
	Notice that $0 < \cos(\alpha) \le \bm{\mu}\tp \nmbi[k]  \le  \bm{\mu}\tp \nmbi[j]$ for all $k \in \mathcal{T}$ and all $j \in \mathcal{N}$.
	As such, it follows from~\eqref{eq:RelationEighthSphere2} with $i = k$ that
	\begin{align}	
			\hspace{-0.5cm}	
			& 
			\Scale[0.95]{
				0
				\le
				\cos(\alpha)
				\sum\limits_{j \in \mathcal{N}_{\sss{k}}} 
				f_{\sss{\kappa(i,j)}}^{\sss{\prime}}(1 - \nmbi[i]\tp\nmbi[j]) 
				\left(
					 1 -  \nmbi[k]\tp \nmbi[j]
				\right)			
				\le 
			}
			\\
			\hspace{-0.5cm}	
			\Scale[0.95]{
				\le
			}
			& 
			\Scale[0.95]{
				\sum\limits_{j \in \mathcal{N}_{\sss{k}}} 
				f_{\sss{\kappa(i,j)}}^{\sss{\prime}}(1 - \nmbi[i]\tp\nmbi[j])
				\left(
					\bm{\mu}\tp \nmbi[j] - (\bm{\mu}\tp\nmbi[k]) \nmbi[k]\tp \nmbi[j]
				\right)
		 		= 0.
	 		}
 		\label{eq:RelationEighthSphere}
	\end{align}
	Notice that the lower bound (on the left side of~\eqref{eq:RelationEighthSphere})  is non-negative and zero if and only if all neighbors of agent  $k$ are synchronized with itself (note that $\lim_{\sss{s \rightarrow 2^{\sss{-}}}}f_{\sss{\kappa(i,j)}}^{\sss{\prime}}(s)$ may be $0$, but since $\nmb \in \bar{\mathcal{C}}(\alpha)$, $f_{\sss{\kappa(i,j)}}^{\sss{\prime}}(s)|_{s = 1 - \nmbi[i]\tp\nmbi[j]}$ can only vanish if $s \rightarrow 0^{\sss{+}}$). 
	As such, it follows from~\eqref{eq:RelationEighthSphere} that all neighbors of agent  $k$ are contained in $\mathcal{T}$, i.e., $\mathcal{N}_{\sss{k}} \subset \mathcal{T}$. 
	As such, the same procedure as before can be followed for the neighbors of agent  $k$, to conclude that the neighbors of the neighbors of  agent  $k $ are all synchronized. 
	In a connected graph, by applying the previous reasoning at most $N-1$ times, it follows that all unit vectors are synchronized, or, equivalently, that $\exists \bm{\nu} \in \Stwo : \nmb  = (\onesvec_{\sss{N}} \otimes \bm{\nu})$. 
\end{proof}

\review{\label{com:c9}
Proposition~\ref{prop:pi2Cone} has the following interpretation. 
Recall that $\{\xmb \in (\mathcal{S}^{\sss{2}})^{\sss{N}} \times \Rn[3N] : (B \otimes \Idmat) \emb(\nmb) = \zvec, \Oi[i] = \zvec \text{ for } i \in \mathcal{L},  \OP{\nmbibody[j]}\Oi[j] = \zvec \text{ for } j \in \bar{\mathcal{L}}\} \subseteq \Omega_{\xmb}^{\sss{\text{eq}}}$, where $\Omega_{\xmb}^{\sss{\text{eq}}}$ is the set of equilibrium points.
For example, we have seen that, for specific graphs, all equilibrium configurations are such that all unit vectors belong to a common plane (see Theorems~\ref{thm:IndependentCycleEquilibriumConfiguration} and~\ref{thm:IndependentCycleEquilibriumConfiguration2}), as illustrated in Fig.~\ref{fig:EquilibriumConfigurations2}.
However, if we can guarantee that along a trajectory $\xmb(\cdot)$ of $\dot{\xmb}(t) = \fmb_{\sss{\xmb}}(t,\xmb(t),\bar{\Tmb}^{\sss{cl}}(t,\xmb(t)))$, $\exists \alpha \in [0,\frac{\pi}{2}) : \nmb(t) \in \bar{\mathcal{C}}(\alpha) , \forall t \ge 0$, i.e., if we can  guarantee that all unit vectors remain in an closed $\alpha$-cone smaller than an open $\frac{\pi}{2}$-cone, then we can invoke Proposition~\ref{prop:pi2Cone} to conclude that $\lim_{\sss{t \rightarrow \infty}} (B \otimes \Idmat) \emb(\nmb(t)) = \zvec \Rightarrow \lim_{\sss{t \rightarrow \infty}} (\nmbi[i](t) - \nmbi[j](t)) = \zvec$;
i.e., that convergence of $\emb(\nmb(\cdot))$ to the null space of $B \otimes \Idmat$ implies synchronization of the agents.
}

This motivates us to introduce a distance $d^{\sss{\star}} > 0$, which will be useful in guaranteeing that, along a trajectory $\xmb(\cdot)$, $\exists \alpha \in [0,\frac{\pi}{2}) : \nmb(t) \in \bar{\mathcal{C}}(\alpha) , \forall t \ge 0$.
Consider then
\begin{align}
	\Scale[0.9]{
		d^{\sss{\star}}
		=
		\min\limits_{\sss{k \in \mathcal{M}}}
		f_{\sss{k}}
		\left(
			1 - \cos
			\left(
				\frac{\pi}{3} \frac{1}{N-1}
			\right)
		\right)
		<
		d^{\sss{\min}},	
	}
	\label{eq:thetaStar}
\end{align}
which satisfies $ f_{\sss{k}}^{\sss{-1}}(d^{\sss{\star}}) \le 1 - \cos \left( \frac{\pi}{3} \frac{1}{N-1} \right) $ for all $k \in \mathcal{M}$.
Notice that $d^{\sss{\star}} < d^{\sss{\min}} $, since $d^{\sss{\min}} = \min_{\sss{k \in \mathcal{M}}} \lim_{\sss{s \rightarrow 2^{\sss{-}}}} f_{\sss{k}} (s)$, $1 - \cos(\frac{\pi}{3} \frac{1}{N-1}) < 2 $ for all $N \ge 2$, and since all $f_{\sss{k}} (\cdot)$ are increasing functions in $(0,2)$.
As shown next, if $D(\nmb(0)) < \frac{d^{\sss{\star}}}{M}$ (and $H(\bm{\omega}(0) = 0$), then the network of unit vectors is forever contained in a closed $\alpha$-cone, for some $\alpha \in [0, \frac{\pi}{2})$.

\begin{thm}
	\label{thm:LocalStability}
	\normalfont
	Consider an arbitrary network graph, the vector field~\eqref{eq:StateVectorField}, the control law~\eqref{eq:DistributedControlLawVectorialModified}, and a trajectory $\xmb(\cdot)$ of $\dot{\xmb}(t) = \fmb_{\sss{\xmb}}(t,\xmb(t),\bar{\Tmb}^{\sss{cl}}(t,\xmb(t)))$
	If $\xmb(0) \in \Omega_{\sss{\xmb}}^{\sss{0}} = \{ \xmb \in \Omega_{\sss{\nmb}}^{\sss{D}} \times \Rn[3N] :  V(\xmb) < d^{\sss{\star}} \} $ then synchronization is asymptotically reached, i.e., $\lim_{\sss{t \rightarrow \infty}} (\nmbi[i](t) - \nmbi[j](t))= \zvec$, for all $(i,j)\in\mathcal{N}^{\sss{2}}$.
	Moreover, all implications of Proposition~\ref{prop:ConvergenceToNullSpaceExtended} also follow.
\end{thm}

\begin{proof}
	\normalfont
	Since  $d^{\sss{\star}} < d^{\sss{\min}} $, we can invoke Proposition~\ref{prop:ConvergenceToNullSpaceExtended}, and infer that $\lim_{\sss{t \rightarrow \infty}} (B \otimes \Idmat)\emb(t) = \zvec$ (as well as all other implications stated in the Proposition). 
	Since $\dot{V}(\xmb(\cdot)) \le 0$, it follows that $f_{\sss{k}}(1 - \nmbt{k}\tp(\cdot)\nmbh{k}(\cdot)) \le D(\nmb(\cdot)) \le V(\xmb(\cdot)) \le V(0) < d^{\sss{\star}}$, for all $k \in \mathcal{M}$.
	In turn, this implies that $\theta(\nmbt{k}(\cdot),\nmbh{k}(\cdot)) \le \arccos(1 - f_{\sss{k}}^{\sss{-1}}(d^{\sss{\star}})) < \frac{\pi}{3} \frac{1}{N-1}$, for all $k \in \mathcal{M}$. 
	Since the angular displacement between any two unit vectors $\nmbi[i]$ and $\nmbi[j]$ in a connected graph satisfies  $\theta(\nmbi[i](\cdot),\nmbi[j](\cdot)) \le (N-1) \max_{\sss{k\in\mathcal{M}}} \theta(\nmbt{k}(\cdot),\nmbh{k}(\cdot))$ (see Proposition~\ref{prop:TriangularInequality}), it follows that $\sup_{\sss{t \ge 0}}\theta(\nmbi[i](t),\nmbi[j](t)) < \frac{\pi}{3}$ between any two unit vectors $\nmbi[i]$ and $\nmbi[j]$. 
	As such, it follows from Proposition~\ref{prop:BelongToCone2} that $\nmb(\cdot) \in \bar{\mathcal{C}}(\frac{3}{2}\sup_{\sss{t \ge 0}}\theta(\nmbi[i](t),\nmbi[j](t)))$, where  $\frac{3}{2}\sup_{\sss{t \ge 0}}\theta(\nmbi[i](t),\nmbi[j](t)) < \frac{\pi}{2}$.
	Finally, we invoke Proposition~\ref{prop:pi2Cone}, which implies that $\lim_{\sss{t \rightarrow \infty}} (B \otimes \Idmat)\emb(\nmb(t)) = \zvec \Rightarrow \lim_{\sss{t \rightarrow \infty}} (\nmbi[i](t) - \nmbi[j](t))= \zvec$. 
\end{proof}

Let us provide a corollary to Theorem~\ref{thm:LocalStability}, with an easier to visualize region of attraction.
\ReviewAdded{
\begin{cor}
	\label{cor:cor4}
	\normalfont
	Theorem~\ref{thm:LocalStability} holds if $r := \frac{H(\bm{\omega}(0))}{d^{\sss{\min}}} < 1$ and if
	\begin{align}
		\Scale[0.85]{
			\nmb(0) 
			\in 
			\mathcal{C}
			\left(
				\frac{1}{2}
				\arccos
				\left(
					1 - 
					\min\limits_{\sss{k \in \mathcal{M}}} 
					f_{\sss{k}}^{\sss{-1}}
					\left(
						\frac{d^{\sss{\star}}}{M}
						\left(
							1
							-
							r
						\right)
					\right)
				\right)
			\right)
		}
		\label{eq:InitialCones3}
	\end{align}
	with $d^{\sss{\star}}$ as in~\eqref{eq:thetaStar}.
\end{cor}
For proving Corollary~\ref{cor:cor4} it suffices to check that if its conditions are satisfied, then $V(\xmb(0)) < d^{\sss{\star}}$.
\begin{rem}
	\normalfont
	Comparing Theorems~\ref{thm:NoFullDomain} and~\ref{thm:LocalStability}, it follows that the region of attraction in Theorem~\ref{thm:NoFullDomain} is larger than that in Theorem~\ref{thm:LocalStability}. 
	Loosely speaking, the region of attraction in Theorem~\ref{thm:NoFullDomain} is $\frac{d^{\sss{\min}}}{d^{\sss{\star}}} > 1$ times larger than the region of attraction in Theorem~\ref{thm:LocalStability}.
	This difference comes from the network graph topology, and in fact, a tree network graph provides stronger results.
\end{rem}

}

Theorems~\ref{thm:NoFullDomain} and~\ref{thm:LocalStability} provide asymptotic results, namely $\lim_{\sss{t \rightarrow \infty }} \emb(\nmb(t)) = \zvec$. 
Remark~\ref{rem:ExponentialConvergence} in Appendix provides some insight on exponential convergence to $0$.
			
		\section{Simulations}
		\label{sec:Simulations}
		We now present simulations that illustrate some of the results presented previously.
For the simulations, we have a group of ten agents, whose network graph is that presented in Fig.~\ref{fig:Agents10Network}. The moments of inertia were generated by adding a random symmetric matrix (between $-\Idmat$ and $\Idmat$, entry-wise) to the identity matrix. For the initial conditions, we have chosen $\bm{\omega}(0)= \zvec$ and we have randomly generated one set of 10 rotation matrices. 
In Fig.~\ref{fig:ComparisonSimulations}, $\bar{\nmb}_{\sss{i}} = [1 \, 0 \, 0]\tp$ for $i = \{1,2,3,4,5\}$ and $\bar{\nmb}_{\sss{i}} = [0 \, 1 \, 0]\tp$ for $i = \{6,7,8,9,10\}$, and since these are not necessarily principal axes, we apply the control law~\eqref{eq:DistributedControlLawVectorialModified}, with $\bar{\mathcal{L}} = \emptyset $ and $\mathcal{L} = \mathcal{N}$.
In Fig.~\ref{fig:ComparisonSimulations_Constrained}, $\bar{\nmb}_{\sss{i}}$ is the principal axis of $J_{\sss{i}}$, with largest eigenvalue, for $i = \{1,2,3,4,5\}$, and $\bar{\nmb}_{\sss{i}} = [1 \, 0 \, 0]\tp$  for $i = \{6,7,8,9,10\}$. 
Therefore, we apply the control law~\eqref{eq:DistributedControlLawVectorialModified}, with $\bar{\mathcal{L}} = \{1,2,3,4,5\} $ and $\mathcal{L} = \{6,7,8,9,10\} $.
For the edge~1, we have chosen $f_{\sss{1}}(s) = 10 \tan^2\left(0.5\arccos(1 - s)\right)$. 
For the other edges, we have chosen $f_{\sss{k}}(s) = 5 s$, for $k = \mathcal{M}\backslash \{ 1\}$.
Notice that we have chosen a distance function (for edge 1) that grows unbounded when two unit vectors are diametrically opposed. 
As such, it follows that agents 1 and 6 will never be diametrically opposed, under the condition that they are not initially diametrically opposed. 
We have also chosen $\bm{\sigma}(\xmb) = k  \frac{\sigma_x \xmb}{\sqrt{\sigma_x^2 + \xmb\tp\xmb}}$ with $k = 10$ and $\sigma_x = 1$. For this choice, we find that $\sigma^{\max} = k \sigma_x = 10$. As such, for all agents, except 1 and 6, an upper bound on the norm of their torque is given by $\sigma^{\max} + 2 \cdot 5 = 20$ (the factor $2$ relates to the fact that all agents, except 1 and 6, have two neighbors, and the factor $5$ comes from $f_{\sss{k}}(s) = 5 s \Rightarrow f_{\sss{k}}^{\sss{\prime}}(s) = 5$). 
For agents 1 and 6, no upper bound can be found (more precisely, a bound can be found, but it depends on the initial conditions). 
Given these choices, it follows from Corollary~\ref{cor:cor4} that if $\nmb(0) \in \mathcal{C}(\approx 1^{\circ})$ then synchronization is guaranteed.
We emphasize, nonetheless, that Corollary~\ref{cor:cor4} provides conservative conditions for synchronization to be achieved, and the domain of attraction is in fact larger.
We also emphasize that, for tree graphs, the domain of attraction is considerably larger: for example, if we removed the edges between agents~1 and~2, and between agents~6 and~7, we would obtain a tree graph, and Corollary~\ref{cor:cor3} would read as $\nmb(0) \in \mathcal{C}(\approx 18^{\circ})$.
Finally, we emphasize that we can increase the size of the cones in Corollaries~\ref{cor:cor3} and~ \ref{cor:cor4}, by choosing different distance functions, as exemplified in Example~\ref{ex:NoFullDomain}. 

For each Fig.~\ref{fig:ComparisonSimulations} and Fig.~\ref{fig:ComparisonSimulations_Constrained}, we provide two simulations: one where the control law is that in~\eqref{eq:DistributedControlLawVectorialModified} and another where the control law is that in~\eqref{eq:DistributedControlLawVectorialModified} but corrupted by noise; specifically, for each agent $i \in \mathcal{N}$, $\Tmb_{\sss{i}}(t) = \bar{\Tmb}_{\sss{i}}^{\sss{cl}}(t,\xmb(t)) + 0.1 \lambda_{\sss{i}} [0 \, 0 \, 1]\tp$, where $\lambda_{\sss{i}}$ corresponds to the largest eigenvalue of $J_{\sss{i}}$.

The trajectories of the unit vectors for described conditions are presented in Figs.~\ref{subfig:UnitVectorTrajectories1}--\ref{subfig:UnitVectorTrajectories2} and~\ref{subfig:UnitVectorTrajectories1_Constrained}--\ref{subfig:UnitVectorTrajectories2_Constrained}.
Notice that despite not satisfying conditions of Theorem~\ref{thm:LocalStability} (the unit vectors are not always in a $\frac{\pi}{2}$ cone), incomplete  attitude synchronization is still achieved. 
This can be verified in Figs.~\ref{subfig:ErrorAngleUnitVector1}-\ref{subfig:ErrorAngleUnitVector2} and~\ref{subfig:ErrorAngleUnitVector1_Constrained}-\ref{subfig:ErrorAngleUnitVector2_Constrained}, which present the angular error between neighboring agents. 
In Figs.~\ref{subfig:UnitVectorTrajectories1_Constrained} and~\ref{subfig:UnitVectorTrajectories2_Constrained}, the control laws are different between agents~$1$--$5$ and~$6$--$10$. 
The former perform synchronization of principal axes, by applying the constrained control law~\eqref{eq:DistributedControlLawModified2};
while the later perform synchronization of their first axes, i.e., $\bar{\nmb}_{\sss{i}} = [1 \, 0 \, 0]\tp$, by applying the control law~\eqref{eq:DistributedControlLawDynamics}.
In Figs.~\ref{subfig:ErrorAngleUnitVector2} and~\ref{subfig:ErrorAngleUnitVector2_Constrained}, for which the control laws were corrupted by noise, perfect synchronization is not asymptotically achieved. 
Instead, the unit vectors converge to a configuration where they remain close to each other (error of $\approx 5^{\circ}$ between neighbors).

\begin{figure}
	\centering
	\subcapcentertrue
	\subfigure[Trajectories w/o noise]{
		\includegraphics[width=0.2\textwidth]
		{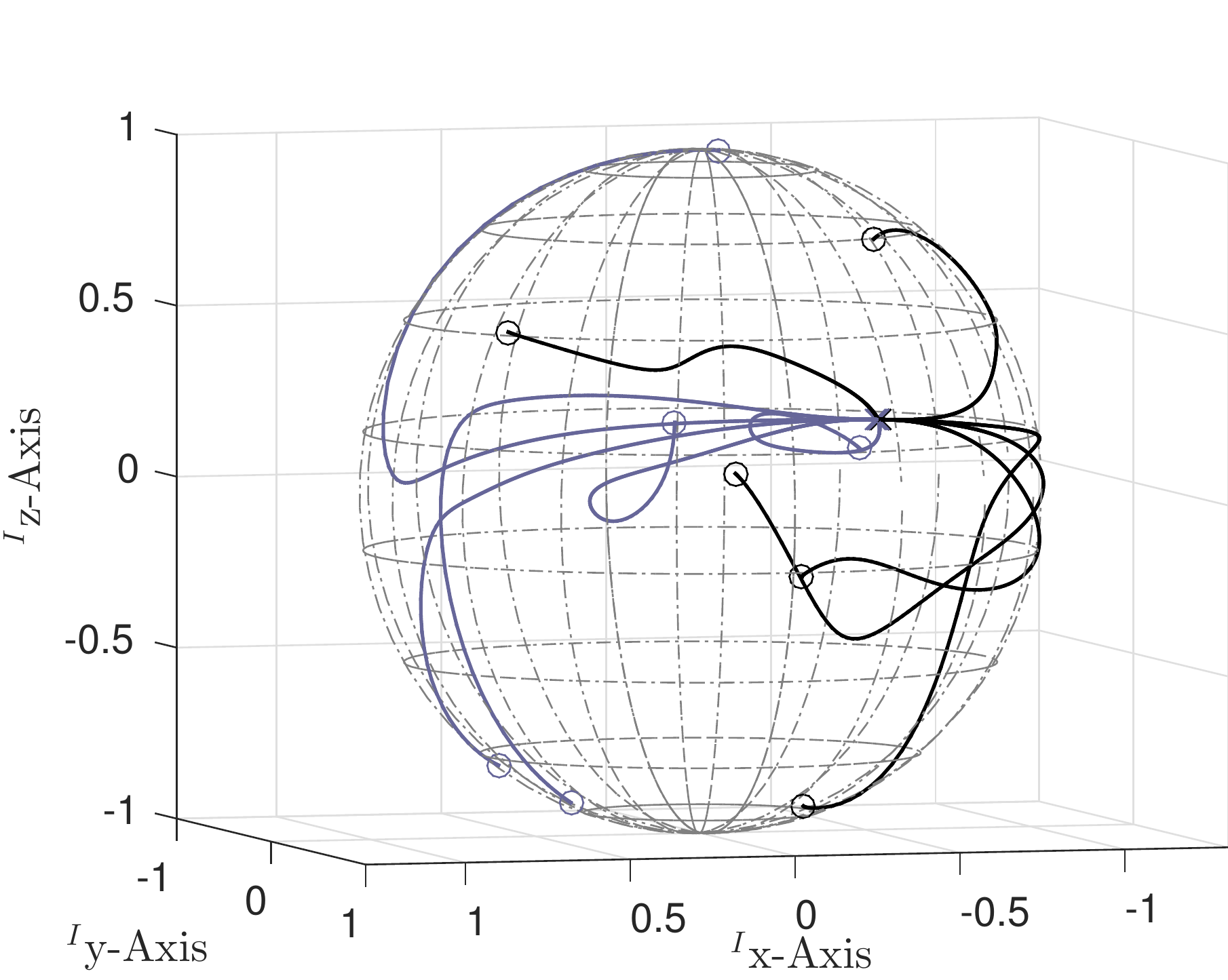}
		\label{subfig:UnitVectorTrajectories1}
	}
	\subfigure[Trajectories  with noise]{
		\includegraphics[width=0.2\textwidth]
		{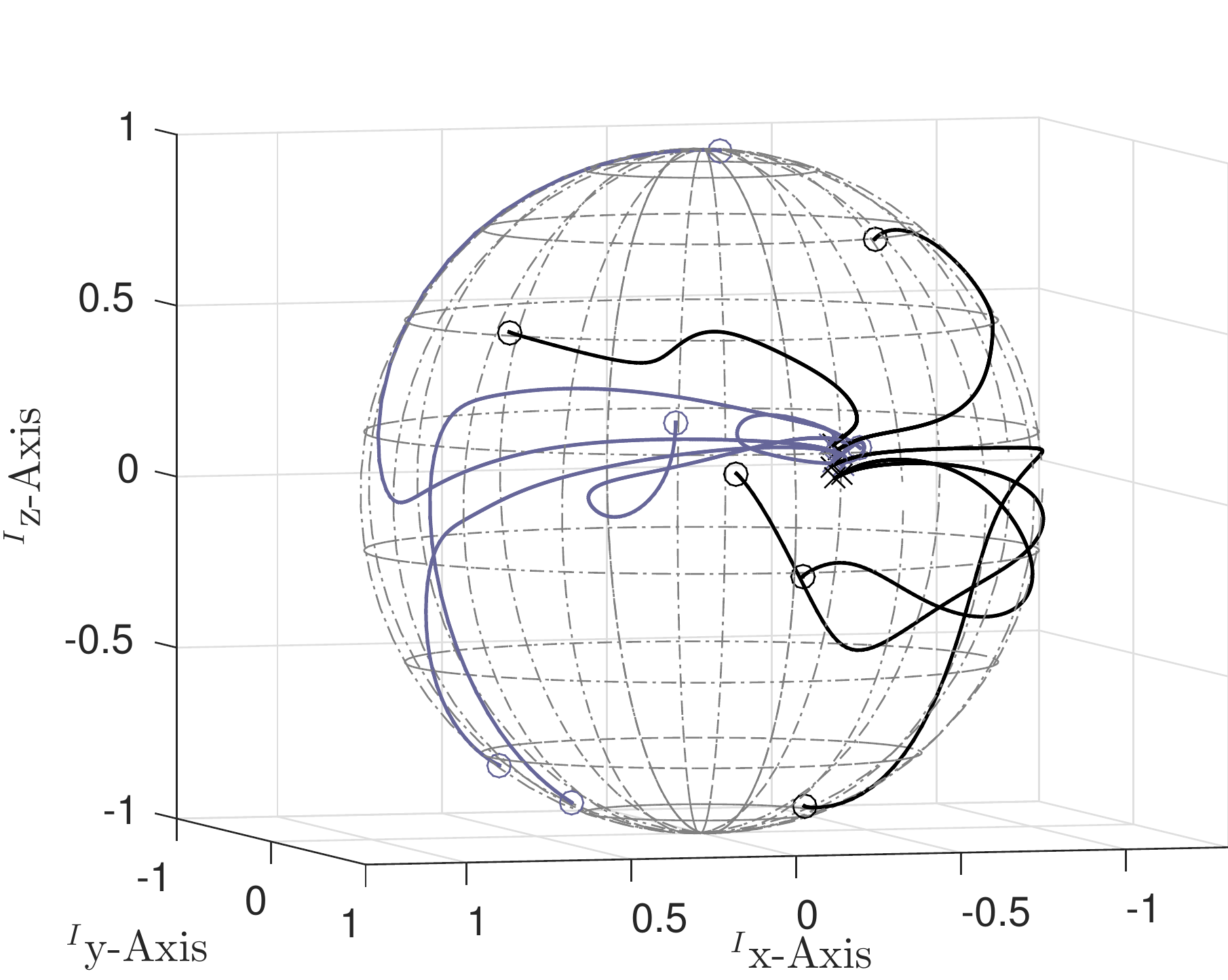}
		\label{subfig:UnitVectorTrajectories2}
	}	
	\subfigure[Error angle between neighbors without noise]{
		\includegraphics[width=0.2\textwidth]
		{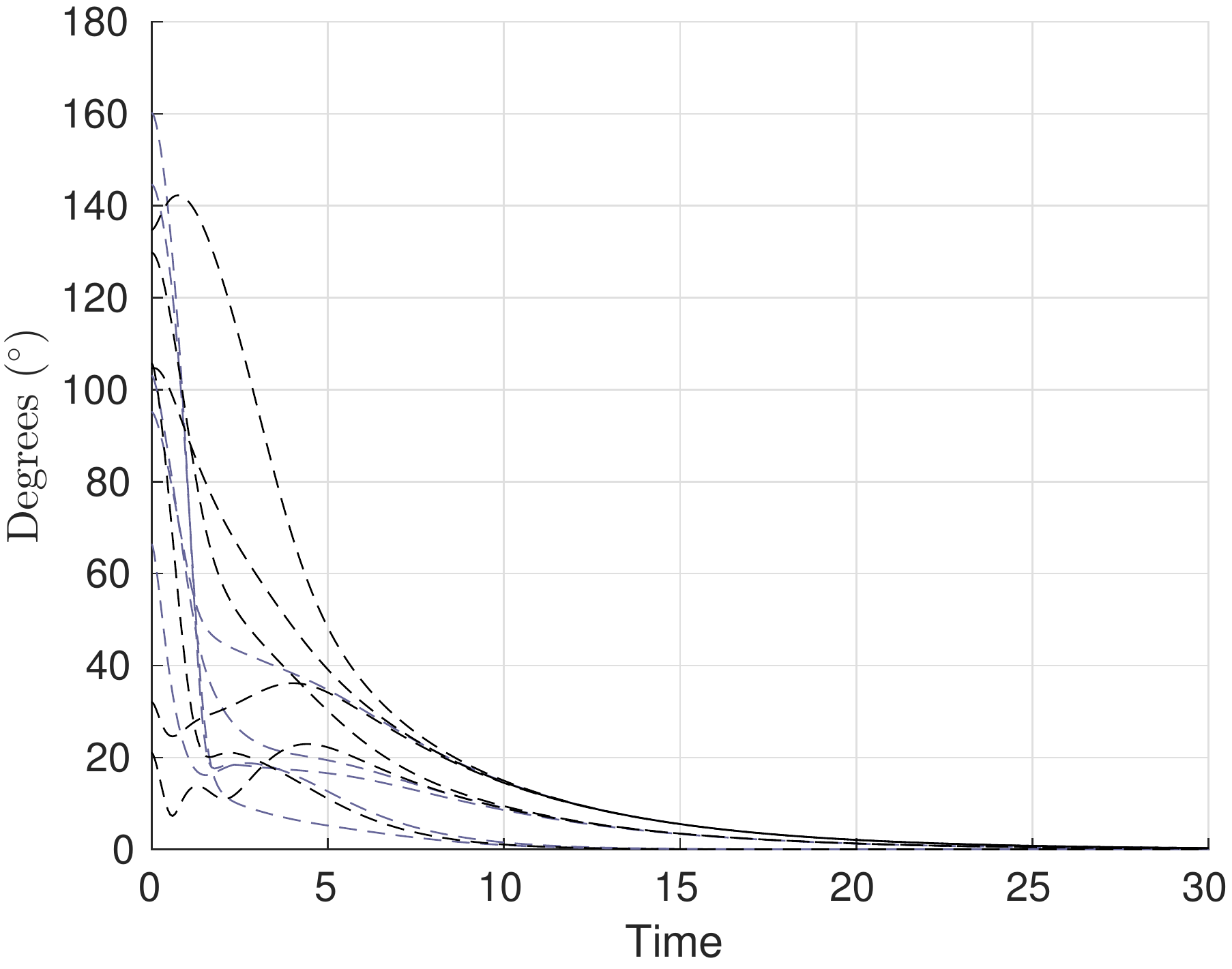}
		\label{subfig:ErrorAngleUnitVector1}
	}
	\subfigure[Error angle between neighbors with noise]{
		\includegraphics[width=0.2\textwidth]
		{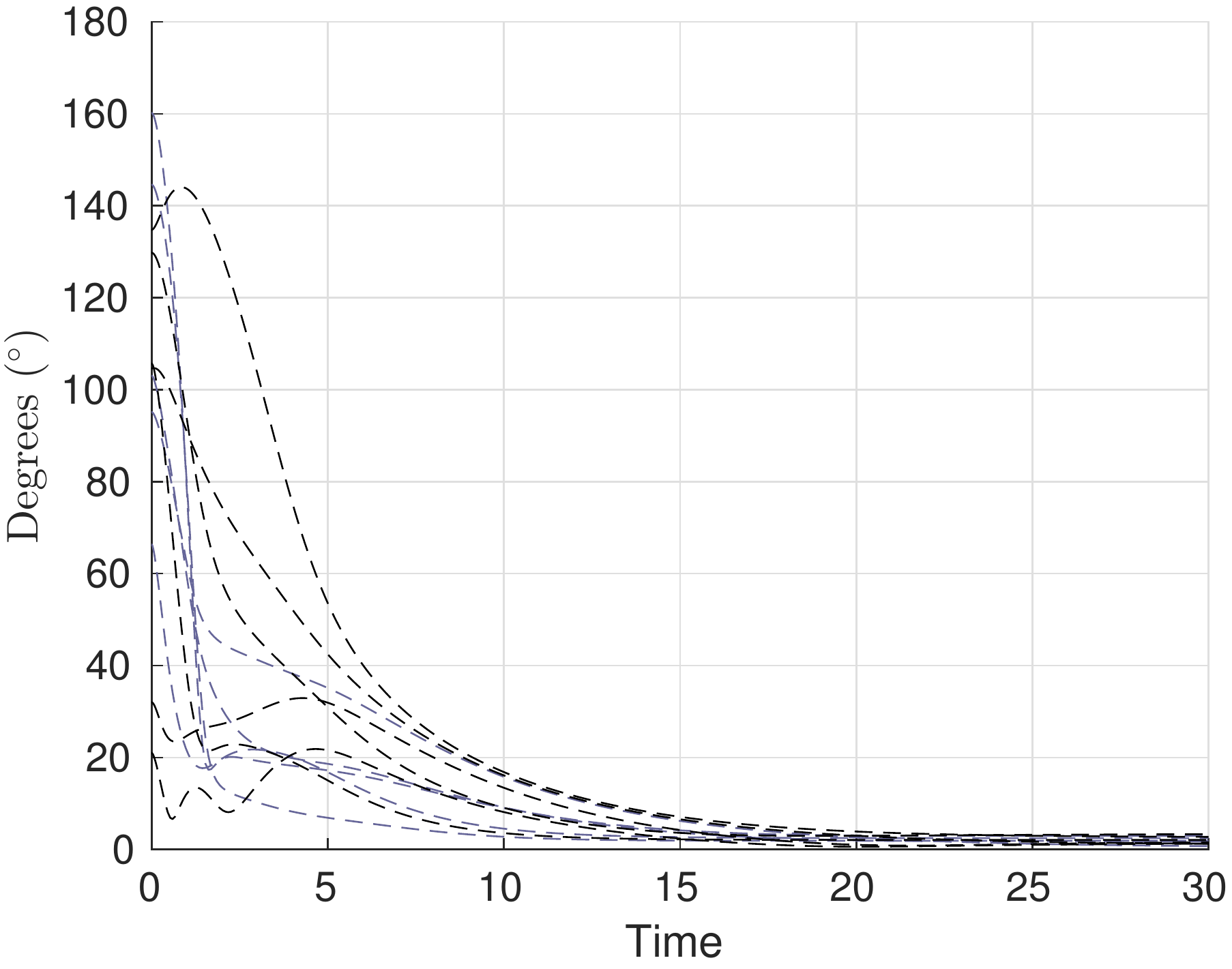}
		\label{subfig:ErrorAngleUnitVector2}
	}	
	\caption{Synchronization in network of 10 unit vectors with and without noise, where blue agents perform synchronization of their first axes ($\bar{\nmb}_{\sss{i}} = [1 \, 0 \, 0]\tp$) and black agents synchronization of their second axes ($\bar{\nmb}_{\sss{i}} = [0 \, 1 \, 0]\tp$).}
	\label{fig:ComparisonSimulations}
\end{figure}
\begin{figure}
	\centering
	\subcapcentertrue
	\subfigure[Trajectories w/o noise]{
		\includegraphics[width=0.2\textwidth]
		{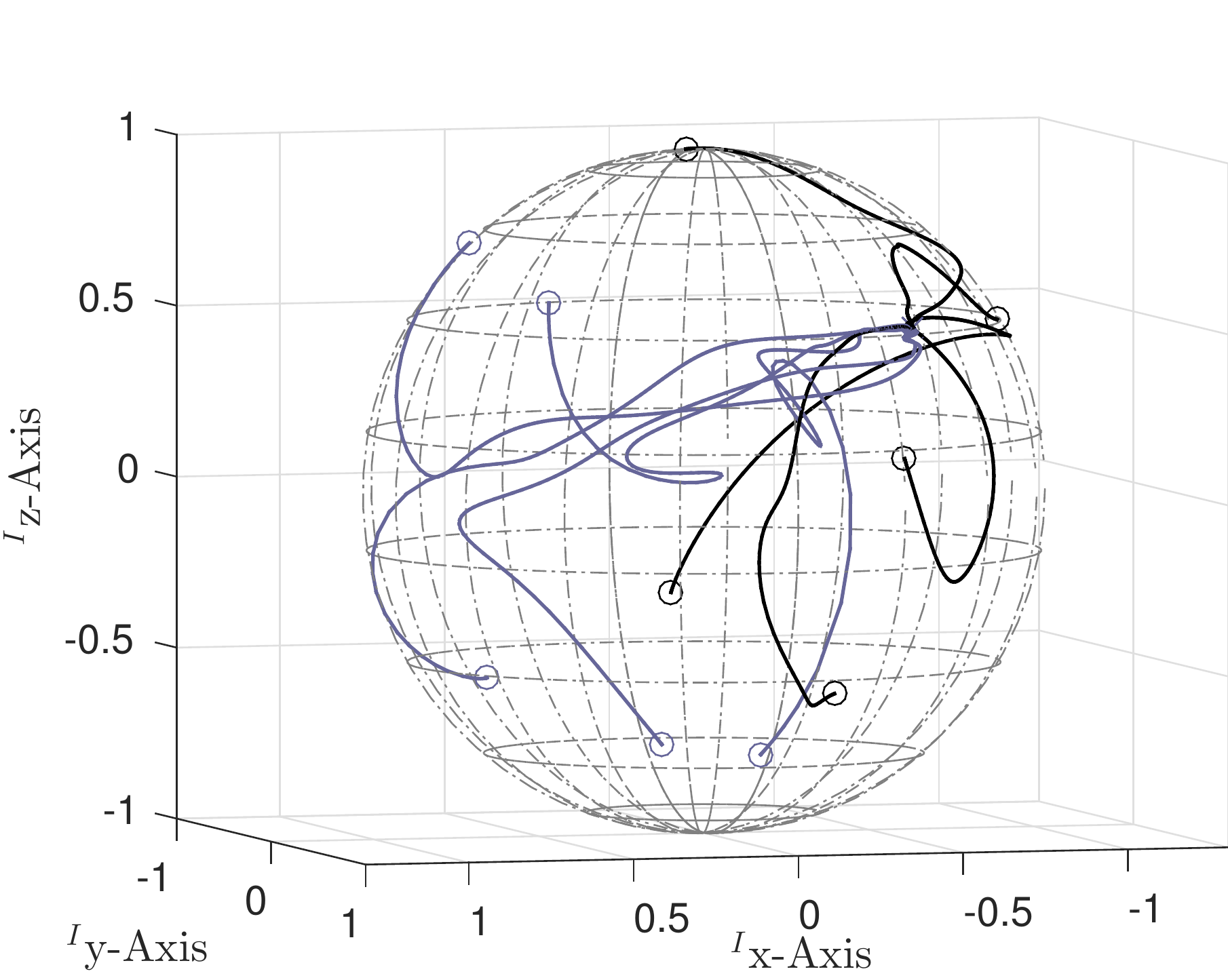}
		\label{subfig:UnitVectorTrajectories1_Constrained}
	}
	\subfigure[Trajectories with noise]{
		\includegraphics[width=0.2\textwidth]
		{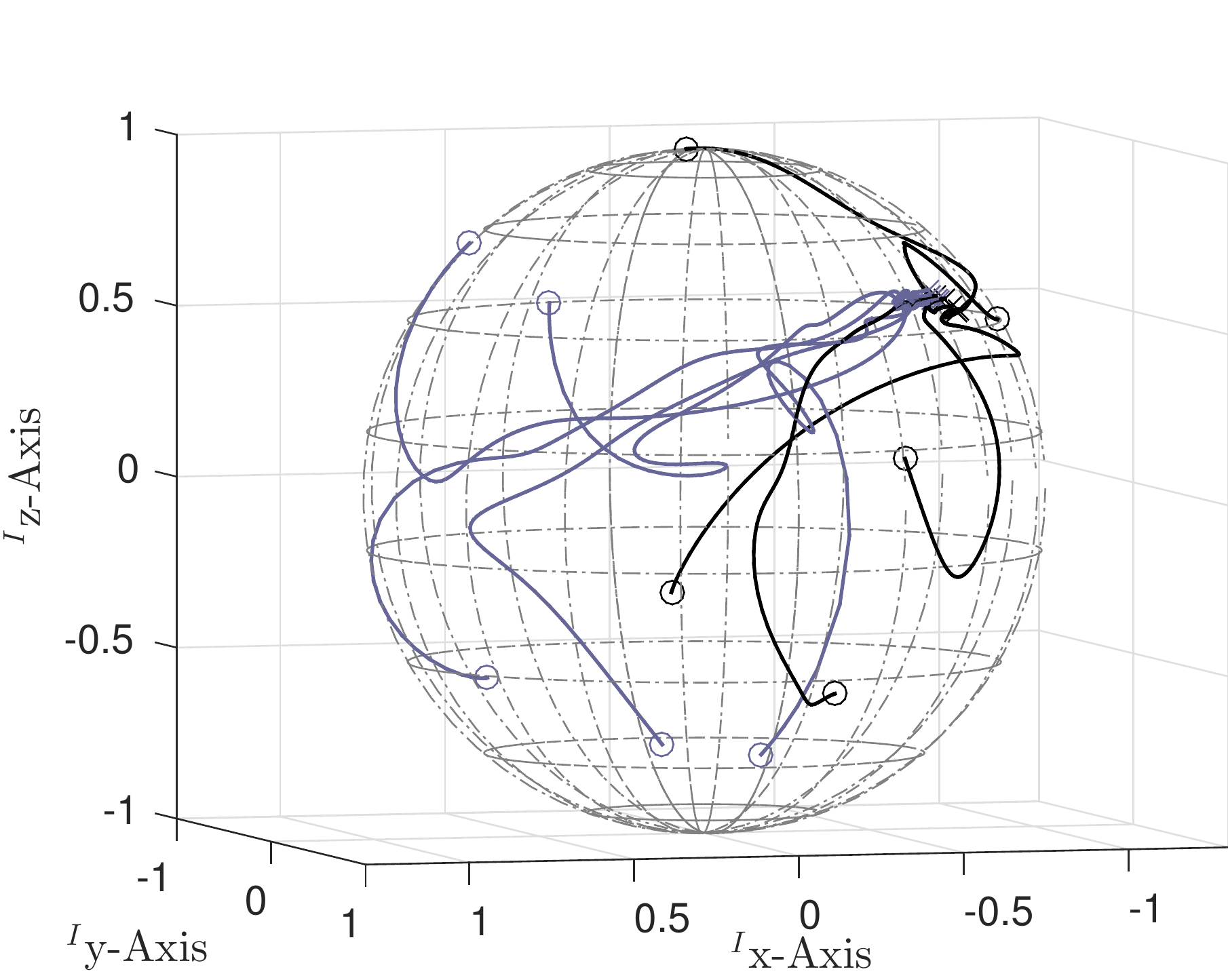}
		\label{subfig:UnitVectorTrajectories2_Constrained}
	}	
	\subfigure[Error angle between neighbors without noise]{
		\includegraphics[width=0.2\textwidth]
		{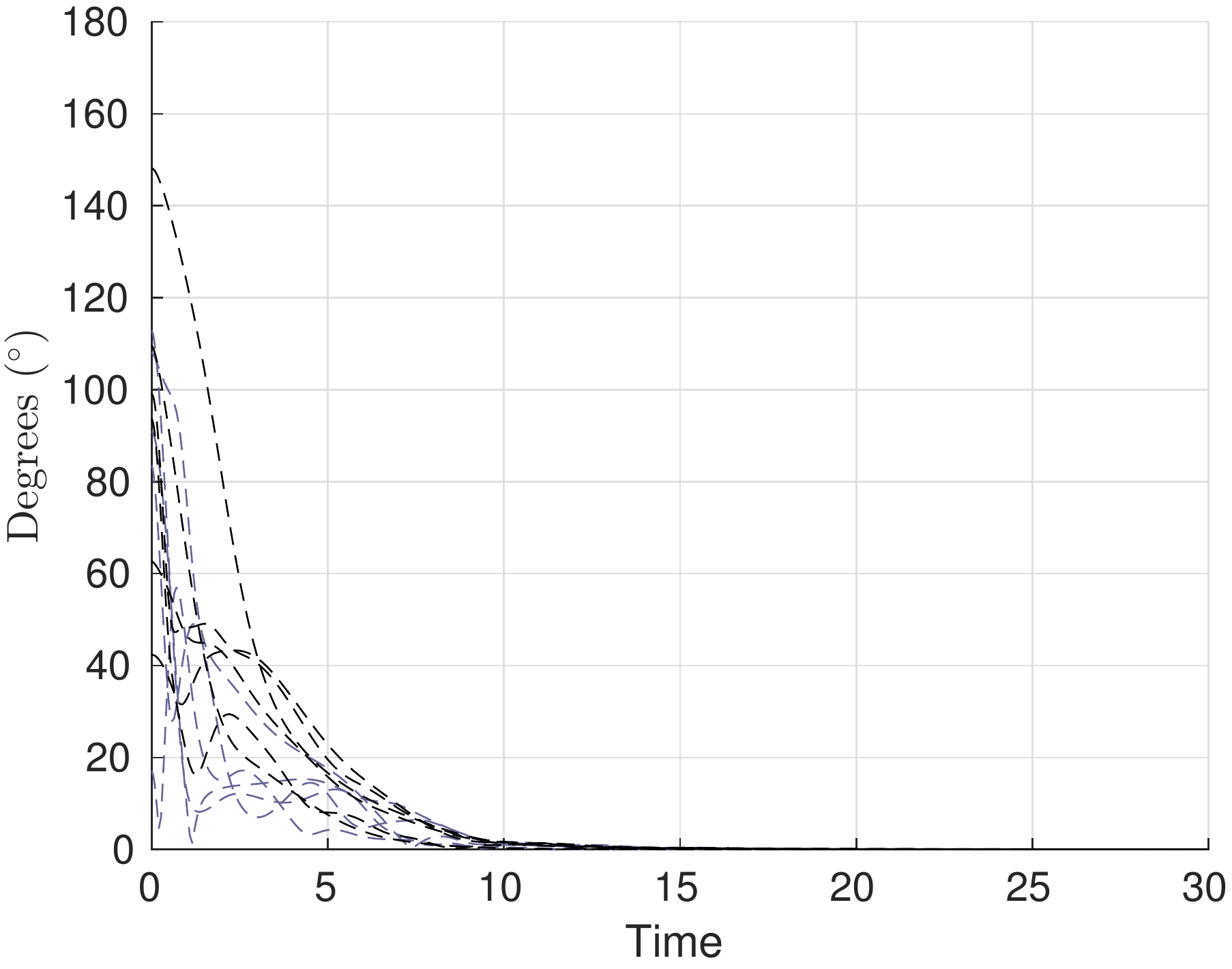}
		\label{subfig:ErrorAngleUnitVector1_Constrained}
	}
	\subfigure[Error angle between neighbors with noise]{
		\includegraphics[width=0.2\textwidth]
		{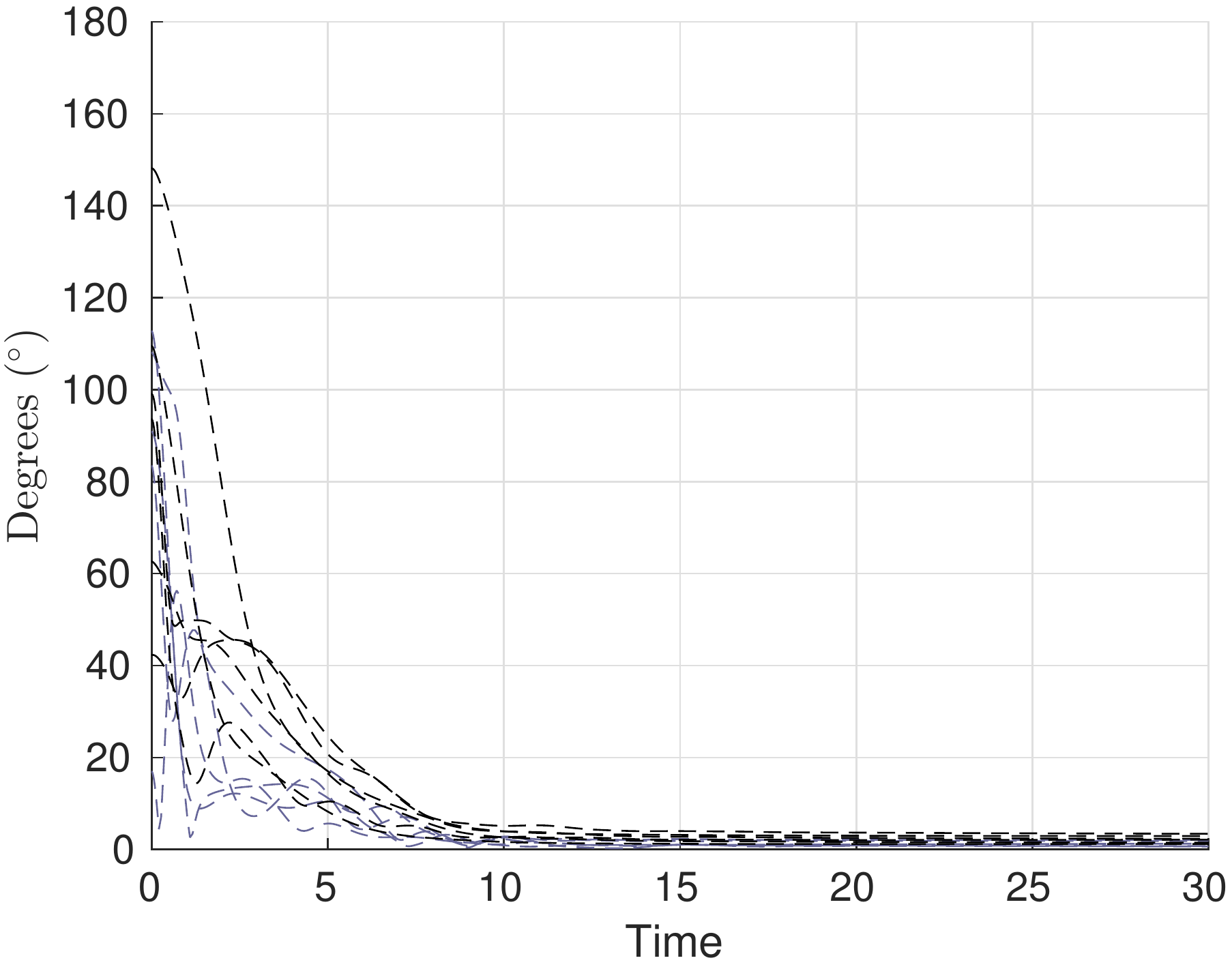}
		\label{subfig:ErrorAngleUnitVector2_Constrained}
	}	
	\caption{Synchronization in network of 10 unit vectors with and without noise, where blue agents perform synchronization of principal axes and black agents synchronization of their first axes (i.e., $\bar{\nmb}_{\sss{i}} = [1 \, 0 \, 0]\tp$).}
	\label{fig:ComparisonSimulations_Constrained}
\end{figure}



\begin{figure}
			\centering
			\begin{tikzpicture}[-,>=stealth',shorten >=1pt,auto,node distance=3cm,
			                    thick,main node/.style={circle,draw,font=\sffamily\Large\bfseries},scale=0.35,every node/.style={scale=0.35}]
			
			  \node[main node] (1) {$\nmbi[1]$};
			  \node[main node] (2) [above right of=1] {$\nmbi[2]$};
			  \node[main node] (3) [right of=2] {$\nmbi[3]$};
			  \node[main node] (4) [right of=3] {$\nmbi[4]$}; 
			  \node[main node] (5) [right of=4] {$\nmbi[5]$};
			  \node[main node] (6) [below right  of=5] {$\nmbi[6]$};
			  \node[main node] (10) [below right of=1] {$\nmbi[10]$};
			  \node[main node] (9) [right of =10] {$\nmbi[9]$};
			  \node[main node] (8) [right of=9] {$\nmbi[8]$};
			  \node[main node] (7) [right of=8] {$\nmbi[7]$};

			  \path[every node/.style={font=\sffamily\small}]
			    (1) edge node {} (2)
			        edge node {$\sss{ \bar{\kappa}(1,6) = 1}$} (6)
			        edge node {} (10)
			    (2) edge node {} (3)
			    (3) edge node {} (4)
			    (4) edge node {} (5)
			    (5) edge node {} (6)
			    (10) edge node {} (9)
			    (9) edge node {} (8)
			    (8) edge node {} (7)		
			    (7) edge node {} (6)   	    
			    ;
			\end{tikzpicture}
			
	\caption{Graph with 10 agents, where edge 1 is formed by agents 1 and 6}
	\label{fig:Agents10Network}
\end{figure}
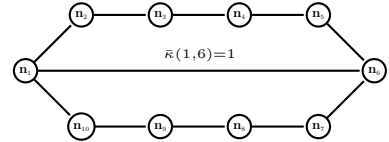		
	
		\section{Conclusions}
		In this paper, we proposed a distributed control strategy that guarantees attitude synchronization of unit vectors, representing a specific body direction of a rigid body.
The proposed control torque laws depend on distance functions in $\mathcal{S}^2$, and we provide conditions on these distance functions that guarantee that \emph{i)} a synchronized network is locally asymptotically stable in an arbitrary connected undirected graph network; \emph{ii)} a synchronized network can be achieved for almost all initial conditions in a tree graph network. 
We imposed conditions on the distance functions that guarantee that these are invariant to rotation of their arguments, which means that the proposed control laws can be implemented by each individual rigid body in the absence of a global common orientation frame, i.e., by using only local information. 
Additionally, if the direction to be synchronized is a principal axis of the rigid body, we proposed a control law that does not require full torque actuation, and, more specifically, it only requires torque in the plane orthogonal to the principal axis.
We also studied the equilibria configurations that come with certain types of graph networks.
Directions for future work include studying the stability of all equilibria configurations, apart from the synchronized configuration;
to determine whether a synchronized network converges to a constant unit vector in a fixed, though unknown, orientation frame; 
and to extend the results to complete attitude synchronization.
For complete attitude synchronization, though, it is not possible to construct a constrained control law, since it requires full torque actuation.	
	
		\bibliographystyle{IEEEtran}
		\bibliography{bibliography}
	
		
		\section{Triangular Inequality}
		\label{app:TriangularInequality}
		\begin{prop}
	\label{prop:TriangularInequality}
	\normalfont
	Consider three unit vectors $\nmb_{\sss{1}}$, $\nmb_{\sss{2}}$ and $\nmb$. The following triangular inequality is satisfied,
	\begin{align}
		\theta(\nmb_{\sss{1}},\nmb_{\sss{2}}) \le \theta(\nmb_{\sss{1}},\nmb) +\theta(\nmb,\nmb_{\sss{2}}),
	\end{align}
	where $\theta: \mathcal{S}^{\sss{2}} \times \mathcal{S}^{\sss{2}} \mapsto [0,\pi]$, defined as $\theta(\nmb_{\sss{1}},\nmb_{\sss{2}}) = \arccos(\nmb_{\sss{1}}\tp \nmb_{\sss{2}})$.
\end{prop}
\begin{proof}
	\normalfont
	Since $\nmb_{\sss{1}}$ and $\nmb_{\sss{2}}$ can be written as
	\begin{align}
		& \nmb_{\sss{1}} = \cos(\theta(\nmb_{\sss{1}},\nmb)) \nmb + \sin(\theta(\nmb_{\sss{1}},\nmb))   \bm{\nu}_{\sss{1}} ,
		\\
		& \nmb_{\sss{2}} = \cos(\theta(\nmb_{\sss{2}},\nmb)) \nmb + \sin(\theta(\nmb_{\sss{2}},\nmb))  \bm{\nu}_{\sss{2}} ,
	\end{align}
	where $\bm{\nu}_{\sss{1}}$, $\bm{\nu}_{\sss{2}} \in \mathcal{S}^{2}$ are unit vectors orthogonal to $\nmb$, then
	\begin{align}
		& 
		\Scale[0.75]{ 
			\nmb_{\sss{1}}\tp\nmb_{\sss{2}} 
			=
		}
		\\
		&		
		\Scale[0.75]{ 
			=
			\cos(\theta(\nmb_{\sss{1}},\nmb)) \cos(\theta(\nmb_{\sss{2}},\nmb))  
			+ 
			\sin(\theta(\nmb_{\sss{1}},\nmb)) \sin(\theta(\nmb_{\sss{2}},\nmb)) 
			\bm{\nu}_{\sss{1}}\tp \bm{\nu}_{\sss{2}} 
		}
		\\
		&
		\Scale[0.75]{
		 	=
			\cos(\theta(\nmb_{\sss{1}},\nmb) + \theta(\nmb_{\sss{2}},\nmb))  
			+
			\underbrace{
				\sin(\theta(\nmb_{\sss{1}},\nmb)) \sin(\theta(\nmb_{\sss{2}},\nmb)) 
				\left(
					1 + \bm{\nu}_{\sss{1}}\tp \bm{\nu}_{\sss{2}} 
				\right)
			}_{\ge 0}
		}
		\\
		&
		\Scale[0.75]{
			\ge
			\cos(\theta(\nmb_{\sss{1}},\nmb) + \theta(\nmb_{\sss{2}},\nmb)).
		}
	\end{align}	
	As such, $	\theta(\nmb_{\sss{1}},\nmb_{\sss{2}}) \le \theta(\nmb_{\sss{1}},\nmb) + \theta(\nmb_{\sss{2}},\nmb)$.
\end{proof}

\begin{proof}[of Proposition~\ref{prop:BelongToCone2}]
	\normalfont
	Let us construct an $\bm{\nu} \in \mathcal{S}^{\sss{2}}$ such that $\theta(\nmbi[k],\bm{\nu}) \le \alpha$ for all $k\in \mathcal{N}$, with $\theta(\cdot,\cdot)$ defined in Proposition~\ref{prop:TriangularInequality}.
	If $\max_{\sss{(i,j) \in \mathcal{N}^{\sss{2}}}}\theta(\nmbi[i],\nmbi[j]) = 0$, it suffices to pick $\bm{\nu} = \nmbi[i]$, for any $i\in \mathcal{N}$, and the Proposition's conclusion follows.
	If $0 < \max_{\sss{(i,j) \in \mathcal{N}^{\sss{2}}}}\theta(\nmbi[i],\nmbi[j]) $, take $(k,l) = \arg \max_{\sss{(i,j) \in \mathcal{N}^{\sss{2}}}}\theta(\nmbi[i],\nmbi[j])$, and consider then $\bm{\nu} = \frac{ \nmbi[k] + \nmbi[l] }{\norm{\nmbi[k] + \nmbi[l]}}$, which is well defined since $0 < \theta(\nmbi[k],\nmbi[l]) < \frac{2}{3} \pi < \pi$. 
	By construction, it follows that $\theta(\nmbi[k],\bm{\nu}) = \theta(\nmbi[l],\bm{\nu}) = \frac{\theta(\nmbi[k],\nmbi[l])}{2} $. 
	From Proposition~\ref{prop:TriangularInequality}, it follows that, for any $j \in \mathcal{N}$, $\theta(\nmbi[j],\bm{\nu}) \le \theta(\nmbi[j],\nmbi[k]) + \theta(\nmbi[k],\bm{\nu} ) \le \theta(\nmbi[l],\nmbi[k]) + \theta(\nmbi[k],\bm{\nu} ) \le  \frac{3}{2}\theta(\nmbi[k],\nmbi[l])$.
	Thus, if $\theta(\nmbi[k],\nmbi[l]) \le  \frac{2}{3}\alpha$, for some $\alpha \in [0,\pi]$, then $\theta(\nmbi[j],\bm{\nu}) \le \alpha$ for all $k\in \mathcal{N}$, and the Proposition's conclusion follows.
\end{proof}	
		
		\section{Solution to PDE~(IV.1)}
		\label{app:Uniqueness}
		The partial differential equation~\eqref{eq:DistanceFunctionProperty} is a restriction imposed on the distance functions. 
The question that follows is what \emph{kinds} of distance functions $d: \mathcal{S}^{2} \times \mathcal{S}^{2} \rightarrow \mathbb{R}_{\sss{0}}^{\sss{+}}$ satisfy~\eqref{eq:DistanceFunctionProperty}.
Distance functions of the type $d(\nmbi[1],\nmbi[2]) = f(1 - \nmbi[1]\tp \nmbi[2])$ (with $f \in \mathcal{C}^{1}((0,2), \mathbb{R}_{\sss{\ge 0}})$) satisfy~\eqref{eq:DistanceFunctionProperty}.
In fact, those are the only type of distance functions that satisfy~\eqref{eq:DistanceFunctionProperty}, as we show next.

\begin{prop}
	\label{prop:UniquenessAuxiliar}
	\normalfont
	Consider two orthogonal unit vectors, i.e., $
	(\vmb_{\sss{1}},\vmb_{\sss{2}}) 
	\in 
	\Omega 
	=
	\{ 
		(\nmb_{\sss{1}},\nmb_{\sss{2}}) \in \mathcal{S}^{\sss{2}} \times \mathcal{S}^{\sss{2}} :  
		\nmb_{\sss{1}}\tp \nmb_{\sss{2}} = 0
	\}$.
	Additionally, consider a function $\eta \in \mathcal{C}^{\sss{1}}(\Rn[3] \times \Rn[3], \Rn[]_{\sss{\ge 0}})$ that satisfies
	\begin{align}
		\sk{\vmb_{\sss{1}}}
		\frac{\partial \eta(\vmb_{\sss{1}},\vmb_{\sss{2}})}{\partial \vmb_{\sss{1}}} 
		+
		\sk{\vmb_{\sss{2}}}
		\frac{\partial \eta(\vmb_{\sss{1}},\vmb_{\sss{2}})}{\partial \vmb_{\sss{2}}}
		=
		\zvec,	
		\label{eq:RestrictionFUnique}
	\end{align}
	for all $(\vmb_{\sss{1}},\vmb_{\sss{2}}) \in  \Omega$.
	Then, $\eta(\vmb_{\sss{1}},\vmb_{\sss{2}})$ is constant for all $(\vmb_{\sss{1}},\vmb_{\sss{2}}) \in  \Omega$.
\end{prop}

\begin{proof}
	\normalfont
	The function $\eta(\cdot,\cdot)$ is constant for all $(\vmb_{\sss{1}},\vmb_{\sss{2}}) \in  \Omega$, if the differential $d\eta(\cdot,\cdot)$ is zero along $(\vmb_{\sss{1}},\vmb_{\sss{2}}) \in  \Omega$. 
	This condition is proved next.
	Taking the inner product of~\eqref{eq:RestrictionFUnique} with $\vmb_{\sss{i}}$ ($i \in \{1,2\}$ and $j = \{1,2\}\backslash \{i\}$), it follows that $\vmb_{\sss{i}}\tp \sk{\vmb_{\sss{j}}}
	\frac{\partial \eta(\vmb_{\sss{1}},\vmb_{\sss{2}})}{\partial \vmb_{\sss{j}}} = 0$, which means $\frac{\partial \eta(\vmb_{\sss{1}},\vmb_{\sss{2}})}{\partial \vmb_{\sss{j}}} $ is orthogonal to both $\sk{\vmb_{\sss{1}}} \vmb_{\sss{2}}$ and $\sk{\vmb_{\sss{2}}} \vmb_{\sss{1}}$.
	Since $\vmb_{\sss{1}}$, $\vmb_{\sss{2}}$ and $\sk{\vmb_{\sss{1}}} \vmb_{\sss{2}}$ form a basis of $\Rn[3]$, it then follows that
	\begin{align}
		\Scale[0.95]{
			\frac{\partial \eta(\vmb_{\sss{1}},\vmb_{\sss{2}})}{\partial \vmb_{\sss{1}}} 
		}
		&
		\Scale[0.95]{
		=
		\left(
			\vmb_{\sss{1}}\tp
			\frac{\partial \eta(\vmb_{\sss{1}},\vmb_{\sss{2}})}{\partial \vmb_{\sss{1}}} 
		\right)
		\vmb_{\sss{1}}
		+
		\left(
			\vmb_{\sss{2}}\tp
			\frac{\partial \eta(\vmb_{\sss{1}},\vmb_{\sss{2}})}{\partial \vmb_{\sss{1}}} 
		\right)
		\vmb_{\sss{2}}	
		}
		\\
		&
		\Scale[0.95]{
		\triangleq
		h_{\sss{1}}(\vmb_{\sss{1}},\vmb_{\sss{2}})
		\vmb_{\sss{1}}
		+
		h_{\sss{3}}(\vmb_{\sss{1}},\vmb_{\sss{2}})
		\vmb_{\sss{2}},
		}
		\label{eq:Part1}	
		\\
		\Scale[0.95]{
			\frac{\partial \eta(\vmb_{\sss{1}},\vmb_{\sss{2}})}{\partial \vmb_{\sss{2}}} 
		}
		& 
		\Scale[0.95]{
			=
			\left(
				\vmb_{\sss{1}}\tp
				\frac{\partial \eta(\vmb_{\sss{1}},\vmb_{\sss{2}})}{\partial \vmb_{\sss{2}}} 
			\right)
			\vmb_{\sss{1}}
			+
			\left(
				\vmb_{\sss{2}}\tp
				\frac{\partial \eta(\vmb_{\sss{1}},\vmb_{\sss{2}})}{\partial \vmb_{\sss{2}}} 
			\right)
			\vmb_{\sss{2}}	
		}
		\\
		& 
		\Scale[0.95]{
			\triangleq
			h_{\sss{4}}(\vmb_{\sss{1}},\vmb_{\sss{2}})
			\vmb_{\sss{1}}
			+
			h_{\sss{2}}(\vmb_{\sss{1}},\vmb_{\sss{2}})
			\vmb_{\sss{2}}	
		}
		\label{eq:Part2}.				
	\end{align}
	By replacing~\eqref{eq:Part1} and~\eqref{eq:Part2} in~\eqref{eq:RestrictionFUnique}, it follows that $h_{\sss{3}}(\vmb_{\sss{1}},\vmb_{\sss{2}}) = h_{\sss{4}}(\vmb_{\sss{1}},\vmb_{\sss{2}}) \triangleq h (\vmb_{\sss{1}},\vmb_{\sss{2}})$ for all $(\vmb_{\sss{1}},\vmb_{\sss{2}}) \in  \Omega$ (notice that $\sk{\vmb_{\sss{1}}} \vmb_{\sss{2}} \ne \zvec$).
	Now, recall that $(\vmb_{\sss{1}},\vmb_{\sss{2}}) \in  \Omega$ means that $\vmb_{\sss{1}}\tp \vmb_{\sss{1}} = 1$, $\vmb_{\sss{2}}\tp \vmb_{\sss{2}} = 1$ and $\vmb_{\sss{1}}\tp \vmb_{\sss{2}} = 0$.
	In turn, this implies that $d\vmb_{\sss{1}}\tp \vmb_{\sss{1}} = 0$, $d\vmb_{\sss{2}}\tp \vmb_{\sss{2}} = 0$ and $d\vmb_{\sss{1}}\tp \vmb_{\sss{2}} + d\vmb_{\sss{2}}\tp \vmb_{\sss{1}}= 0$.
	From~\eqref{eq:Part1}, it follows that $
	d\vmb_{\sss{1}}\tp \frac{\partial 
	\eta(\vmb_{\sss{1}},\vmb_{\sss{2}})}{\partial \vmb_{\sss{1}}} 
	= 
	h_{\sss{1}}(\cdot,\cdot)
	d\vmb_{\sss{1}}\tp\vmb_{\sss{1}}
	+
	h(\cdot,\cdot)
	d\vmb_{\sss{1}}\tp\vmb_{\sss{2}}
	=
	h(\cdot,\cdot)
	d\vmb_{\sss{1}}\tp\vmb_{\sss{2}}
	$, and, from~\eqref{eq:Part2}, it follows that $
	d\vmb_{\sss{2}}\tp \frac{\partial 
	\eta(\vmb_{\sss{1}},\vmb_{\sss{2}})}{\partial \vmb_{\sss{2}}} 
	= 
	h(\cdot,\cdot)
	d\vmb_{\sss{2}}\tp\vmb_{\sss{1}}
	+
	h_{\sss{2}}(\cdot,\cdot)
	d\vmb_{\sss{2}}\tp\vmb_{\sss{2}}
	=
	h(\cdot,\cdot)
	d\vmb_{\sss{2}}\tp\vmb_{\sss{1}}
	$.
	Since $
	d\eta(\vmb_{\sss{1}},\vmb_{\sss{2}}) = 
	d\vmb_{\sss{1}}\tp \frac{\partial 	
	\eta(\vmb_{\sss{1}},\vmb_{\sss{2}})}{\partial \vmb_{\sss{1}}} 			
	+
	d\vmb_{\sss{2}}\tp \frac{\partial 	
	\eta(\vmb_{\sss{1}},\vmb_{\sss{2}})}{\partial \vmb_{\sss{2}}} 	
	=	
	$
	$
	h(\vmb_{\sss{1}},\vmb_{\sss{2}})
	\left(
		d\vmb_{\sss{1}}\tp\vmb_{\sss{2}}
		+
		d\vmb_{\sss{2}}\tp\vmb_{\sss{1}}
	\right)
	= 
	0
	$, the Proposition's conclusion follows.
\end{proof}

\begin{prop}
	\normalfont
	Consider a function  $\eta \in \mathcal{C}^{\sss{1}}(\Rn[3] \times \Rn[3], \Rn[]_{\sss{\ge 0}})$ that satisfies
	\begin{align}
		\sk{\nmb_{\sss{1}}}
		\frac{\partial \eta(\nmb_{\sss{1}},\nmb_{\sss{2}})}{\partial \nmb_{\sss{1}}} 
		+
		\sk{\nmb_{\sss{2}}}
		\frac{\partial \eta(\nmb_{\sss{1}},\nmb_{\sss{2}})}{\partial \nmb_{\sss{2}}}
		=
		\zvec,	
		\label{eq:RestrictionFUnique2}
	\end{align}
	for all $(\nmb_{\sss{1}},\nmb_{\sss{2}}) \in  \mathcal{S}^{\sss{2}} \times \mathcal{S}^{\sss{2}}$.
	Then, $\eta(\nmb_{\sss{1}},\nmb_{\sss{2}}) = h(1 - \nmb_{\sss{1}}\tp\nmb_{\sss{2}})$ for $(\nmb_{\sss{1}},\nmb_{\sss{2}}) \in  \mathcal{S}^{\sss{2}} \times \mathcal{S}^{\sss{2}}$ and for any $h \in \mathcal{C}^{1}((0,2), \mathbb{R}_{\sss{\ge 0}})$.
\end{prop}

\begin{proof}
	\normalfont
	Condition~\eqref{eq:RestrictionFUnique2} is satisfied for $(\nmb_{\sss{1}},\nmb_{\sss{2}}) \in  \mathcal{S}^{\sss{2}} \times \mathcal{S}^{\sss{2}}$, and, in particular, it is satisfied for $(\nmb_{\sss{1}},\nmb_{\sss{2}}) \in  \varTheta \triangleq \{ (\vmb_{\sss{1}},\vmb_{\sss{2}}) \in \mathcal{S}^{\sss{2}} \times \mathcal{S}^{\sss{2}}: (\vmb_{\sss{1}}\tp\vmb_{\sss{2}})^2 \ne 1 \}$ which is a set of full measure w.r.t. $\mathcal{S}^{\sss{2}} \times \mathcal{S}^{\sss{2}}$.
	For $(\nmb_{\sss{1}},\nmb_{\sss{2}}) \in  \varTheta$, consider the variables
	\begin{align}
		\vmb_{\sss{1}} = \frac{\nmbi[1] + \nmbi[2]}{\norm{\nmbi[1] + \nmbi[2]}},
		\vmb_{\sss{2}} = \frac{\nmbi[1] - \nmbi[2]}{\norm{\nmbi[1] - \nmbi[2]}},
		v_{\sss{3}}    = \nmbi[1]\tp\nmbi[2],
	\end{align}
	and notice that, by construction, $\nmb_{\sss{1}}$ and $\nmb_{\sss{2}}$ can be obtained from $\vmb_{\sss{1}}$, $\vmb_{\sss{2}}$ and $v_{\sss{3}}$ (indeed, $\nmbi[1,2] = v_{\sss{3}}\vmb_{\sss{1}} \pm \sqrt{1 - v_{\sss{3}}^2} \vmb_{\sss{2}}$), meaning that a bijection exists between $(\nmb_{\sss{1}},\nmb_{\sss{2}})$ and $(\vmb_{\sss{1}},\vmb_{\sss{2}},v_{\sss{3}})$; moreover, $\vmb_{\sss{1}}$ is orthogonal to $\vmb_{\sss{2}}$, i.e., $\vmb_{\sss{1}}\tp\vmb_{\sss{2}} = 0$.
	The proof hereafter explores the previous change of variables. 
	First, notice that there exists a function $\tilde{h}(\cdot,\cdot,\cdot)$ such that $
	\eta(\nmb_{\sss{1}},\nmb_{\sss{2}}) = \tilde{h}(
	\vmb_{\sss{1}}(\nmb_{\sss{1}},\nmb_{\sss{2}}),
	\vmb_{\sss{2}}(\nmb_{\sss{1}},\nmb_{\sss{2}}),
	v_{\sss{3}}(\nmb_{\sss{1}},\nmb_{\sss{2}})
	)$ or alternatively $
		\eta(
		\nmb_{\sss{1}}(\vmb_{\sss{1}},\vmb_{\sss{2}},v_{\sss{3}}),
		\nmb_{\sss{2}}(\nmb_{\sss{1}},\nmb_{\sss{2}},v_{\sss{3}})
		)
		=
		\tilde{h}(\vmb_{\sss{1}},\vmb_{\sss{2}},v_{\sss{3}})$.
	Then, from~\eqref{eq:RestrictionFUnique2}, it follows that (for simplicity, all arguments in the equations below are omitted)
	\begin{align}
		\hspace{-0.4cm}
		& 
		\sk{\nmb_{\sss{1}}}
		\left(
			\frac{\partial \vmb_{\sss{1}}}{\partial \nmb_{\sss{1}}} 
			\frac{\partial \tilde{h}}{\partial \vmb_{\sss{1}}} 
			+
			\frac{\partial \vmb_{\sss{2}}}{\partial \nmb_{\sss{1}}}
			\frac{\partial \tilde{h}}{\partial \vmb_{\sss{2}}} 
			+
			\frac{\partial v_{\sss{3}}}{\partial \nmb_{\sss{1}}}
			\frac{\partial \tilde{h}}{\partial v_{\sss{3}}} 
		\right)
		+
		\cdots
		\\
		\hspace{-0.6cm}
		&
		\sk{\nmb_{\sss{2}}}
		\left(
			\frac{\partial \vmb_{\sss{1}}}{\partial \nmb_{\sss{2}}} 
			\frac{\partial \tilde{h}}{\partial \vmb_{\sss{1}}} 
			+
			\frac{\partial \vmb_{\sss{2}}}{\partial \nmb_{\sss{2}}}
			\frac{\partial \tilde{h}}{\partial \vmb_{\sss{2}}} 
			+
			\frac{\partial v_{\sss{3}}}{\partial \nmb_{\sss{2}}}
			\frac{\partial \tilde{h}}{\partial v_{\sss{3}}} 
		\right)
		=
		\zvec	
		\label{eq:partial1}
		\\
		\hspace{-0.6cm}
		\Leftrightarrow
		&
		\left(
			\sk{\nmb_{\sss{1}}}
			\frac{\partial \vmb_{\sss{1}}}{\partial \nmb_{\sss{1}}} 
			+
			\sk{\nmb_{\sss{2}}}
			\frac{\partial \vmb_{\sss{1}}}{\partial \nmb_{\sss{2}}} 
		\right)
		\frac{\partial \tilde{h}}{\partial \vmb_{\sss{1}}} 
		+ \cdots
		\\
		\hspace{-0.6cm}
		& 
		\left(
			\sk{\nmb_{\sss{1}}}
			\frac{\partial \vmb_{\sss{2}}}{\partial \nmb_{\sss{1}}} 
			+
			\sk{\nmb_{\sss{2}}}
			\frac{\partial \vmb_{\sss{2}}}{\partial \nmb_{\sss{2}}} 
		\right)
		\frac{\partial \tilde{h}}{\partial \vmb_{\sss{2}}} 
		+\cdots
		\\
		\hspace{-0.6cm}
		& 
		\left(
			\sk{\nmb_{\sss{1}}}
			\nmb_{\sss{2}} 
			+
			\sk{\nmb_{\sss{2}}}
			\nmb_{\sss{1}}
		\right)
		\frac{\partial \tilde{h}}{\partial v_{\sss{3}}} 
		=
		\zvec
		\label{eq:partial2}	
		\\	
		\hspace{-0.6cm}
		\Leftrightarrow
		& \sk{\vmb_{\sss{1}}}
		\frac{\partial \tilde{h}}{\partial \vmb_{\sss{1}}} 
		+
		\sk{\vmb_{\sss{2}}}
		\frac{\partial \tilde{h}}{\partial \vmb_{\sss{2}}} 
		=
		\zvec	
		\label{eq:partial3}	
	\end{align}
	where the chain rule has been applied to derive~\eqref{eq:partial1} from~\eqref{eq:RestrictionFUnique2}; and from~\eqref{eq:partial2} to~\eqref{eq:partial3}, the relations (recall that $\frac{\partial }{\partial \xmb} \left(\frac{\xmb}{\norm{\xmb}}\right) = \left( \Idmat - \frac{\xmb}{\norm{\xmb}} \frac{\xmb\tp}{\norm{\xmb}} \right) \frac{1}{\norm{\xmb}}$)
	\begin{align}
		\frac{\partial \vmb_{\sss{1}}}{\partial \nmb_{\sss{i}}}
		& =
		\hphantom{\pm}
		\left(
			\Idmat - \vmb_{\sss{1}} \vmb_{\sss{1}}\tp
		\right)
		\frac{1}{\norm{\nmbi[1] + \nmbi[2]}}
		,
		\, 
		\text{for}
		\, 
		i \in \{1,2\},
		\\
		\frac{\partial \vmb_{\sss{2}}}{\partial \nmb_{\sss{j}}}
		& =
		\pm
		\left(
			\Idmat - \vmb_{\sss{2}} \vmb_{\sss{2}}\tp
		\right)
		\frac{1}{\norm{\nmbi[1] - \nmbi[2]}}
		,
		\,
		\text{for}
		\, 		
		j \in \{1,2\},		
	\end{align}	
	together with the relations $
	\sk{\nmb_{\sss{1}}} \vmb_{\sss{1}}
	+
	\sk{\nmb_{\sss{2}}} \vmb_{\sss{1}}
	=
	\zvec$ and $
	\sk{\nmb_{\sss{1}}} \vmb_{\sss{2}}
	-
	\sk{\nmb_{\sss{2}}} \vmb_{\sss{2}}
	=
	\zvec
	$ have been used.
	%
	%
	It follows from Proposition~\ref{prop:UniquenessAuxiliar} and from~\eqref{eq:partial3}, that $\tilde{h}(\vmb_{\sss{1}}, \vmb_{\sss{2}}, v_{\sss{3}})$ does not depend on $\vmb_{\sss{1}}$ and $\vmb_{\sss{2}}$, and, as such, $\eta(\nmb_{\sss{1}},\nmb_{\sss{2}}) = \tilde{h}(\cdot,\cdot,\nmb_{\sss{1}}\tp\nmb_{\sss{2}}) = h(1 - \nmb_{\sss{1}}\tp\nmb_{\sss{2}})$.
\end{proof}
%
%
%
%

\section{Auxiliary Examples (Section~\ref{subsec:Preliminaries})}
\label{app:PreliminariesAuxiliaryExamples}

\begin{exa}
	In Fig.~\ref{fig:GraphsIncidenceMatrix}, three graphs with equivalent incidence matrices are presented (incidence matrices in~\eqref{eq:B1}, \eqref{eq:B2} and~\eqref{eq:B3}).
	(Two matrices $D,E \in \Rn[n\times m]$ are said to be equivalent iff there exist non-singular $P\in \Rn[n\times n]$ and $Q\in \Rn[m\times m]$ such that $D = P E Q$~\cite{zwillinger2014table}.)
	The incidence matrix in Fig.~\ref{fig:Graph1} is given by
	\begin{align}
		B_{\sss{1}} 
		= 
		\begin{bmatrix}
			\hphantom{-}1 & \hphantom{-}0 & \hphantom{-}1 \\
			-1            & \hphantom{-}1 & \hphantom{-}0 \\
			\hphantom{-}0 & -1            & -1
		\end{bmatrix}.
		\label{eq:B1}
	\end{align}
	Consider a graph, with incidence matrix $\tilde{B}_{\sss{1}} \in \Rn[N \times M]$, and a second graph, with incidence matrix $\tilde{B}_{\sss{2}} \in \Rn[\sss{N \times M}]$,  obtained from the previous after a swap of nodes. 
	Then $\tilde{B}_{\sss{1}} = \Gamma_{\sss{1}} \tilde{B}_{\sss{2}} \bar{\Idmat}$, where $\Gamma_{\sss{1}} \in \Rn[\sss{N \times N}]$ is a permutation matrix corresponding to the swap of nodes (permutation of rows) and $\bar{\Idmat}\in \Rn[\sss{M \times M}]$ is a diagonal matrix with diagonal elements $\pm 1$ (column $k \in \mathcal{M}$ needs to be multiplied by $-1$ when, after permutation of rows, $i$ becomes larger than $j$ with $(i,j) = \bar{\kappa}^{\sss{-1}}(k)$).
	The graph in Fig.~\ref{fig:Graph2} is obtained from that in Fig.~\ref{fig:Graph1} by swapping nodes~1 and~2, and, as such, its incidence matrix is given by
	\begin{align}
		B_{\sss{2}} 
		= &
		\Gamma_{\sss{1}}
		B_{\sss{1}} 
		\bar{\Idmat}
		= 
		\begin{bmatrix}
			0 & 1 & 0 \\
			1 & 0 & 0 \\
			0 & 0 & 1
		\end{bmatrix}		
		\begin{bmatrix}
			\hphantom{-}1 & \hphantom{-}0 & \hphantom{-}1 \\
			-1            & \hphantom{-}1 & \hphantom{-}0 \\
			\hphantom{-}0 & -1            & -1
		\end{bmatrix}
		\begin{bmatrix}
			-1            & 0 & 0 \\
			\hphantom{-}0 & 1 & 0 \\
			\hphantom{-}0 & 0 & 1
		\end{bmatrix}			
		\\
		= &		
		\begin{bmatrix}
			-1            & \hphantom{-}1 & \hphantom{-}0 \\		
			\hphantom{-}1 & \hphantom{-}0 & \hphantom{-}1 \\
			\hphantom{-}0 & -1            & -1
		\end{bmatrix}	
		\begin{bmatrix}	
			-1            & 0 & 0 \\
			\hphantom{-}0 & 1 & 0 \\
			\hphantom{-}0 & 0 & 1
		\end{bmatrix}	
		=
		\begin{bmatrix}
			\hphantom{-}1 & \hphantom{-}1 & \hphantom{-}0 \\		
			-1            & \hphantom{-}0 & \hphantom{-}1 \\
			\hphantom{-}0 & -1            & -1
		\end{bmatrix},
		\label{eq:B2}
	\end{align}	
	where column 1 needs to be multiplied by -1 since rows~1 and~2 have permuted and $\bar{\kappa}(1,2) = 1$. 
	Consider again a graph, with incidence matrix $\tilde{B}_{\sss{1}} \in \Rn[N \times M]$, and a second graph, with incidence matrix $\tilde{B}_{\sss{2}} \in \Rn[\sss{N \times M}]$,  obtained from the previous after a swap of edges (edges' numbers).
	Then $\tilde{B}_{\sss{1}} = \tilde{B}_{\sss{2}} \Gamma_{\sss{2}}$, where $\Gamma_{\sss{2}} \in \Rn[\sss{M \times M}]$ is a permutation matrix corresponding to the swap of edges (permutation of columns).	
	The graph in Fig.~\ref{fig:Graph3} is obtained from that in Fig.~\ref{fig:Graph2} by swapping edges~1 and~3, and, as such, its incidence matrix is given by
	\begin{align}
		\hspace{-0.5cm}
		\Scale[0.85]{
		B_{\sss{3}} 
		= 
		B_{\sss{2}} 
		\Gamma_{\sss{2}}
		=
		\begin{bmatrix}
			\hphantom{-}1 & \hphantom{-}1 & \hphantom{-}0 \\		
			-1            & \hphantom{-}0 & \hphantom{-}1 \\
			\hphantom{-}0 & -1            & -1
		\end{bmatrix}	
		\begin{bmatrix}
			0 & 0 & 1 \\
			0 & 1 & 0 \\
			1 & 0 & 0
		\end{bmatrix}			
		=		
		\begin{bmatrix}
			\hphantom{-}0 & \hphantom{-}1 & \hphantom{-}1 \\		
			\hphantom{-}1 & \hphantom{-}0 & -1 \\
			-1            & -1            & \hphantom{-}0
		\end{bmatrix}.
		}
		\label{eq:B3}
	\end{align}			
\end{exa}
%
%
\begin{figure*}
	\centering
	\subcapcentertrue	
	\subfigure[Graph with incidence matrix $B_{\sss{1}}$ in~(XI.1)]{
			\centering
			\begin{tikzpicture}[-,>=stealth',shorten >=1pt,auto,node distance=3cm,
			                    thick,main node/.style={circle,draw,font=\sffamily\Large\bfseries},scale=0.4,every node/.style={scale=0.5}]
			
			  \node[main node] (1) {$\nmbi[1]$};
			  \node[main node] (2) [above left of=1] {$\nmbi[2]$};
			  \node[main node] (3) [below left of=2] {$\nmbi[3]$};

			  \path[every node/.style={font=\sffamily\small}]
			    (1) edge node[xshift=0.7cm, yshift=0.25cm] {{\fontsize{0.05cm}{1em}$\bar{\kappa}(1,2) = 1 $}} (2)
			    (2) edge node[xshift=-1.3cm, yshift=0.25cm] {{\fontsize{0.05cm}{1em}$\bar{\kappa}(2,3) = 2 $}} (3)
			    (3) edge node[xshift= 0.0cm, yshift=-0.4cm] {{\fontsize{0.05cm}{1em}$\bar{\kappa}(1,3) = 3 $}} (1)
			    ;
			\end{tikzpicture}	
			\label{fig:Graph1}
	}
	\subfigure[Same graph as in Fig.~\ref{fig:Graph1} apart from a swap between nodes~1 and~2 (incidence matrix $B_{\sss{2}}$ in~(XI.2)).]{
			\centering
			\begin{tikzpicture}[-,>=stealth',shorten >=1pt,auto,node distance=3cm,
			                    thick,main node/.style={circle,draw,font=\sffamily\Large\bfseries},scale=0.4,every node/.style={scale=0.5}]
			
			  \node[main node] (1) {$\nmbi[2]$};
			  \node[main node] (2) [above left of=1] {$\nmbi[1]$};
			  \node[main node] (3) [below left of=2] {\color{gray}$\nmbi[3]$};

			  \path[every node/.style={font=\sffamily\small}]
			    (1) edge node[xshift=0.7cm, yshift=0.25cm] {\color{gray}{\fontsize{0.05cm}{1em}$\bar{\kappa}(1,2) = 1 $}} (2)
			    (2) edge node[xshift=-1.3cm, yshift=0.25cm] {\color{gray}{\fontsize{0.05cm}{1em}$\bar{\kappa}(1,3) = 2 $}} (3)
			    (3) edge node[xshift= 0.0cm, yshift=-0.4cm] {\color{gray}{\fontsize{0.05cm}{1em}$\bar{\kappa}(2,3) = 3 $}} (1)
			    ;
			\end{tikzpicture}	
			\label{fig:Graph2}
	}
	\subfigure[Same graph as in Fig.~\ref{fig:Graph2} apart from a swap between edges~1 and~3 (incidence matrix $B_{\sss{3}}$ in~(XI.3)).]{
			\centering
			\begin{tikzpicture}[-,>=stealth',shorten >=1pt,auto,node distance=3cm,
			                    thick,main node/.style={circle,draw,font=\sffamily\Large\bfseries},scale=0.4,every node/.style={scale=0.5}]
			
			  \node[main node] (1) {\color{gray}$\nmbi[2]$};
			  \node[main node] (2) [above left of=1] {\color{gray}$\nmbi[1]$};
			  \node[main node] (3) [below left of=2] {\color{gray}$\nmbi[3]$};

			  \path[every node/.style={font=\sffamily\small}]
			    (1) edge node[xshift=0.7cm, yshift=0.25cm] {{\fontsize{0.05cm}{1em}$\bar{\kappa}(1,2) = 3 $}} (2)
			    (2) edge node[xshift=-1.3cm, yshift=0.25cm] {\color{gray}{\fontsize{0.05cm}{1em}$\bar{\kappa}(1,3) = 2 $}} (3)
			    (3) edge node[xshift= 0.0cm, yshift=-0.4cm] {{\fontsize{0.05cm}{1em}$\bar{\kappa}(2,3) = 1 $}} (1)
			    ;
			\end{tikzpicture}	
			\label{fig:Graph3}
	}			
	\caption{Three graphs with equivalent incidence matrices.}
	\label{fig:GraphsIncidenceMatrix}
\end{figure*}
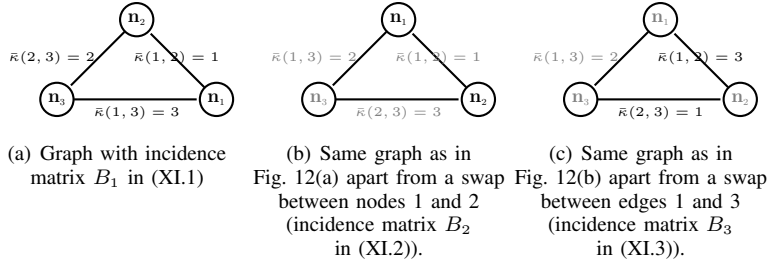
\begin{exa}
	Figure~\ref{fig:IndependentCycles} displays a graph with two independent cycles, with incidence matrix and its null space in~\eqref{eq:IndependentCycles}.
	\begin{align}
		\hspace{-0.5cm}
		\Scale[0.7]{
			B_{\sss{1}} =
			\begin{bmatrix}
				 \hphantom{-}1 &  \hphantom{-}0 &  \hphantom{-}1 &  \hphantom{-}1 &  \hphantom{-}1 &  \hphantom{-}0 \\
			    -1                    &  \hphantom{-}1 &  \hphantom{-}0 &  \hphantom{-}0 &  \hphantom{-}0 &  \hphantom{-}0 \\
				 \hphantom{-}0 & -1                    &  -1                   &  \hphantom{-}0 &  \hphantom{-}0 &  \hphantom{-}0 \\
				 \hphantom{-}0 &  \hphantom{-}0 &  \hphantom{-}0 & -1                    &  \hphantom{-}0 &  \hphantom{-}1 \\
				 \hphantom{-}0 &  \hphantom{-}0 &  \hphantom{-}0 &  \hphantom{-}0 &  -1                   &  -1
			\end{bmatrix}
			, \,\,
			\mathcal{N}(B_{\sss{1}})
			=
			\text{span}
			\left\{
				\begin{bmatrix}
				 \hphantom{-}1  \\
			     \hphantom{-}1  \\
				 -1             \\
				 \hphantom{-}0  \\
				 \hphantom{-}0  \\
				 \hphantom{-}0 
				\end{bmatrix}
				,
				\begin{bmatrix}
				 \hphantom{-}0   \\
			     \hphantom{-}0   \\
				 \hphantom{-}0   \\
				 \hphantom{-}1   \\
				 -1              \\
				 \hphantom{-}1
				\end{bmatrix}						
			\right\}.	
		}
		\label{eq:IndependentCycles}
	\end{align}
\end{exa}
\begin{exa}
	Figure~\ref{fig:SharedCycles} displays a graph with one independent cycle and two cycles that share only one edge, with incidence matrix and its null space in~\eqref{eq:SharedCycles} ($B_{\sss{1}}$ defined in~\eqref{eq:IndependentCycles}).
	\begin{align}
		\hspace{-0.5cm}
		\Scale[0.7]{
			B =
			\left[
				\begin{array}{cc}
					 B_{\sss{1}} & 
					 \begin{array}{cc}
					  \hphantom{-}1 & \hphantom{-}0  \\
					  \hphantom{-}0 & \hphantom{-}0  \\
					  \hphantom{-}0 & \hphantom{-}1  \\
					  \hphantom{-}0 & \hphantom{-}0  \\
					  \hphantom{-}0 & \hphantom{-}0  \\
					 \end{array} \\
					 \zvec & 
					 \begin{array}{cc}
					  -1 & -1
					 \end{array}
				\end{array}
			\right]
			, \,\,
			\mathcal{N}(B)
			=	
	 		\text{span}
	 		\left\{
	 			\begin{bmatrix}
	  			 \hphantom{-}1   \\
	  		     \hphantom{-}1   \\
	  			 -1              \\
	  			 \hphantom{-}0   \\
	  			 \hphantom{-}0   \\
	  			 \hphantom{-}0   \\
	   			 \hphantom{-}0   \\
	   			 \hphantom{-}0  
	 			\end{bmatrix}
	 			,
	 			\begin{bmatrix}
	  			 \hphantom{-}0   \\
	  		     \hphantom{-}0   \\
	  			 \hphantom{-}1   \\
	  			 \hphantom{-}0   \\
	  			 \hphantom{-}0   \\
	  			 \hphantom{-}0   \\
	   		     -1              \\
	   			 \hphantom{-}1  
	 			\end{bmatrix}		
	  			,
	  			\begin{bmatrix}
	  			 \hphantom{-}0   \\
	  		     \hphantom{-}0   \\
	  			 \hphantom{-}0   \\
	  			 \hphantom{-}1   \\
	  			 -1              \\
	  			 \hphantom{-}1   \\
	   		     \hphantom{-}0   \\
	   			 \hphantom{-}0  			  
	  			\end{bmatrix}								
	 		\right\}.	
 		}	
		\label{eq:SharedCycles}
	\end{align}
\end{exa}

\section{Proofs of Propositions~\ref{prop:BIndependentCycles}, \ref{prop:BAlmostIndependentCycles}}
\label{app:Proofs}

\begin{proof}[of Proposition~\ref{prop:BIndependentCycles}]
	\normalfont
	Without loss of generality, consider a graph with only one cycle composed of $M_c$ edges. Then, its associated incidence matrix $B$ can be divided in two parts, i.e., $B \triangleq [B_{\sss{\bar{c}}} \, B_{\sss{c}}]$: $B_{\sss{c}}$ corresponds to all edges that are part of the cycle; and $B_{\sss{\bar{c}}}$ corresponds to all other edges. $B_{\sss{c}}$ can be further partitioned in two parts, i.e., $B_{\sss{c}} \triangleq [B_{\sss{c}}^{\sss{1}} \, B_{\sss{c}}^{\sss{2}}]$, where $B_{\sss{c}}^{\sss{2}}$ is a single column matrix corresponding to a single edge of the cycle and such that $[ B_{\sss{\bar{c}}} \, B_{\sss{c}}^{\sss{1}}]$ corresponds to a tree.
	Since $[ B_{\sss{\bar{c}}} \, B_{\sss{c}}^{\sss{1}}]$ is a tree, it must be positive definite (see Proposition~\ref{prop:Bdefinitepositive}); thus, it follows that the null space of $B$ has dimension 1, i.e., $| \mathcal{N}(B) | = 1$, and there exists a unique (up to a scalar multiplication) non-zero vector $\emb \in \Rn[M]$ such that $B \emb = \zvec$. Without loss of generality, assume that
	\begin{align}
		\hspace{-0.2cm}
		B_{\sss{c}} 
		\triangleq 
		\setlength{\dashlinegap}{2pt}
		\left[
			\begin{array}{c:c}
				B_{\sss{c}}^{\sss{1}} & B_{\sss{c}}^{\sss{2}} \\
			\end{array}
		\right]
		=
		\Scale[0.75]{
		\setlength{\dashlinegap}{2pt}
		\left[
			\begin{array}{cccc:c}
				\hphantom{-}1      &  \hphantom{-}0     & \cdots & \hphantom{-}0      & \hphantom{-}1     \\
				-1                 &  \hphantom{-}1     & \cdots & \hphantom{-}0      & \hphantom{-}0      \\
				\hphantom{-}0      & -1                 & \cdots & \hphantom{-}0      & \hphantom{-}0      \\
				\hphantom{-}\vdots & \hphantom{-}\vdots & \ddots & \hphantom{-}\vdots & \hphantom{-}\vdots \\
				\hphantom{-}0      & \hphantom{-}0      & \cdots & \hphantom{-}1      & \hphantom{-}0    \\
				\hphantom{-}0      & \hphantom{-}0      & \cdots & -1                 & -1    \\
				\hphantom{-}\zvec  & \hphantom{-}\zvec  & \zvec  & \hphantom{-}\zvec  & \hphantom{-}\zvec
			\end{array}
		\right],	
		}
		\label{eq:Bcycle}					
	\end{align}  		
	which means, $B_{\sss{c}}^{\sss{2}} = \sum_{\sss{i = 1}}^{\sss{ M_c - 1}} (B_{\sss{c}}^{\sss{1}})_{\sss{i}}$. Consequently, the non-zero vector $\emb = [\zvec_{\sss{M - M_c}}\tp \, \onesvec_{\sss{ M_c - 1}}\tp \, -1]\tp$ belongs to the null-space of $B$, since $[ B_{\sss{\bar{c}}} \, B_{\sss{c}}^{\sss{1}} \, B_{\sss{c}}^{\sss{2}}] \emb = \sum_{\sss{i = 1}}^{\sss{ M_c - 1}} (B_{\sss{c}}^{\sss{1}})_{\sss{i}} - B_{\sss{c}}^{\sss{2}} = \zvec$. 
	For $B \otimes  \Idmat_{n}$, its null-space is spanned by the columns of $
						\emb \otimes \Idmat_{n}
					 =
					 [
						 	\zvec \,  \Idmat_{n} \, \cdots \,  \Idmat_{n} \, -\Idmat_{n}
					 ]\tp$.
	In general, the incidence matrix will be of the form $B = \Gamma_{\sss{1}} [B_{\sss{\bar{c}}} \, B_{\sss{c}}^{\sss{1}} \, B_{\sss{c}}^{\sss{2}}] \bar{\Idmat} \, \Gamma_{\sss{2}}$, with $\Gamma_{\sss{1}} \in \Rn[N \times N]$ and $\Gamma_{\sss{2}} \in \Rn[M \times M]$ as permutation matrices and  $\bar{\Idmat}\in \Rn[\sss{M \times M}]$ as a diagonal matrix with diagonal elements $\pm 1$; i.e., the incidence matrix $B$ will be equivalent to $[B_{\sss{\bar{c}}} \, B_{\sss{c}}^{\sss{1}} \, B_{\sss{c}}^{\sss{2}}] $  apart from a reordering of nodes' and edges' numbering.
	In that case, the null space of B will be spanned by $\Gamma_{\sss{2}}\tp \bar{\Idmat}  [ \zvec_{\sss{M - M_c}}\tp \, \onesvec_{\sss{ M_c - 1}}\tp \, -1]\tp$, or, alternatively,  $\mathcal{N}(B) = \{ \emb \in \Rn[M]: e_{\sss{k}} = \pm e_{\sss{l}}, \forall k,l \in C  \}$, where $C$ is the set of the edges that form the cycle.
	For $m$ independent cycles, $\{C_{\sss{i}}\}_{\sss{i=\{1,\cdots,m\}}}$, i.e., cycles that do not share edges, $\mathcal{N}(B) = \{ \emb \in \Rn[M]: e_{\sss{k}} = \pm e_{\sss{l}}, \forall k,l \in C_i, i = \{ 1, \cdots, m\}   \}$.
	Similar conclusions extend to $\mathcal{N}(B \otimes \Idmat)$.
\end{proof}

\begin{proof}[of Proposition~\ref{prop:BAlmostIndependentCycles}]
	\normalfont
	Without loss of generality, consider a graph with only one pair of cycles that share only one edge. 
	Similarly to the proof of Proposition~\ref{prop:BIndependentCycles}, we partition the incidence matrix as $B \triangleq [B_{\sss{\bar{c}}} \, B_{\sss{c}}]$ and through a similar argument we conclude that $|\mathcal{N}(B)| = 2$. In this scenario, $B_{\sss{c}}$ is decomposed in three parts. Without loss of generality, assume that
	\begin{align}
		&
		B_{\sss{c}} 
		\triangleq 
		\setlength{\dashlinegap}{2pt}
		\left[
			\begin{array}{c:c:c}
				B_{\sss{c}}^{\sss{1}} & B_{\sss{c}}^{\sss{2}} & B_{\sss{\text{adj}}}\\
			\end{array}
		\right]
		\\
		&
		=
		\Scale[0.75]{
		\setlength{\dashlinegap}{2pt}
		\left[
			\begin{array}{ccccc:cccc:c}
		\hphantom{-}1       & \hphantom{-}0     &\cdots & \hphantom{-}0      & \hphantom{-}0     & \hphantom{-}1     &  \hphantom{-}0      & \cdots &  \hphantom{-}0      & \hphantom{-}1       \\
		-1                  & \hphantom{-}1     &\cdots & \hphantom{-}0      & \hphantom{-}0     & \hphantom{-}0     &  \hphantom{-}0      & \cdots &  \hphantom{-}0      & \hphantom{-}0       \\
		\hphantom{-}\vdots  & -1                &\ddots & \hphantom{-}\vdots & \hphantom{-}0     & \hphantom{-}0     &  \hphantom{-}\vdots & \ddots &  \hphantom{-}\vdots & \hphantom{-}\vdots  \\
		\hphantom{-} 0      & \hphantom{-}0     &\cdots & \hphantom{-} 1     & \hphantom{-}0     & \hphantom{-}0     &  \hphantom{-}0      & \cdots &  \hphantom{-}0      &  \hphantom{-} 0     \\
		\hphantom{-} 0      & \hphantom{-}0     &\zvec  & -1                 & \hphantom{-}1     & \hphantom{-}0     &  \hphantom{-}0      & \zvec  &  \hphantom{-}0      &  \hphantom{-}0      \\
		\hphantom{-} 0      & \hphantom{-}0     &\zvec  & \hphantom{-} 0     & \hphantom{-}0     & -1                &  \hphantom{-}1      & \cdots &  \hphantom{-}0      &  \hphantom{-}0      \\
		\hphantom{-} 0      & \hphantom{-}0     &\cdots & \hphantom{-}0      & \hphantom{-}\vdots& \hphantom{-}0     &  -1                 & \ddots &  \hphantom{-}\vdots &  \hphantom{-}\vdots \\	 
		\hphantom{-} 0      & \hphantom{-}0     &\cdots & \hphantom{-}0      & \hphantom{-}0     & \hphantom{-}0     &  \hphantom{-}0      & \ddots &  \hphantom{-}1      &  \hphantom{-} 0     \\			
		\hphantom{-} 0      & \hphantom{-}0     &\cdots & \hphantom{-}0      & -1                & \hphantom{-}\vdots&  \hphantom{-}\vdots & \ddots &  -1                 &  -1                 \\	 			
		\hphantom{-} \zvec  & \hphantom{-}\zvec &\cdots & \hphantom{-}\zvec  & \hphantom{-}0     & \hphantom{-}0     &  \hphantom{-}\zvec  & \cdots &  \hphantom{-}\zvec  & \hphantom{-}\zvec
			\end{array}
		\right],	
		}
		\label{eq:Bcycle2}					
	\end{align}
	where $ [ B_{\sss{c}}^{\sss{1}} \, B_{\sss{\text{adj}}} ]$ corresponds to one cycle (of dimension $M_{c_{\sss{1}}} \ge 3$), $[ B_{\sss{c}}^{\sss{2}} \, B_{\sss{\text{adj}}} ]$ corresponds to the second cycle (of dimension $M_{c_{\sss{2}}} \ge 3$), and $B_{\sss{\text{adj}}}$ corresponds to the shared edge. As such, $B_{\sss{\text{adj}}} = \sum_{\sss{i = 1}}^{\sss{M_{c_{\sss{1}}}-1}} (B_{\sss{c}}^{\sss{1}})_{\sss{i}} = \sum_{\sss{i = 1}}^{\sss{M_{c_{\sss{2}}}-1}} (B_{\sss{c}}^{\sss{2}})_{\sss{i}} $, which means that $\mathcal{N}(B)$ is given by
	\begin{align}
		\text{span}
		\left\{
			\begin{bmatrix}
				\zvec\\
				\onesvec_{\sss{M_{c_{\sss{1}}}-1}} \\
				\zvec_{\sss{M_{c_{\sss{2}}}-1}} \\
				-1
			\end{bmatrix}
			,
			\begin{bmatrix}
				\zvec\\
				\zvec_{\sss{M_{c_{\sss{1}}}-1}} \\
				\onesvec_{\sss{M_{c_{\sss{2}}}-1}} \\
				-1
			\end{bmatrix}				
		\right\},
		\label{eq:2CyclesNullSpace2}
	\end{align}
	while $\mathcal{N}(B \otimes \Idmat_n)$ is spanned by the columns of 
	\begin{align}
			\begin{bmatrix}
				\zvec\\
				\onesvec_{\sss{M_{\sss{c_{\sss{1}}}}-1}} \otimes \Idmat_n \\
				\zvec_{\sss{M_{\sss{c_{\sss{2}}}}-1}} \otimes \Idmat_n \\
				-\Idmat_n
			\end{bmatrix}
			,
			\begin{bmatrix}
				\zvec\\
				\zvec_{\sss{M_{\sss{c_{\sss{1}}}}-1}}\otimes \Idmat_n \\
				\onesvec_{\sss{M_{\sss{c_{\sss{2}}}}-1}} \otimes \Idmat_n \\
				-\Idmat_n
			\end{bmatrix}				
		.
		\label{eq:2CyclesNullSpace}
	\end{align}	
	The null space in~\eqref{eq:2CyclesNullSpace2} can be equivalently written as $[\zvec\tp \, \alpha\onesvec_{\sss{M_{\sss{c_{\sss{1}}}}-1}}\tp \, \beta\onesvec_{\sss{M_{\sss{c_{\sss{2}}}}-1}}\tp \, -(\alpha + \beta) ]\tp$ for any $\alpha, \beta \in \Rn[]$.
	In general, the incidence matrix will be of the form $B = \Gamma_{\sss{1}} [ B_{\sss{\bar{c}}} \, B_{\sss{c}}^{\sss{1}} \, B_{\sss{c}}^{\sss{2}} \, B_{\sss{\text{adj}}} ] \bar{\Idmat} \, \Gamma_{\sss{2}}$, with $\Gamma_{\sss{1}} \in \Rn[N \times N]$ and $\Gamma_{\sss{2}} \in \Rn[M \times M]$ as permutation matrices and  $\bar{\Idmat}\in \Rn[\sss{M \times M}]$ as a diagonal matrix with diagonal elements $\pm 1$; i.e., the incidence matrix $B$ will be equivalent to $[ B_{\sss{\bar{c}}} \, B_{\sss{c}}^{\sss{1}} \, B_{\sss{c}}^{\sss{2}} \, B_{\sss{\text{adj}}} ]$  apart from a reordering of nodes' and edges' numbering.
	With that in mind, the null space in~\eqref{eq:2CyclesNullSpace2} can also be written as 
	$
	\mathcal{N}(B) = 
	\{ 
		\emb \in \Rn[M]:
		e_{\sss{k}} =  \pm e_{\sss{l}}, 
		\forall k,l \in 
		C_{}^{\sss{1}}
		\backslash 
		\{C_{}^{\sss{1}} \cap C_{}^{\sss{2}}\},
		e_{\sss{p}} =  \pm e_{\sss{q}}, 
		\forall p,q \in 
		C_{}^{\sss{2}}
		\backslash 
		\{C_{}^{\sss{1}} \cap C_{}^{\sss{2}}\},
	\}$		 
	, with $C_{}^{\sss{1}} = \{ M-M_{\sss{c_{\sss{1}}}}-M_{\sss{c_{\sss{2}}}} + 1, \ldots, M - M_{\sss{c_{\sss{2}}}} -1 \} \cup \{M\}$ and $C_{}^{\sss{2}} = \{ M - M_{\sss{c_{\sss{2}}}}, \ldots, M - 1 \} \cup \{M\}$ (i.e, $C_{}^{\sss{1}}$ corresponds to the set of edges of one cycle, and $C_{}^{\sss{2}}$ corresponds to the set of edges of the other cycle, where the $M^{\sss{th}}$ edge is the shared edge).
	Including independent cycles or other pairs of cycles that share only one edge does not affect the previous conclusions regarding the null space, and, therefore, the conclusions of the Proposition follow.
	%
\end{proof}	
		
		\section{Instability}
		\label{sec:Instability}
		Under Theorem's~\ref{thm:FullDomain} conditions, a configuration where two (or more) neighboring unit vectors are diametrically opposed, with all the other neighboring unit vectors synchronized, is an equilibrium solution.
In fact, consider a constant unit vector $\nmb^{\sss{\star}}$, $\bm{\omega} = 0$ and $\nmbi[i] = \pm \nmb^{\sss{\star}}$ for all $i \in \mathcal{N}$. Under these circumstances, $\Tmbi{i} = \zvec$ for all $i \in \mathcal{N}$ (see~\eqref{eq:DistributedControlLawDynamics}), implying that all those configurations correspond to equilibrium configurations.
Multiple equilibria exist because the control laws~\eqref{eq:DistributedControlLawDynamics} are continuous and the states $\nmbi[i]$ (for all $i \in \mathcal{N}$) evolve in a non-contractible set (see~\cite{liberzon2003switching}).
Notice, however, that the equilibrium $\bm{\omega} = 0$ and $\nmbi[i] = \nmb^{\sss{\star}}$ for all $i \in \mathcal{N}$ is guaranteed to be (locally) asymptotically stable (see Theorem~\ref{thm:LocalStability}).
For all other equilibria, their stability, or lack thereof, is a subject of current research.

In this section, we consider simpler dynamical agents and study the stability of equilibria where two (or more) neighboring unit vectors are diametrically opposed. 
Consider then agents $\nmbi[i]$, for $i \in \mathcal{N}$, with kinematics~\eqref{eq:RotationMatrixKinematics}, and whose angular velocities $\Omi{i}$ are assumed to be controllable (intuitively, this corresponds to controlling single integrators rather than double integrators).
Consider also the Lyapunov function $V(\nmb) = \sum{}_{\sss{k =1}}^{\sss{M}} d_{\sss{k}}(\nmbt{k},\nmbh{k})$, whose time derivative yields $\dot{V} =  \bm{\omega}\tp \Rmat (B \otimes \Idmat) \emb \sum{}_{\sss{k =1}}^{\sss{M}} d_{\sss{k}}(\nmbt{k},\nmbh{k})$.
This motivates the choice of the control law $ [\Omi{1}\tp \, \ldots \, \Omi{N}\tp ]\tp= \bm{\omega} = - \Rmat\tp (B \otimes \Idmat) \emb$, since it implies that $\dot{V} = - \norm{(B \otimes \Idmat) \emb}^2 \le - \lambda_{\sss{\min}}(B \tp B) \norm{\emb}^2$.
Since $\dot{V}(t)$ is uniformly continuous, Barbalat's Lemma suffices to conclude that $\emb$ converges to $\zvec$, which, in turn, implies that there exists a (possibly time-varying) unit vector $\nmb^{\sss{\star}}(t)$ such that $\lim_{t \rightarrow \infty} (\nmbi[i](t) \pm \nmb^{\sss{\star}}(t)) = \zvec$, for all $i \in \mathcal{N}$.

Let us focus on the equilibria configuration where $\nmb^{\sss{\star}}$ is time invariant, i.e, focus on scenarios where all the unit vectors converge a constant unit vector, either $+\nmb^{\sss{\star}}$ or $-\nmb^{\sss{\star}}$ (the study of stability, or lack thereof, for time-varying $\nmb^{\sss{\star}}(t)$ is a subject of current research).
For the remaining part of this section, it is shown that the previous equilibrium configuration is unstable whenever two (or more) neighboring unit vectors are diametrically opposed.
This conclusion follows from a linearization procedure, which is described next.
For a constant $\nmb^{\star}$, consider the variables
\begin{align}
	& \bm{x}_{\sss{i}}
	\triangleq
	\OP{\nmb^{\sss{\star}}} 
	\nmbi[i],
	\label{eq:xi}	
	\\
	& s_{\sss{i}} \triangleq \sign{\nmbi[i]\tp \nmb^{\sss{\star}}},
	\label{eq:si}
\end{align}
for all $i \in \mathcal{N}$. From~\eqref{eq:xi} and~\eqref{eq:si}, it is possible to reconstruct $\nmbi[i]$, since
\begin{align}
	\nmbi[i]
	& = 
	(\nmbi[i]\tp \nmb^{\sss{\star}} ) \nmb^{\sss{\star}}
	+
	\OP{\nmb^{\sss{\star}}} \nmbi[i]
	=
	(\nmbi[i]\tp \nmb^{\sss{\star}} ) \nmb^{\sss{\star}}
	+ 
	\bm{x}_{\sss{i}}
	\label{eq:nmbi2}
	\\
	& =
	\sign{\nmbi[i]\tp \nmb^{\sss{\star}}}
	\sqrt{1 - \norm{\bm{x}_{\sss{i}}}^2} \nmb^{\sss{\star}}
	+ 
	\bm{x}_{\sss{i}}
	\label{eq:nmbi4}	
	\\
	& =
	s_{\sss{i}} \sqrt{1 - \norm{\bm{x}_{\sss{i}}}^2} \nmb^{\sss{\star}}
	+ 
	\bm{x}_{\sss{i}}
	\label{eq:nmbi3}
\end{align}
where the relation $(\nmbi[1]\nmbi[2])^2 + \norm{\OP{\nmbi[1]}\nmbi[2]}^2 = (\nmbi[1]\nmbi[2])^2 + \norm{\sk{\nmbi[1]}\nmbi[2]}^2 = 1$ has been used from~\eqref{eq:nmbi2} to~\eqref{eq:nmbi4}. 
Since $ \nmbi[i]\tp \nmb^{\sss{\star}}$ can only change sign when $\nmbi[i]$ is orthogonal to $\nmb^{\sss{\star}}$, it follows that 
\begin{align}
	\dot{s}_{\sss{i}} = 0 ,
	\label{eq:sdot}
\end{align}
if $\norm{\bm{x}_{\sss{i}}} < 1$ for all time ($\norm{\bm{x}_{\sss{i}}} = \sqrt{1 - (\nmbi[i]\tp \nmb^{\sss{\star}})^2}$).
Additionally, from~\eqref{eq:nmbi3}, it follows that 
\begin{align}
	\nmbi[i]\tp\nmbi[k]
	& =
	\Scale[0.9]{
		(
			s_{\sss{i}} \sqrt{1 - \norm{\bm{x}_{\sss{i}}}^2} \nmb^{\sss{\star}}
			+ 
			\bm{x}_{\sss{i}}
		)\tp
		(
			s_{\sss{k}} \sqrt{1 - \norm{\bm{x}_{\sss{k}}}^2} \nmb^{\sss{\star}}
			+ 
			\bm{x}_{\sss{k}}
		)
	}
	\\
	& =	
	s_{\sss{i}} s_{\sss{k}}
	\sqrt{1 - \norm{\bm{x}_{\sss{i}}}^2} \sqrt{1 - \norm{\bm{x}_{\sss{k}}}^2} 
	+
	\bm{x}_{\sss{i}}\tp \bm{x}_{\sss{k}}		
	,
	\label{eq:innerprod}	
	\end{align}
for all $i,k \in \mathcal{N}$, and where the fact that, for all $i\in\mathcal{N}$, $\bm{x}_{\sss{i}}$ is orthogonal to $\nmb^{\sss{\star}}$ has been used. 
For $\bm{x}_{\sss{i}} = 	\bm{x}_{\sss{k}} = \zvec$, $s_{\sss{i}} s_{\sss{k}} = \sign{(\nmbi[i]\tp \nmb^{\sss{\star}} ) (\nmbi[k]\tp \nmb^{\sss{\star}} )} = 1$ if $\nmbi[i]$ and $\nmbi[k]$ are synchronized and $s_{\sss{i}} s_{\sss{k}} = \sign{(\nmbi[i]\tp \nmb^{\sss{\star}} ) (\nmbi[k]\tp \nmb^{\sss{\star}} )} = -1$ if $\nmbi[i]$ and $\nmbi[k]$ are diametrically opposed.
Thus, by continuity, for \emph{small} $\norm{\bm{x}_{\sss{i}}}$ and $\norm{\bm{x}_{\sss{k}}}$, $s_{\sss{i}} s_{\sss{k}} = \sign{(\nmbi[i]\tp \nmb^{\sss{\star}} ) (\nmbi[k]\tp \nmb^{\sss{\star}} )} = 1$ if $\nmbi[i]$ and $\nmbi[k]$ are \emph{almost} synchronized (i.e., $1 - \nmbi[i]\tp \nmbi[k]$ is \emph{small}) and $s_{\sss{i}} s_{\sss{k}} = \sign{(\nmbi[i]\tp \nmb^{\sss{\star}} ) (\nmbi[k]\tp \nmb^{\sss{\star}} )} = -1$ if $\nmbi[i]$ and $\nmbi[k]$ are \emph{almost} diametrically opposed (i.e., $-1 - \nmbi[i]\tp \nmbi[k]$ is \emph{small}).

Since $\nmb^{\sss{\star}}$ is constant, it follows from~\eqref{eq:xi} that
\begin{align}
	\dot{\bm{x}}_{\sss{i}}
	=
	&
	\OP{\nmb^{\sss{\star}}} 
	\nmbiDot[i]
	\\
	=
	&
	\OP{\nmb^{\sss{\star}}} 
	(-\sk{\nmbi[i]}\Omi{i})
	\\
	=
	&
	\Scale[0.95]{
		\OP{\nmb^{\sss{\star}}} 
		\left(
			-\sk{\nmbi[i]}\sk{\nmbi[i]} 
			\sum_{k \in \mathcal{N}_{\sss{i}}} 
			f^{\sss{\prime}}_{\sss{\kappa(i,k)}} (d_{\sss{\theta}}(\nmbi[i],\nmbi[k]))
			\nmbi[k] 
		\right)
	}
	\\
	=
	&
	\OP{\nmb^{\sss{\star}}} 
	\left(
		\left(
			\Idmat - \nmbi[i] \nmbi[i]\tp
		\right) 
		\sum_{k \in \mathcal{N}_{\sss{i}}} 
		f^{\sss{\prime}}_{\sss{\kappa(i,k)}} (d_{\sss{\theta}}(\nmbi[i],\nmbi[k]))
		\nmbi[k] 
	\right)	
	\\
	=
	&
	\Scale[0.92]{
		\sum\limits_{\sss{k \in \mathcal{N}_{\sss{i}}}}
		f^{\sss{\prime}}_{\sss{\kappa(i,k)}} (d_{\sss{\theta}}(\nmbi[i],\nmbi[k]))
		\left(
			\OP{\nmb^{\sss{\star}}} \nmbi[k]
			-
			(\nmbi[i]\tp\nmbi[k])
			\OP{\nmb^{\sss{\star}}} \nmbi[i]
		\right)	
	}
	\\
	=
	&
	\sum_{k \in \mathcal{N}_{\sss{i}}}
	f^{\sss{\prime}}_{\sss{\kappa(i,k)}} (d_{\sss{\theta}}(\nmbi[i],\nmbi[k]))
	\left(
		\bm{x}_j
		-
		(\nmbi[i]\tp\nmbi[k])
		\bm{x}_{\sss{i}}
	\right)
	\\
	=
	&
	\sum_{k \in \mathcal{N}_{\sss{i}}}
	f^{\sss{\prime}}_{\sss{\kappa(i,k)}} (d_{\sss{\theta}}(\nmbi[i],\nmbi[k]))
	\left(
		\bm{x}_j
		-
		\bm{x}_{\sss{i}}
	\right)
	+
	\\
	&
	\sum_{k \in \mathcal{N}_{\sss{i}}}
	f^{\sss{\prime}}_{\sss{\kappa(i,k)}} (d_{\sss{\theta}}(\nmbi[i],\nmbi[k]))
	\left(
		1
		-
		(\nmbi[i]\tp\nmbi[k])
	\right)
	\bm{x}_{\sss{i}}
	,		
	\label{eq:xiDot}
\end{align}
where $\nmbi[i]$ and $\nmbi[k]$ can be expressed as in~\eqref{eq:nmbi3}. 
By stacking all $\bm{x}_{\sss{i}}$ as $\bm{x} \triangleq \threerow{\bm{x}_{\sss{1}}\tp}{\ldots}{\bm{x}_{\sss{N}}\tp}\tp$ and $s_{\sss{i}}$ as $\bm{s} \triangleq \threerow{s_{\sss{1}}\tp}{\ldots}{s_{\sss{N}}\tp}\tp$, \eqref{eq:sdot} and~\eqref{eq:xiDot} can be written compactly as
\begin{align}
	\dot{\bm{x}}
	= &
	-
	(B \otimes \Idmat)
	D_{\sss{1}}(\bm{x},\bm{s})
	(B \otimes \Idmat)\tp
	\bm{x}
	+
	D_{\sss{2}}(\bm{x},\bm{s})	
	\bm{x}
	\\
	=: &	 
	A(\bm{x},\bm{s}) \bm{x}
	\label{eq:xDot},
	\\
	\dot{\bm{s}}
	= &
	\zvec \quad \text{if $\norm{\bm{x}_{\sss{i}}} < 1$ for all $i \in \mathcal{N}$}
	\label{eq:sDot}
\end{align}
with
\begin{align}
	& D_{\sss{1}} (\bm{x},\bm{s})
	= 
	\left(
		\underset{\sss{k \in \mathcal{M}}}{\oplus}
		f^{\sss{\prime}}_{\sss{k}}(d_{\sss{\theta}}(\nmbt{k},\nmbh{k}))
	\right)
	\otimes 
	\Idmat
	\\
	& 
	D_{\sss{2}} (\bm{x},\bm{s})
	= 
	\left(
		\underset{\sss{i \in \mathcal{N}}}{\oplus}
		\sum\limits_{\sss{j \in \mathcal{N}_{\sss{i}}}}
		f^{\sss{\prime}}_{\sss{\kappa(i,j)}} (d_{\sss{\theta}}(\nmbi[i],\nmbi[j]))
		\left(
			1
			-
			(\nmbi[i]\tp\nmbi[k])
		\right)
	\right)	
	\otimes  \Idmat
	,
\end{align}
where both $D_{\sss{1}} (\bm{x},\bm{s})$ and $D_{\sss{2}} (\bm{x},\bm{s})$ are positive semi-definite.
As expected, the states $\bm{x}^{\star} = \zvec$ and $\bm{s}^{\star} = \bar{\onesvec} \in \{\zmb \in \Rn[N] : \zmb_{\sss{i}}^2 = 1, \forall i = \{1,\cdots,N\}\}$ are equilibria of the system with dynamics~\eqref{eq:xDot} and~\eqref{eq:sDot}.
The linearized system for~\eqref{eq:xDot} and~\eqref{eq:sDot}, and around each equilibrium, is given by 
\begin{align}
	\begin{bmatrix}
		\dot{\bm{x}} 
		\\
		\dot{\bm{s}}
	\end{bmatrix}
	= &
	\begin{bmatrix}
		A(\zvec,\bm{s}^{\star}) & \zvec
		\\
		\zvec & \zvec
	\end{bmatrix}
	\begin{bmatrix}
		\bm{x} 
		\\
		\bm{s}
	\end{bmatrix}
	\label{eq:LinearizedAMatrix}
	\\
	= &
	-
	\begin{bmatrix}
		(B \otimes \Idmat) D_{\sss{1}}(\zvec,\bm{s}^{\star}) (B \otimes \Idmat)\tp - D_{\sss{2}}(\zvec,\bm{s}^{\star}) & \zvec
		\\
		\zvec & \zvec
	\end{bmatrix}
	\begin{bmatrix}
		\bm{x} 
		\\
		\bm{s}
	\end{bmatrix}	
\end{align}
where the condition $\norm{\bm{x}_{\sss{i}}} < 1$, for all $i \in \mathcal{N}$, is dropped since the linearization is performed around $\bm{x}^{\star} = \zvec$. 
Notice that
\begin{align}
	&
	\onesvec_{\sss{3N}}\tp 
	\left(
		(B \otimes \Idmat) D_{\sss{1}}(\zvec,\bm{s}^{\star}) (B \otimes \Idmat)\tp 
		- 
		D_{\sss{2}}(\zvec,\bm{s}^{\star})
	\right)
	\onesvec_{\sss{3N}}
	=
	\\
	= & 
	- 
	\onesvec_{\sss{3N}}\tp  
	D_{\sss{2}}(\zvec,\bm{s}^{\star}) 
	\onesvec_{\sss{3N}}
	\\
	= &
 	-
 	3
	\sum\limits_{\sss{l \in \mathcal{N}}}
	\sum\limits_{\sss{k \in \mathcal{N}_{\sss{l}}}}
	(1 - s_{\sss{l}}^{\sss{\star}} s_{\sss{k}}^{\sss{\star}}) f^{\sss{\prime}}_{\sss{\kappa(l,k)}}(2)
\end{align}
is zero if and only if $\bm{s}^{\star} = \onesvec_{\sss{N}}$ and negative otherwise (this includes all scenarios where at least two neighboring unit vectors are diametrically opposed), as long as $f^{\sss{\prime}}_{\sss{k}}(2) > 0$ for all $k \in \mathcal{M}$.
As such, the matrix in~\eqref{eq:LinearizedAMatrix}, which is symmetric, is not positive semi-definite for $\bm{s}^{\star} \ne \onesvec_{\sss{N}}$ (which means it has at least one eigenvalue with negative real part -- Theorem~4.1.8 in~\cite{horn2012matrix}), which means any equilibrium where two neighboring unit vectors are diametrically opposed is unstable. 
On the other hand, the matrix in~\eqref{eq:LinearizedAMatrix} is positive semi-definite for $\bm{s}^{\star} = \onesvec_{\sss{N}}$ (this corresponds to the scenario where all unit vectors are synchronized), and the linearization cannot be used to infer stability of this equilibrium (nonetheless, this equilibrium is asymptotically stable).
		
		\section{Convergence to Constant Unit Vector}
		\label{app:ConvergenceToConstant}
		Theorem~\ref{thm:NoFullDomain} provides conditions for convergence to a synchronized network, but it does not provide any insight on whether $ \lim\limits_{t \rightarrow \infty} \nmbi[i](t)$ exists, i.e., whether the synchronized network converges to a constant unit vector or if it converges to a time-varying unit vector.
\begin{exa}
	Consider the following unit vector,
	\begin{align}
		\nmb(t) = 
		\begin{bmatrix}
			\cos(\ln (t))) &
			\sin(\ln (t))) &
			0
		\end{bmatrix}\tp
	\end{align}
	whose limit as time grows to infinity does not exist. Nonetheless, its time derivative is given by
	\begin{align}
		\dot{\nmb}(t)
		& =
		\frac{1}{t}
		\begin{bmatrix}
			-            \sin(\ln (t))) \\
			\hphantom{-} \cos(\ln (t))) \\
			0
		\end{bmatrix}
		=
		\begin{bmatrix}
			0             & -\frac{1}{t} & 0  \\
			\frac{1}{t}   & 0            & 0  \\
			0             & 0            & 0  \\
		\end{bmatrix}
		\begin{bmatrix}
			\cos(\ln (t))) \\
			\sin(\ln (t))) \\
			0
		\end{bmatrix}
		\\
		& =
		\sk{\frac{1}{t} \threerow{0}{0}{1}\tp} \nmb(t)
		\triangleq
		\sk{\bm{\omega}(t)} \nmb(t)
	\end{align}
	where $\bm{\omega}(\cdot)$ converges to $\zvec$. This means that convergence of $\bm{\omega}(\cdot)$ to $\zvec$ does not guarantee convergence of its unit vector to a constant unit vector. 
\end{exa}

\begin{prop}
	\label{prop:ComparisonTest}
	Consider two functions on $[a,b]$, $h_{\sss{1}}(x)$ and $h_{\sss{2}}(x)$, such that $0 \le h_{\sss{1}}(x)\le h_{\sss{2}}(x)$ for all $x \in [a,b]$. Then, if $\int_{a}^{b} h_{\sss{2}}(x) dx$ is convergent, so is $\int_{a}^{b} h_{\sss{1}}(x) dx$ (comparison test)~\cite{buck1965advanced}.
\end{prop}

\begin{prop}
	\label{prop:AbsoluteValueConvergence}
	Consider a function $h(x)$ on $[a,b]$. If $\int_{a}^{b} |h(x)| dx$ is convergent, then so is $\int_{a}^{b} h(x) dx$ \cite{buck1965advanced}.
\end{prop}

\begin{prop}
	\label{prop:OmegaConvergenceExpToZero}
	Consider a unit vector $\nmb$ with angular velocity $\bm{\omega}$. If $\bm{\omega}$ converges exponentially fast to zero, then $\nmb$ converges to a constant unit vector. 
\end{prop}
\begin{proof}
	Consider an arbitrary constant unit vector $\emb \in \mathcal{S}^{2}$ (for example, the standard basis vectors in $\Rn[3]$). Then $\emb\tp \nmb(t)$ is a scalar function of time, and it satisfies
	$
		\emb\tp \nmb(t) 
		=
		\int_{0}^{t} \emb\tp \sk{\bm{\omega}(\tau)} \nmb(\tau) d\tau
 	$.
 	Since $|\emb\tp \sk{\bm{\omega}(\tau)} \nmb(\tau)| \le \norm{\bm{\omega}(\tau)} \le C \exp(- \lambda \tau )$ (for some positive $C$ and $\lambda$), it follows from Propositions~\ref{prop:ComparisonTest} and~\ref{prop:AbsoluteValueConvergence} that $\lim\limits_{t \rightarrow \infty }\emb\tp \nmb(t)$ exists. Finally, since $\emb$ is an arbitrary unit vector, $\lim\limits_{t \rightarrow \infty } \nmb(t)$ also exists.
\end{proof}

The goal in this section is to study whether a network that converges to a synchronized configuration actually converges to a constant unit vector (in an unknown fixed orientation frame).
An analysis for the dynamic case (control at the torque level) is not performed, but rather an analysis for the kinematic case (control at the angular velocity level). 
Consider then that $\bm{\omega}$ is the control variable chosen as $\bm{\omega} = \Rmat\tp (B \otimes \Idmat) \emb$.

Consider that $g_{\sss{k}}$ is of class $\mathcal{P} \cup \mathcal{P}^{0}$ for all $k \in \mathcal{M}$. Also, assume that, for each edge $k$, there exists an interval $I = [0\,\,b]$, for some $0 < b < \pi$, such that if $d_{\theta} (\nmbt{k},\nmbh{k}) \in I$ then
\begin{align}
	\munderbar{\alpha}_{\sss{k}} \norm{\embi{k}}^{\sss{2}}
	\le
	d_{\sss{k}}(\nmbt{k},\nmbh{k})
	\le
	\bar{\alpha}_{\sss{k}}\norm{\embi{k}}^{\sss{2}},
	\label{eq:ConditionConvergence}
\end{align}
where $\embi{k} = g_{\sss{k}} (d_{\theta} (\nmbt{k},\nmbh{k}) ) \sk{\nmbt{k}} \nmbh{k}$.

Then, if the network converges to a synchronized configuration, i.e. $\nmbi{1} = \cdots  = \nmbi{N}$, it follows that $\lim\limits_{t \rightarrow \infty} \nmbi[i](t) = \nmb^{\star}$ for some constant $\nmb^{\star}$ and for all $i \in \mathcal{N}$. 

In order to come to this conclusion, consider the Lyapunov function $V(\nmb) = \sum_{k =1}^{M} d_{\sss{k}}(\nmbt{k},\nmbh{k})$ whose time derivative is rendered $\dot{V} = - \norm{(B \otimes \Idmat) \emb} \le 0$ by the chosen control law. If the network converges to the equilibrium configuration where $\nmbi{1} = \cdots  = \nmbi{N}$, it follows that $V(t)$ converges asymptotically to zero, and consequently, for all edges $k \in \mathcal{M}$, $d_{\theta} (\nmbt{k},\nmbh{k})$ will eventually enter the set $I$ (which has non-zero measure). Say $T$ is the time instant when, for all edges $k \in \mathcal{M}$, $d_{\theta} (\nmbt{k},\nmbh{k}) \in I$ for all $t\ge T$. From the previous discussion, a finite $T$ exists, and for $t\ge T$ one can say
\begin{align}
	\sum\limits_{\sss{k =1}}^{\sss{M}} \munderbar{\alpha}_{\sss{k}} \norm{\embi{k}(\nmbt{k},\nmbh{k})}^{\sss{2}}
	\le 
	& V(\nmb)
	\le
	\sum\limits_{\sss{k =1}}^{\sss{M}} \bar{\alpha}_{\sss{k}}\norm{\embi{k}(\nmbt{k},\nmbh{k})}^{\sss{2}}
	\\
	\min\limits_{\sss{k}} (\munderbar{\alpha}_{\sss{k}} ) \norm{\emb(\nmb)}^{\sss{2}}
	\le
	& V(\nmb)
	\le
	\max\limits_{\sss{k}}(\bar{\alpha}_{\sss{k}}) \norm{\emb(\nmb)}^{\sss{2}}.	
\end{align}
Since $\dot{V} = - \norm{(B \otimes \Idmat) \emb}^{\sss{2}} \le - \lambda_{\min}(B\tp B) \norm{\emb}^{\sss{2}} $, it follows that $\norm{\emb}$ converges exponentially fast to $\zvec$ (tree graph), and so does $\bm{\omega}_{\sss{i}}$ (since $ \bm{\omega} = \Rmat\tp (B \otimes \Idmat) \emb$) for all $i \in \mathcal{N}$. As such, it follows from Proposition~\ref{prop:OmegaConvergenceExpToZero} that all $\nmbi[i]$ converge to a constant unit vector.  

We can perform a very similar analysis for the dynamic case, but we need to construct a different Lyapunov function (than that provided in Section~\ref{subsubsec:LyapunovFunction}), specifically one whose time derivative depends on both $\norm{(B \otimes \Idmat) \emb}$  and $\norm{\Omi{i}}$. We do not present such Lyapunov function in this paper. Convergence to a constant vector in arbitrary graphs is also a topic for future research.

We present next an example of a distance function that satisfies conditions~\eqref{eq:ConditionConvergence}.

\begin{exa}
	Consider the distance function presented in Example~\ref{ex:DistanceFunctionCos} with $\alpha = 1$, i.e., $d(\nmb_{\sss{1}},\nmb_{\sss{2}}) = a ( 1 - \nmb_{\sss{1}}\tp\nmb_{\sss{2}})$. For this distance function $g = a$ and $\norm{\emb(\nmbi[1],\nmbi[2])} = a \norm{\sk{\nmb_{\sss{1}}} \nmb_{\sss{2}}} $. This distance function satisfies
	\begin{align}
		\frac{1}{a} \frac{\norm{\emb(\nmbi[1],\nmbi[2])}^{\sss{2}} }{2} \le d(\nmb_{\sss{1}},\nmb_{\sss{2}}) \le \frac{1}{a} \norm{\emb(\nmbi[1],\nmbi[2])}^{\sss{2}} 
	\end{align}
	for $d_{\theta} (\nmbi[1],\nmbi[2]) \in [0,\frac{\pi}{2}]$.
\end{exa}

\section{Extra Remarks}
\begin{rem}
	\label{rem:DistanceFunction}
	In Definition~\ref{eq:gClasses}, since $f(\cdot)$ is monotonous in $(0,2)$, $\lim_{\sss{s \rightarrow 0^{\sss{+}}}} f(s)$ and $\lim_{\sss{s \rightarrow 2^{\sss{-}}}} f(s)$ exist; while $\lim_{\sss{s \rightarrow 0^{+}}} f'(s)$ and $\lim_{\sss{s \rightarrow 2^{-}}} f'(s)$ do not necessarily exist.
	However, in what follows the important limits are $\lim_{\sss{s \rightarrow 0^{+}}} f'(s) \sqrt{s}$ and $\lim_{\sss{s \rightarrow 2^{-}}} f'(s) \sqrt{2 - s}$ which exist: 
	$\lim_{\sss{s \rightarrow 0^{+}}} f'(s) \sqrt{s} = 0$ since $\limsup_{\sss{s \rightarrow 0^{+}}} f'(s) < \infty$ and $\liminf_{\sss{s \rightarrow 0^{+}}} f'(s) \ge 0 $;
	$\lim_{\sss{s \rightarrow 2^{-}}} f'(s) \sqrt{2 - s} = 0$ when $\limsup_{\sss{s \rightarrow 2^{-}}} f'(s) < \infty$;
	and $\lim_{\sss{s \rightarrow 2^{-}}} f'(s) \sqrt{2 - s} $ may or may not be $0$ when $\limsup_{\sss{s \rightarrow 2^{-}}} f'(s) = \infty$.
	Also, if $f \in \mathcal{P}^{\sss{\infty}}$ then $\lim_{\sss{s \rightarrow 2^{-}}} f'(s) \sqrt{2 - s} \ne 0$: indeed, suppose $\lim_{\sss{s \rightarrow 2^{-}}} f'(s) \sqrt{2 - s} = 0$, which implies that $\sup_{\sss{s \in (0,2)}} f'(s) \sqrt{2 - s} =: b < \infty$; 
	since
	\begin{align}
		\lim_{\sss{s \rightarrow 2^{\sss{-}}}} f(s) 
		=
		& 
		\lim_{\sss{s \rightarrow 0^{\sss{+}}}} f(s) 
		+ 
		\lim_{\sss{s \rightarrow 2^{\sss{-}}}} \int_{0}^{s} f'(x) dx 
		\\
		=
		& 
		\lim_{\sss{s \rightarrow 2^{\sss{-}}}} \int_{0}^{s} \frac{1}{\sqrt{2 - x}} (f'(x) \sqrt{2 - x}) dx 
		\\
		\le
		& 
		b \lim_{\sss{s \rightarrow 2^{\sss{-}}}} \int_{0}^{s} \frac{1}{\sqrt{2 - x}} dx
		=
		2\sqrt{2} b < \infty,
	\end{align} 
	which implies that $f \not\in \mathcal{P}^{\sss{\infty}}$.
	As such, $f \in \mathcal{P}^{\sss{\infty}} \Rightarrow f \in \mathcal{P}^{\sss{\bar{0}}}$.
	Figure~\ref{fig:ClassesOf_f_Functions} illustrates how the properties that $f(\cdot)$ satisfies affects the classes it belongs to.
\end{rem}
\begin{defn}
	\label{def:SigmaFunction}
	Given $\bm{\sigma}: \Rn[3] \mapsto \Rn[3]$, we say $\bm{\sigma} \in \Sigma$ if $\exists \sigma \in \mathcal{C}^{\sss{1}}(\Rn[{}]_{\ge 0},\Rn[{}]_{\ge 0}) : \bm{\sigma}(\xmb) = \sigma(\|\xmb\|) \xmb$, and $\sigma'^{\sss{\max}} = \sup_{\sss{\xmb \in \Rn[3]}} \|D \bm{\sigma}(\xmb)\| = \sup_{\sss{\xmb \in \Rn[3]}} \| \frac{\partial \bm{\sigma}(\xmb) }{\partial \xmb} \| < \infty$ (i.e., $\bm{\sigma}(\cdot)$ is Lipschitz, see Lemma~3.1 in~\cite{khalil2002nonlinear})
	Moreover, we denote $\sigma^{\sss{\max}} = \sup_{\sss{\xmb \in \Rn[3]}} \norm{\bm{\sigma}(\xmb)}$.
\end{defn}

\begin{rem}\label{rem:ExponentialConvergence}
	\normalfont
	Theorems~\ref{thm:NoFullDomain} and~\ref{thm:LocalStability} provide asymptotic results, namely $\lim_{\sss{t \rightarrow \infty }} \emb(\nmb(t)) = \zvec$ (or alternatively that $\lim_{\sss{t \rightarrow \infty }} \text{dist}(\nmb(t),\Omega^{\sss{\star}}) = 0$, where $\Omega^{\sss{\star}} =\{ \nmb \in (\mathcal{S}^{\sss{2}})^{\sss{N}}:  \nmb = \onesvec_{\sss{N}} \otimes \nmb^{\sss{\star}}, \forall  \nmb^{\sss{\star}} \in \mathcal{S}^{\sss{2}} \} $).
	An open question is when can it be guaranteed that those limits converge to $0$ exponentially fast.
	Similarly to~\cite{olfati2006swarms}, consider $\dot{\theta}_{\sss{1}} = f'(1  - \cos(\theta_{\sss{1}} - \theta_{\sss{2}})) \sin(\theta_{\sss{2}} - \theta_{\sss{1}})$ and $\dot{\theta}_{\sss{2}} = f'(1  - \cos(\theta_{\sss{1}} - \theta_{\sss{2}})) \sin(\theta_{\sss{2}} - \theta_{\sss{1}})$, where we omitted the time dependencies;
	it follows that, if we denote $\theta = \theta_{\sss{1}} - \theta_{\sss{2}}$, then $\dot{\theta} = - 2f'(1  - \cos(\theta)) \sin(\theta)$.
	In~\cite{olfati2006swarms}, only $f(s) = a_{\sss{12}} s \Rightarrow f'(s) = a_{\sss{12}} $ is considered, but in this manuscript we consider a wider set of functions. 
	Consider then the function $f(s) = \frac{1}{8}(\sqrt{s (2 - s)} (s -1) + \arccos(1- s))$, and notice that $f \in \mathcal{P}_{\sss{0}}$ (see Fig.~\ref{fig:gClasses}).  
	It follows that $\dot{\theta} = - |\sin(\theta)|\sin(\theta)$; therefore, if $\theta(0) \in [0,\frac{\pi}{2}]$, it follows that $\theta(t) = \text{arccot}(t + \cot(\theta(0) ))$, which does not converge exponentially fast to $0$ (since $\int_{\sss{0}}^{\sss{\infty}} \theta(t) dt = \infty $).
	Consider now the function $f(s) = s$, and notice that $f \in \mathcal{P}_{\sss{\bar{0}}}$.
	It follows that $\dot{\theta} = - 2\sin(\theta) < -\theta$, where the inequality follows for $\theta \in [0,\frac{\pi}{2}]$; therefore, if $\theta(0) \in [0,\frac{\pi}{2}]$, it follows that $\theta(t) \le \exp(-t)$, which does converge exponentially fast to $0$.
	This suggests that, in this manuscript's framework, exponential convergence cannot be guaranteed for all $f\in\bar{\mathcal{P}}$, while asymptotic stability can still be guaranteed.
\end{rem}
		
\end{document}